\documentclass[conference,onecolumn]{IEEEtran}

\usepackage{cite}

\usepackage[cmex10]{amsmath}
\usepackage{array}

\usepackage[pdftex]{graphicx}

\usepackage{amsfonts,amssymb,amsthm,bbm,verbatim}
\usepackage[margin=0.9in]{geometry}
\usepackage[utf8]{inputenc}
\usepackage{hyperref}
\hypersetup{colorlinks = true, linkcolor = blue, citecolor = blue}
\usepackage[all]{hypcap} 
\usepackage{physics} 
\usepackage[usenames,dvipsnames]{color}
\usepackage[T1]{fontenc}

\def\G{{\mathcal{G}}}
\def\F{{\mathcal{F}}}
\def\H{{\mathcal{H}}}
\def\P{{\mathcal{P}}}
\def\X{{\mathcal{X}}}
\def\Y{{\mathcal{Y}}}

\def\U{{\mathcal{U}}}
\def\D{{\mathcal{D}}}
\def\L{{\mathcal{L}}}
\def\Range{{\mathcal{R}}}
\def\Dirichlet{{\mathcal{E}}}
\def\E{{\mathbb{E}}}
\def\VAR{{\mathbb{V}\mathbb{A}\mathbb{R}}}
\def\COV{{\mathbb{C}\mathbb{O}\mathbb{V}}}

\def\R{{\mathbb{R}}}
\def\N{{\mathbb{N}}}
\def\1{{\textbf{1}}}
\def\unif{{\textbf{u}}}
\def\w{{w_{\delta}}}

\newtheorem{theorem}{Theorem}
\newtheorem{lemma}{Lemma}
\newtheorem{proposition}{Proposition}
\newtheorem{corollary}{Corollary}

\newtheorem{definition}{Definition}

\IEEEoverridecommandlockouts

\begin{document}

\title{Comparison of Channels: Criteria for\\Domination by a Symmetric Channel}

\author{Anuran Makur and Yury Polyanskiy%
\thanks{This work was supported in part by the National Science Foundation CAREER award under grant agreement CCF-12-53205, and in part by the Center for Science of Information (CSoI), an NSF Science and Technology Center, under grant agreement CCF-09-39370. This work was presented at the 2017 IEEE International Symposium on Information Theory (ISIT) \cite{SymmetricChannelDominationConf}.}%
\thanks{The authors are with the Department of Electrical Engineering and Computer Science, Massachusetts Institute of Technology, Cambridge, MA 02139, USA (e-mail: a\_makur@mit.edu; yp@mit.edu).}%
\thanks{Copyright (c) 2018 IEEE. Personal use of this material is permitted. However, permission to use this material for any other purposes must be obtained from the IEEE by sending a request to pubs-permissions@ieee.org.}}%

\maketitle

\thispagestyle{plain}
\pagestyle{plain}

\begin{abstract}
This paper studies the basic question of whether a given channel $V$ can be dominated (in the precise sense of being more noisy) by a $q$-ary symmetric channel. The concept of ``less noisy'' relation between channels originated in network information theory (broadcast channels) and is defined in terms of mutual information or Kullback-Leibler divergence. We provide an equivalent characterization in terms of $\chi^2$-divergence. Furthermore, we develop a simple criterion for domination by a $q$-ary symmetric channel in terms of the minimum entry of the stochastic matrix defining the channel $V$. The criterion is strengthened for the special case of additive noise channels over finite Abelian groups. Finally, it is shown that domination by a symmetric channel implies (via comparison of Dirichlet forms) a logarithmic Sobolev inequality for the original channel.
\end{abstract}


\tableofcontents
\thispagestyle{plain} 
\hypersetup{linkcolor = red}

\section{Introduction}
\label{Introduction}

For any Markov chain $U \rightarrow X \rightarrow Y$, it is well-known that the data processing inequality, $I(U;Y) \leq I(U;X)$, holds. This result can be
strengthened to \cite{AhlswedeGacsHypercontraction}:
\begin{equation}
\label{Eq: SDPI}
I(U;Y) \leq \eta I(U;X)
\end{equation}
where the contraction coefficient $\eta \in [0,1]$ only depends on the channel $P_{Y|X}$. Frequently, one gets $\eta < 1$ and the resulting inequality is called a \textit{strong data processing inequality} (SDPI). Such inequalities have been recently simultaneously rediscovered and applied in several disciplines; see \cite[Section 2]{GraphSDPI} for a short survey. In \cite[Section 6]{GraphSDPI}, it was noticed that the validity of \eqref{Eq: SDPI} for all $P_{U,X}$ is equivalent to the statement that an erasure channel with erasure probability $1-\eta$ is \textit{less noisy} than the given channel $P_{Y|X}$. In this way, the entire field of SDPIs is equivalent to determining whether a given channel is dominated by an erasure channel. 

This paper initiates the study of a natural extension of the concept of SDPI by replacing the distinguished role played by erasure channels with $q$-ary symmetric channels. We give simple criteria for testing this type of domination and explain how the latter can be used to prove logarithmic Sobolev inequalities. In the next three subsections, we introduce some basic definitions and notation. We state and motivate our main question in subsection \ref{Main Question and Motivation}, and present our main results in section \ref{Main Results}. 

\subsection{Preliminaries}
\label{Preliminaries}

The following notation will be used in our ensuing discussion. Consider any $q,r \in \N \triangleq \left\{1,2,3,\dots\right\}$. We let $\R^{q \times r}$ (respectively $\mathbb{C}^{q \times r}$) denote the set of all real (respectively complex) $q \times r$ matrices. Furthermore, for any matrix $A \in \R^{q \times r}$, we let $A^T \in \R^{r \times q}$ denote the transpose of $A$, $A^{\dagger} \in \R^{r \times q}$ denote the \textit{Moore-Penrose pseudoinverse} of $A$, $\Range(A)$ denote the range (or column space) of $A$, and $\rho\left(A\right)$ denote the \textit{spectral radius} of $A$ (which is the maximum of the absolute values of all complex eigenvalues of $A$) when $q = r$. We let $\R^{q \times q}_{\succeq 0} \subsetneq \R^{q \times q}_{\textsf{sym}}$ denote the sets of positive semidefinite and symmetric matrices, respectively. In fact, $\R^{q \times q}_{\succeq 0}$ is a closed convex cone (with respect to the Frobenius norm). We also let $\succeq_{\textsf{\tiny PSD}}$ denote the \textit{L\"{o}wner partial order} over $\R^{q \times q}_{\textsf{sym}}$: for any two matrices $A,B \in \R^{q \times q}_{\textsf{sym}}$, we write $A \succeq_{\textsf{\tiny PSD}} B$ (or equivalently, $A - B \succeq_{\textsf{\tiny PSD}} 0$, where $0$ is the zero matrix) if and only if $A - B \in \R^{q \times q}_{\succeq 0}$. To work with probabilities, we let $\P_q \triangleq \{p = \left(p_1,\dots,p_q\right) \in \R^q : p_1,\dots,p_q \geq 0 \enspace \text{and} \enspace p_1 + \cdots + p_q = 1\}$ be the probability simplex of row vectors in $\R^q$, $\P_q^{\circ} \triangleq \{p = \left(p_1,\dots,p_q\right) \in \R^q : p_1,\dots,p_q > 0 \enspace \text{and} \enspace p_1 + \cdots + p_q = 1\}$ be the relative interior of $\P_q$, and $\R^{q \times r}_{\textsf{sto}}$ be the convex set of row stochastic matrices (which have rows in $\P_r$). Finally, for any (row or column) vector $x = (x_1,\dots,x_q) \in \R^q$, we let $\textsf{\small diag}(x) \in \R^{q \times q}$ denote the diagonal matrix with entries $\left[\textsf{\small diag}(x)\right]_{i,i} = x_i$ for each $i \in \{1,\dots,q\}$, and for any set of vectors $\mathcal{S} \subseteq \R^{q}$, we let $\textsf{\small conv}\left(\mathcal{S}\right)$ be the convex hull of the vectors in $\mathcal{S}$. 

\subsection{Channel preorders in information theory} 
\label{Channel Preorders in Information Theory}

Since we will study preorders over discrete channels that capture various notions of relative ``noisiness'' between channels, we provide an overview of some well-known channel preorders in the literature. Consider an input random variable $X \in \X$ and an output random variable $Y \in \Y$, where the alphabets are $\X = [q] \triangleq \{0,1,\dots,q-1\}$ and $\Y = [r]$ for $q,r \in \N$ without loss of generality. We let $\P_q$ be the set of all probability mass functions (pmfs) of $X$, where every pmf $P_X = \left(P_X(0),\dots,P_X(q-1)\right) \in \P_q$ and is perceived as a row vector. Likewise, we let $\P_r$ be the set of all pmfs of $Y$. A \textit{channel} is the set of conditional distributions $W_{Y|X}$ that associates each $x \in \X$ with a conditional pmf $W_{Y|X}(\cdot|x) \in \P_r$. So, we represent each channel with a stochastic matrix $W \in \R^{q \times r}_{\textsf{sto}}$ that is defined entry-wise as:
\begin{equation}
\forall x \in \X, \forall y \in \Y, \enspace \left[W\right]_{x+1,y+1} \triangleq W_{Y|X}(y|x)
\end{equation}
where the $(x+1)$th row of $W$ corresponds to the conditional pmf $W_{Y|X}(\cdot|x) \in \P_r$, and each column of $W$ has at least one non-zero entry so that no output alphabet letters are redundant. Moreover, we think of such a channel as a (linear) map $W:\P_q \rightarrow \P_r$ that takes any row probability vector $P_X \in \P_q$ to the row probability vector $P_Y = P_X W \in \P_r$. 

One of the earliest preorders over channels was the notion of \textit{channel inclusion} proposed by Shannon in \cite{ShannonPartialOrder}.\footnote{Throughout this paper, we will refer to various information theoretic orders over channels as \textit{preorders} rather than \textit{partial orders} (although the latter is more standard terminology in the literature). This is because we will think of channels as individual stochastic matrices rather than equivalence classes of stochastic matrices (e.g. identifying all stochastic matrices with permuted columns), and as a result, the anti-symmetric property will not hold.} Given two channels $W \in \R^{q \times r}_{\textsf{sto}}$ and $V \in \R^{s \times t}_{\textsf{sto}}$ for some $q,r,s,t \in \N$, he stated that $W$ includes $V$, denoted $W \succeq_{\textsf{\tiny inc}} V$, if there exist a pmf $g \in \P_m$ for some $m \in \N$, and two sets of channels $\{A_k \in \R^{r \times t}_{\textsf{sto}} : k = 1,\dots,m\}$ and $\{B_k \in \R^{s \times q}_{\textsf{sto}} : k = 1,\dots,m\}$, such that:
\begin{equation}
V = \sum_{k = 1}^{m}{g_k B_k W A_k} .
\end{equation}
Channel inclusion is preserved under channel addition and multiplication (which are defined in \cite{ZeroErrorCapacity}), and the existence of a code for $V$ implies the existence of as good a code for $W$ in a probability of error sense \cite{ShannonPartialOrder}. The channel inclusion preorder includes the \textit{input-output degradation} preorder, which can be found in \cite{StochasticMatricesAndContractionCoefficients}, as a special case. Indeed, $V$ is an input-output degraded version of $W$, denoted $W \succeq_{\textsf{\tiny iod}} V$, if there exist channels $A \in \R^{r \times t}_{\textsf{sto}}$ and $B \in \R^{s \times q}_{\textsf{sto}}$ such that $V = B W A$. We will study an even more specialized case of Shannon's channel inclusion known as \textit{degradation} \cite{DegradationCover, Degradation}.  

\begin{definition}[Degradation Preorder]
\label{Def: Degradation Preorder}
A channel $V \in \R^{q \times s}_{\textsf{sto}}$ is said to be a \textit{degraded} version of a channel $W \in \R^{q \times r}_{\textsf{sto}}$ with the same input alphabet, denoted $W \succeq_{\textsf{\tiny deg}} V$, if $V = W A$ for some channel $A \in \R^{r \times s}_{\textsf{sto}}$. 
\end{definition}

We note that when Definition \ref{Def: Degradation Preorder} of degradation is applied to general matrices (rather than stochastic matrices), it is equivalent to Definition C.8 of \textit{matrix majorization} in \cite[Chapter 15]{Majorization}. Many other generalizations of the majorization preorder over vectors (briefly introduced in Appendix \ref{App: Basics of Majorization Theory}) that apply to matrices are also presented in \cite[Chapter 15]{Majorization}.

K\"{o}rner and Marton defined two other preorders over channels in \cite{ChannelPartialOrders} known as the \textit{more capable} and \textit{less noisy} preorders. While the original definitions of these preorders explicitly reflect their significance in channel coding, we will define them using equivalent mutual information characterizations proved in \cite{ChannelPartialOrders}. (See \cite[Problems 6.16-6.18]{CsiszarKorner} for more on the relationship between channel coding and some of the aforementioned preorders.) We say a channel $W \in \R^{q \times r}_{\textsf{sto}}$ is \textit{more capable} than a channel $V \in \R^{q \times s}_{\textsf{sto}}$ with the same input alphabet, denoted $W \succeq_{\textsf{\tiny mc}} V$, if $I(P_X,W_{Y|X}) \geq I(P_X,V_{Y|X})$ for every input pmf $P_X \in \P_q$, where $I(P_X,W_{Y|X})$ denotes the mutual information of the joint pmf defined by $P_X$ and $W_{Y|X}$. The next definition presents the less noisy preorder, which will be a key player in our study.

\begin{definition}[Less Noisy Preorder] 
\label{Def: Less Noisy Preorder}
Given two channels $W \in \R^{q \times r}_{\textsf{sto}}$ and $V \in \R^{q \times s}_{\textsf{sto}}$ with the same input alphabet, let $Y_W$ and $Y_V$ denote the output random variables of $W$ and $V$, respectively. Then, $W$ is \textit{less noisy} than $V$, denoted $W \succeq_{\textsf{\tiny ln}} V$, if $I(U;Y_W) \geq I(U;Y_V)$ for every joint distribution $P_{U,X}$, where the random variable $U \in \U$ has some arbitrary range $\U$, and $U \rightarrow X \rightarrow (Y_W,Y_V)$ forms a Markov chain.  
\end{definition}

An analogous characterization of the less noisy preorder using Kullback-Leibler (KL) divergence or relative entropy is given in the next proposition.

\begin{proposition}[KL Divergence Characterization of Less Noisy \cite{ChannelPartialOrders}]
\label{Prop: KL Divergence Characterization of Less Noisy}
Given two channels $W \in \R^{q \times r}_{\textsf{sto}}$ and $V \in \R^{q \times s}_{\textsf{sto}}$ with the same input alphabet, $W \succeq_{\textsf{\tiny ln}} V$ if and only if $D(P_X W||Q_X W) \geq D(P_X V||Q_X V)$ for every pair of input pmfs $P_X,Q_X \in \P_q$, where $D(\cdot||\cdot)$ denotes the KL divergence.\footnote{Throughout this paper, we will adhere to the convention that $\infty \geq \infty$ is true. So, $D(P_X W||Q_X W) \geq D(P_X V||Q_X V)$ is not violated when both KL divergences are infinity.}
\end{proposition}

We will primarily use this KL divergence characterization of $\succeq_{\textsf{\tiny ln}}$ in our discourse because of its simplicity. Another well-known equivalent characterization of $\succeq_{\textsf{\tiny ln}}$ due to van Dijk is presented below, cf. \cite[Theorem 2]{vanDijk}. We will derive some useful corollaries from it later in subsection \ref{Characterizations via the Loewner Partial Order and Spectral Radius}.

\begin{proposition}[van Dijk Characterization of Less Noisy \cite{vanDijk}]
\label{Prop: van Dijk Characterization of Less Noisy}
Given two channels $W \in \R^{q \times r}_{\textsf{sto}}$ and $V \in \R^{q \times s}_{\textsf{sto}}$ with the same input alphabet, consider the functional $F : \P_q \rightarrow \R$:
$$ \forall P_X \in \P_q, \enspace F(P_X) \triangleq I(P_X,W_{Y|X}) - I(P_X,V_{Y|X}) . $$
Then, $W \succeq_{\textsf{\tiny ln}} V$ if and only if $F$ is concave.
\end{proposition}

The more capable and less noisy preorders have both been used to study the capacity regions of broadcast channels. We refer readers to \cite{BCCapacityElGamal,BCCapacityNair,BCCapacityNairShamai}, and the references therein for further details. We also remark that the more capable and less noisy preorders tensorize, as shown in \cite[Problem 6.18]{CsiszarKorner} and \cite[Proposition 16]{GraphSDPI}, \cite[Proposition 5]{UniversalPolarCodes}, respectively. 

\begin{figure}[!t] 
\centering
\includegraphics[trim = 0mm 145mm 0mm 100mm, width=0.5\linewidth]{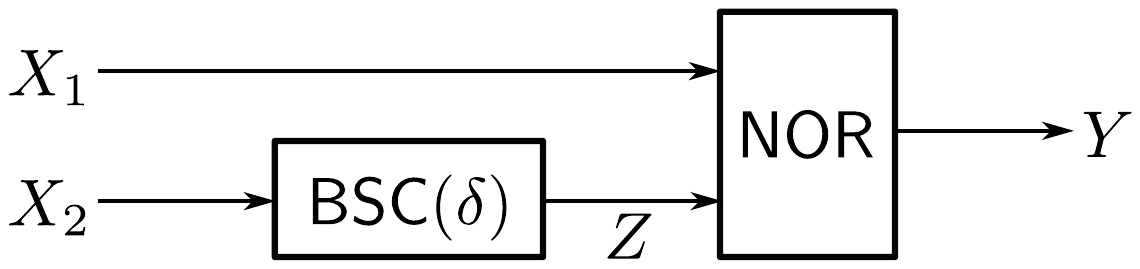} 
\caption{Illustration of a Bayesian network where $X_1,X_2,Z,Y \in \{0,1\}$ are binary random variables, $P_{Z|X_2}$ is a $\textsf{\scriptsize BSC}(\delta)$ with $\delta \in (0,1)$, and $P_{Y|X_1,Z}$ is defined by a deterministic NOR gate.}
\label{Figure: NOR Gate Counter-Example}
\end{figure}

On the other hand, these preorders exhibit rather counter-intuitive behavior in the context of Bayesian networks (or directed graphical models). Consider a Bayesian network with ``source'' nodes (with no inbound edges) $X$ and ``sink'' nodes (with no outbound edges) $Y$. If we select a node $Z$ in this network and replace the channel from the parents of $Z$ to $Z$ with a less noisy channel, then we may reasonably conjecture that the channel from $X$ to $Y$ also becomes less noisy (motivated by the results in \cite{GraphSDPI}). However, this conjecture is false. To see this, consider the Bayesian network in Figure \ref{Figure: NOR Gate Counter-Example} (inspired by the results in \cite{Unger2010}), where the source nodes are $X_1 \sim \textsf{\small Bernoulli}\!\left(\frac{1}{2}\right)$ and $X_2 = 1$ (almost surely), the node $Z$ is the output of a binary symmetric channel (BSC) with crossover probability $\delta \in (0,1)$, denoted $\textsf{\small BSC}(\delta)$, and the sink node $Y$ is the output of a NOR gate. Let $I(\delta) = I(X_1,X_2;Y)$ be the end-to-end mutual information. Then, although $\textsf{\small BSC}(0) \succeq_{\textsf{\tiny ln}} \textsf{\small BSC}(\delta)$ for $\delta \in (0,1)$, it is easy to verify that $I(\delta) > I(0) = 0$. So, when we replace the $\textsf{\small BSC}(\delta)$ with a less noisy $\textsf{\small BSC}(0)$, the end-to-end channel does \textit{not} become less noisy (or more capable).

The next proposition illustrates certain well-known relationships between the various preorders discussed in this subsection.

\begin{proposition}[Relations between Channel Preorders]
\label{Prop: Relations between Channel Preorders}
Given two channels $W \in \R^{q \times r}_{\textsf{sto}}$ and $V \in \R^{q \times s}_{\textsf{sto}}$ with the same input alphabet, we have:
\begin{enumerate}
\item $W \succeq_{\textsf{\tiny deg}} V \enspace \Rightarrow \enspace W \succeq_{\textsf{\tiny iod}} V \enspace \Rightarrow \enspace W \succeq_{\textsf{\tiny inc}} V$,
\item $W \succeq_{\textsf{\tiny deg}} V \enspace \Rightarrow \enspace W \succeq_{\textsf{\tiny ln}} V \enspace \Rightarrow \enspace W \succeq_{\textsf{\tiny mc}} V$.
\end{enumerate}
\end{proposition}

These observations follow in a straightforward manner from the definitions of the various preorders. Perhaps the only nontrivial implication is $W \succeq_{\textsf{\tiny deg}} V \Rightarrow W \succeq_{\textsf{\tiny ln}} V$, which can be proven using Proposition \ref{Prop: KL Divergence Characterization of Less Noisy} and the data processing inequality.

\subsection{Symmetric channels and their properties}
\label{Symmetric Channels and their Properties}

We next formally define $q$-ary symmetric channels and convey some of their properties. To this end, we first introduce some properties of Abelian groups and define additive noise channels. Let us fix some $q \in \N$ with $q \geq 2$ and consider an \textit{Abelian group} $(\X,\oplus)$ of order $q$ equipped with a binary ``addition'' operation denoted by $\oplus$. Without loss of generality, we let $\X = [q]$, and let $0$ denote the identity element. This endows an ordering to the elements of $\X$. Each element $x \in \X$ permutes the entries of the row vector $(0,\dots,q-1)$ to $(\sigma_x(0),\dots,\sigma_x(q-1))$ by (left) addition in the Cayley table of the group, where $\sigma_x:[q] \rightarrow [q]$ denotes a permutation of $[q]$, and $\sigma_x(y) = x \oplus y$ for every $y \in \X$. So, corresponding to each $x \in \X$, we can define a permutation matrix $P_x \triangleq \left[e_{\sigma_x(0)} \cdots e_{\sigma_x(q-1)}\right] \in \R^{q \times q}$ such that:
\begin{equation}
\label{Eq: Permutation Representation}
\left[v_0 \, \cdots \, v_{q-1}\right] P_x = \left[v_{\sigma_x(0)} \, \cdots \, v_{\sigma_x(q-1)}\right]
\end{equation}
for any $v_0,\dots,v_{q-1} \in \R$, where for each $i \in [q]$, $e_i \in \R^q$ is the $i$th standard basis column vector with unity in the $(i+1)$th position and zero elsewhere. The permutation matrices $\left\{P_x \in \R^{q \times q}: x \in \X\right\}$ (with the matrix multiplication operation) form a group that is isomorphic to $(\X,\oplus)$ (see \textit{Cayley's theorem}, and permutation and regular representations of groups in \cite[Sections 6.11, 7.1, 10.6]{AlgebraArtin}). In particular, these matrices commute as $(\X,\oplus)$ is Abelian, and are jointly unitarily diagonalizable by a \textit{Fourier matrix of characters} (using \cite[Theorem 2.5.5]{BasicMatrixAnalysis}). We now recall that given a row vector $x = \left(x_0,\dots,x_{q-1}\right) \in \R^q$, we may define a corresponding \textit{$\X$-circulant matrix}, $\textsf{\small circ}_{\X}(x) \in \R^{q \times q}$, that is given entry-wise by \cite[Chapter 3E, Section 4]{GroupRepsProbability}: 
\begin{equation}
\label{Eq: Group Circulant Matrix Definition}
\forall a,b \in [q], \enspace \left[\textsf{\small circ}_{\X}(x)\right]_{a+1,b+1} \triangleq x_{-a \oplus b} . 
\end{equation} 
where $-a \in \X$ denotes the inverse of $a \in \X$. Moreover, we can decompose this $\X$-circulant matrix as:
\begin{equation}
\label{Eq: Group Circulant Matrix Decomposition} 
\textsf{\small circ}_{\X}(x) = \sum_{i = 0}^{q-1}{x_i P_i^T}
\end{equation}
since $\sum_{i = 0}^{q-1}{x_i \! \left[P_i^T\right]_{a+1,b+1}} \! = \sum_{i = 0}^{q-1}{x_i \! \left[e_{\sigma_{i}(a)}\right]_{b+1}} = x_{-a \oplus b}$ for every $a,b \in [q]$. Using similar reasoning, we can write:
\begin{equation}
\label{Eq: Group Circulant Matrix Row-Column Decomposition} 
\textsf{\small circ}_{\X}(x) = \left[P_0 y \, \cdots \, P_{q-1} y\right] = \left[P_0 x^T \, \cdots \, P_{q-1} x^T\right]^T
\end{equation}
where $y = \left[x_{0} \enspace x_{-1} \cdots \, x_{-(q-1)}\right]^T \in \R^q$, and $P_0 = I_q \in \R^{q \times q}$ is the $q \times q$ identity matrix. Using \eqref{Eq: Group Circulant Matrix Decomposition}, we see that $\X$-circulant matrices are normal, form a commutative algebra, and are jointly unitarily diagonalizable by a Fourier matrix. Furthermore, given two row vectors $x,y \in \R^q$, we can define $x \, \textsf{\small circ}_{\X}(y) = y \, \textsf{\small circ}_{\X}(x)$ as the \textit{$\X$-circular convolution} of $x$ and $y$, where the commutativity of $\X$-circular convolution follows from the commutativity of $\X$-circulant matrices. 

A salient specialization of this discussion is the case where $\oplus$ is addition modulo $q$, and $(\X = [q],\oplus)$ is the cyclic Abelian group $\mathbb{Z}/q\mathbb{Z}$. In this scenario, $\X$-circulant matrices correspond to the standard circulant matrices which are jointly unitarily diagonalized by the \textit{discrete Fourier transform} (DFT) matrix. Furthermore, for each $x \in [q]$, the permutation matrix $P_x^T = P_q^x$, where $P_q \in \R^{q \times q}$ is the \textit{generator cyclic permutation matrix} as presented in \cite[Section 0.9.6]{BasicMatrixAnalysis}:
\begin{equation}
\label{Eq: Generator Cyclic Permutation Matrix}
\forall a,b \in [q], \enspace \left[P_q\right]_{a+1,b+1} \triangleq \Delta_{1,(b-a \, (\textsf{\scriptsize mod} \, q))}
\end{equation} 
where $\Delta_{i,j}$ is the Kronecker delta function, which is unity if $i = j$ and zero otherwise. The matrix $P_q$ cyclically shifts any input row vector to the right once, i.e. $\left(x_1,x_2,\dots,x_{q}\right)P_q = \left(x_{q},x_1,\dots,x_{q-1}\right)$.

Let us now consider a channel with common input and output alphabet $\X = \Y = [q]$, where $(\X,\oplus)$ is an Abelian group. Such a channel operating on an Abelian group is called an \textit{additive noise channel} when it is defined as:
\begin{equation}
\label{Eq: Additive Noise Channel}
Y = X \oplus Z
\end{equation}
where $X \in \X$ is the input random variable, $Y \in \X$ is the output random variable, and $Z \in \X$ is the additive noise random variable that is independent of $X$ with pmf $P_Z = \left(P_Z(0),\dots,P_Z(q-1)\right) \in \P_q$. The channel transition probability matrix corresponding to \eqref{Eq: Additive Noise Channel} is the $\X$-circulant stochastic matrix $\textsf{\small circ}_{\X}\!\left(P_Z\right) \in \R^{q \times q}_{\textsf{sto}}$, which is also doubly stochastic (i.e. both $\textsf{\small circ}_{\X}\!\left(P_Z\right), \textsf{\small circ}_{\X}\!\left(P_Z\right)^T \in \R^{q \times q}_{\textsf{sto}}$). Indeed, for an additive noise channel, it is well-known that the pmf of $Y$, $P_Y \in \P_q$, can be obtained from the pmf of $X$, $P_X \in \P_q$, by $\X$-circular convolution: $P_Y = P_X \, \textsf{\small circ}_{\X}\!\left(P_Z\right)$. We remark that in the context of various channel symmetries in the literature (see \cite[Section VI.B]{ChannelSymmetries} for a discussion), additive noise channels correspond to ``group-noise'' channels, and are input symmetric, output symmetric, Dobrushin symmetric, and Gallager symmetric.

The \textit{$q$-ary symmetric channel} is an additive noise channel on the Abelian group $(\X,\oplus)$ with noise pmf $P_Z = \w \triangleq \left(1-\delta,\delta/(q-1),\dots,\delta/(q-1)\right) \in \P_q$, where $\delta \in [0,1]$. Its channel transition probability matrix is denoted $W_{\delta} \in \R^{q \times q}_{\textsf{sto}}$:
\begin{equation}
\label{Eq: Symmetric Channel}
W_{\delta} \triangleq \textsf{\small circ}_{\X}\!\left(\w\right) = \left[\w^T \enspace P_q^T\w^T \cdots \, \left(P_q^T\right)^{q-1}\!\w^T\right]^T 
\end{equation}
which has $1-\delta$ in the principal diagonal entries and $\delta/(q-1)$ in all other entries regardless of the choice of group $(\X,\oplus)$. We may interpret $\delta$ as the total crossover probability of the symmetric channel. Indeed, when $q = 2$, $W_{\delta}$ represents a BSC with crossover probability $\delta \in [0,1]$. Although $W_{\delta}$ is only stochastic when $\delta \in [0,1]$, we will refer to the parametrized convex set of matrices $\left\{W_{\delta} \in \R^{q \times q}_{\textsf{sym}} : \delta \in \R\right\}$ with parameter $\delta$ as the ``symmetric channel matrices,'' where each $W_{\delta}$ has the form \eqref{Eq: Symmetric Channel} such that every row and column sums to unity. We conclude this subsection with a list of properties of symmetric channel matrices.

\begin{proposition}[Properties of Symmetric Channel Matrices]
\label{Prop: Properties of Symmetric Channel Matrices}
The symmetric channel matrices, $\left\{W_{\delta} \in \R^{q \times q}_{\textsf{sym}} : \delta \in \R\right\}$, satisfy the following properties:
\begin{enumerate}
\item For every $\delta \in \R$, $W_{\delta}$ is a symmetric circulant matrix.
\item The DFT matrix $F_q \! \in \! \mathbb{C}^{q \times q}$, which is defined entry-wise as $\left[F_q\right]_{j,k} = q^{-1/2} \omega^{(j-1)(k-1)}$ for $1 \leq j,k \leq q$ where $\omega = \exp\left(2\pi i/q\right)$ and $i = \sqrt{-1}$, jointly diagonalizes $W_{\delta}$ for every $\delta \in \R$. Moreover, the corresponding eigenvalues or Fourier coefficients, $\{\lambda_j\left(W_{\delta}\right) = \left[F_q^H W_{\delta} F_q\right]_{j,j} : j = 1,\dots,q \}$ are real:
$$ \lambda_j\left(W_{\delta}\right) = \left\{
     \begin{array}{ll}
       1, & j = 1 \\
       1-\delta-\frac{\delta}{q-1}, & j = 2,\dots,q
     \end{array} \right. $$
where $F_q^H$ denotes the Hermitian transpose of $F_q$.
\item For all $\delta \in [0,1]$, $W_{\delta}$ is a doubly stochastic matrix that has the uniform pmf $\unif \triangleq (1/q,\dots,1/q)$ as its stationary distribution: $\unif W_{\delta} = \unif$.
\item For every $\delta \in \R\backslash\big\{\frac{q-1}{q}\big\}$, $W_{\delta}^{-1} = W_{\tau}$ with $\tau = -\delta/\big(1-\delta-\frac{\delta}{q-1}\big)$, and for $\delta = \frac{q-1}{q}$, $W_{\delta} = \frac{1}{q}\1 \1^T$ is unit rank and singular, where $\1 = [1 \cdots 1]^T$.
\item The set $\big\{W_{\delta} \in \R^{q \times q}_{\textsf{sym}}:\delta \in \R\backslash\big\{\frac{q-1}{q}\big\}\big\}$ with the operation of matrix multiplication is an Abelian group.
\end{enumerate}
\end{proposition}

\begin{proof} See Appendix \ref{App: Proofs of Propositions}.
\end{proof}

\subsection{Main question and motivation}
\label{Main Question and Motivation}

As we mentioned at the outset, our work is partly motivated by \cite[Section 6]{GraphSDPI}, where the authors demonstrate an intriguing relation between less noisy domination by an erasure channel and the contraction coefficient of the SDPI \eqref{Eq: SDPI}. For a common input alphabet $\X = [q]$, consider a channel $V \in \R^{q \times s}_{\textsf{sto}}$ and a \textit{$q$-ary erasure channel} $E_{\epsilon} \in \R^{q \times (q+1)}_{\textsf{sto}}$ with erasure probability $\epsilon \in [0,1]$. Recall that given an input $x \in \X$, a $q$-ary erasure channel erases $x$ and outputs $\textsf{\small e}$ (erasure symbol) with probability $\epsilon$, and outputs $x$ itself with probability $1-\epsilon$; the output alphabet of the erasure channel is $\left\{\textsf{\small e}\right\}\cup\X$. It is proved in \cite[Proposition 15]{GraphSDPI} that $E_{\epsilon} \succeq_{\textsf{\tiny ln}} V$ if and only if $\eta_{\textsf{\tiny KL}}\!\left(V\right) \leq 1-\epsilon$, where $\eta_{\textsf{\tiny KL}}\!\left(V\right) \in [0,1]$ is the \textit{contraction coefficient} for KL divergence:
\begin{equation}
\label{Eq: KL Contraction Coefficient}
\eta_{\textsf{\tiny KL}}\!\left(V\right) \triangleq \sup_{\substack{P_X,Q_X \in \P_q \\ 0 < D(P_X||Q_X) < +\infty}}{\frac{D\left(P_X V||Q_X V\right)}{D\left(P_X||Q_X\right)}}
\end{equation}
which equals the best possible constant $\eta$ in the SDPI \eqref{Eq: SDPI} when $V = P_{Y|X}$ (see \cite[Theorem 4]{GraphSDPI} and the references therein). This result illustrates that the $q$-ary erasure channel $E_{\epsilon}$ with the largest erasure probability $\epsilon \in [0,1]$ (or the smallest channel capacity) that is less noisy than $V$ has $\epsilon = 1-\eta_{\textsf{\tiny KL}}\!\left(V\right)$.\footnote{A $q$-ary erasure channel $E_{\epsilon}$ with erasure probability $\epsilon \in [0,1]$ has channel capacity $C(\epsilon) = \log(q) (1-\epsilon)$, which is linear and decreasing.} Furthermore, there are several simple upper bounds on $\eta_{\textsf{\tiny KL}}$ that provide sufficient conditions for such less noisy domination. For example, if the $\ell^1$-distances between the rows of $V$ are all bounded by $2 \alpha$ for some $\alpha \in [0,1]$, then $\eta_{\textsf{\tiny KL}}\!\left(V\right) \leq \alpha$, cf. \cite{CIR93}. Another criterion follows from \textit{Doeblin minorization} \cite[Remark III.2]{SDPIandSobolevInequalities}: if for some pmf $p \in \P_s$ and some $\alpha \in (0,1)$, $V \geq \alpha \, \1 p$ entry-wise, then $E_{\alpha} \succeq_{\textsf{\tiny deg}} V$ and $\eta_{\textsf{\tiny KL}}\!\left(V\right) \leq 1-\alpha$. 

To extend these ideas, we consider the following question: \textit{What is the $q$-ary symmetric channel $W_{\delta}$ with the largest value of $\delta \in \big[0,\frac{q-1}{q}\big]$ (or the smallest channel capacity) such that $W_{\delta} \succeq_{\textsf{\tiny ln}} V$?}\footnote{A $q$-ary symmetric channel $W_{\delta}$ with total crossover probability $\delta \in \big[0,\frac{q-1}{q}\big]$ has channel capacity $C(\delta) = \log(q) - H(\w)$, which is convex and decreasing. Here, $H(\w)$ denotes the Shannon entropy of the pmf $\w$.} Much like the bounds on $\eta_{\textsf{\tiny KL}}$ in the erasure channel context, the goal of this paper is to address this question by establishing simple criteria for testing $\succeq_{\textsf{\tiny ln}}$ domination by a $q$-ary symmetric channel. We next provide several other reasons why determining whether a $q$-ary symmetric channel dominates a given channel $V$ is interesting.
 
Firstly, if $W \succeq_{\textsf{\tiny ln}} V$, then $W^{\otimes n} \succeq_{\textsf{\tiny ln}} V^{\otimes n}$ (where $W^{\otimes n}$ is the $n$-fold tensor product of $W$) since $\succeq_{\textsf{\tiny ln}}$ tensorizes, and $I(U; Y_W^n) \geq I(U; Y_V^n)$ for every Markov chain $U \rightarrow X^n \rightarrow (Y_W^n,Y_V^n)$ (see Definition \ref{Def: Less Noisy Preorder}). Thus, many impossibility results (in statistical decision theory for example) that are proven by exhibiting bounds on quantities such as $I(U;Y_W^n)$ transparently carry over to statistical experiments with observations on the basis of $Y_V^n$. Since it is common to study the $q$-ary symmetric observation model (especially with $q=2$), we can leverage its sample complexity lower bounds for other $V$.

Secondly, we present a self-contained information theoretic motivation. $W \succeq_{\textsf{\tiny ln}} V$ if and only if $C_S = 0$, where $C_S$ is the secrecy capacity of the \textit{Wyner wiretap channel} with $V$ as the main (legal receiver) channel and $W$ as the eavesdropper channel \cite[Corollary 3]{CsiszarKorner1978}, \cite[Corollary 17.11]{CsiszarKorner}. Therefore, finding the maximally noisy $q$-ary symmetric channel that dominates $V$ establishes the minimal noise required on the eavesdropper link so that secret communication is feasible.

Thirdly, $\succeq_{\textsf{\tiny ln}}$ domination turns out to entail a comparison of Dirichlet forms (see subsection \ref{Comparison of Dirichlet Forms}), and consequently, allows us to prove \textit{Poincar\'{e} and logarithmic Sobolev inequalities} for $V$ from well-known results on $q$-ary symmetric channels. These inequalities are cornerstones of the modern approach to Markov chains and concentration of measure \cite{LogSobolevInequalitiesDiaconis,AnalysisofMixingTimes}.

\section{Main results}
\label{Main Results}

In this section, we first delineate some guiding sub-questions of our study, indicate the main results that address them, and then present these results in the ensuing subsections. We will delve into the following four leading questions: 
\begin{enumerate}
\item \textit{Can we test the less noisy preorder $\succeq_{\textsf{\tiny ln}}$ without using KL divergence?} \\
Yes, we can use $\chi^2$-divergence as shown in Theorem \ref{Thm: Chi Squared Divergence Characterization of Less Noisy}.
\item \textit{Given a channel $V \in \R^{q \times q}_{\textsf{sto}}$, is there a simple sufficient condition for less noisy domination by a $q$-ary symmetric channel $W_{\delta} \succeq_{\textsf{\tiny ln}} V$?} \\
Yes, a condition using degradation (which implies less noisy domination) is presented in Theorem \ref{Thm: Sufficient Condition for Degradation by Symmetric Channels}.
\item \textit{Can we say anything stronger about less noisy domination by a $q$-ary symmetric channel when $V \in \R^{q \times q}_{\textsf{sto}}$ is an additive noise channel?} \\
Yes, Theorem \ref{Thm: Additive Less Noisy Domination and Degradation Regions for Symmetric Channels} outlines the structure of additive noise channels in this context (and Figure \ref{Figure: Additive Noise Channel Domination Regions} depicts it).
\item \textit{Why are we interested in less noisy domination by $q$-ary symmetric channels?} \\
Because this permits us to compare Dirichlet forms as portrayed in Theorem \ref{Thm: Domination of Dirichlet Forms}.
\end{enumerate}
We next elaborate on these aforementioned theorems.

\subsection{$\chi^2$-divergence characterization of the less noisy preorder}
\label{Chi Squared Divergence Characterization of the Less Noisy Preorder}

Our most general result illustrates that although less noisy domination is a preorder defined using KL divergence, one can equivalently define it using $\chi^2$-divergence. Since we will prove this result for general measurable spaces, we introduce some notation pertinent only to this result. Let $(\X,\F)$, $(\Y_1,\H_1)$, and $(\Y_2,\H_2)$ be three measurable spaces, and let $W:\H_1 \times \X \rightarrow [0,1]$ and $V:\H_2 \times \X \rightarrow [0,1]$ be two \textit{Markov kernels} (or channels) acting on the same source space $(\X,\F)$. Given any probability measure $P_X$ on $(\X,\F)$, we denote by $P_X W$ the probability measure on $(\Y_1,\H_1)$ induced by the push-forward of $P_X$.\footnote{Here, we can think of $X$ and $Y$ as random variables with codomains $\X$ and $\Y$, respectively. The Markov kernel $W$ behaves like the conditional distribution of $Y$ given $X$ (under regularity conditions). Moreover, when the distribution of $X$ is $P_X$, the corresponding distribution of $Y$ is $P_Y = P_X W$.} Recall that for any two probability measures $P_X$ and $Q_X$ on $(\X,\F)$, their KL divergence is given by:
\begin{equation}
D\!\left(P_X||Q_X\right) \triangleq \left\{
\begin{array}{ll}
\displaystyle{\int_{\X}{\log\!\left(\frac{d P_X}{d Q_X}\right) d P_X}} , & \text{if} \enspace P_X \ll Q_X \\
+\infty , & \text{otherwise}
\end{array} \right.
\end{equation}
and their \textit{$\chi^2$-divergence} is given by:
\begin{equation}
\chi^2 \!\left(P_X||Q_X\right) \triangleq \left\{
\begin{array}{ll}
\displaystyle{\int_{\X}{\left(\frac{d P_X}{d Q_X}\right)^{\! 2} d Q_X} - 1} , & \text{if} \enspace P_X \ll Q_X \\
+\infty , & \text{otherwise}
\end{array} \right.
\end{equation}
where $P_X \ll Q_X$ denotes that $P_X$ is absolutely continuous with respect to $Q_X$, $\frac{d P_X}{d Q_X}$ denotes the Radon-Nikodym derivative of $P_X$ with respect to $Q_X$, and $\log\!\left(\cdot\right)$ is the natural logarithm with base $e$ (throughout this paper). Furthermore, the characterization of $\succeq_{\textsf{\tiny ln}}$ in Proposition \ref{Prop: KL Divergence Characterization of Less Noisy} extends naturally to general Markov kernels; indeed, $W \succeq_{\textsf{\tiny ln}} V$ if and only if $D(P_X W||Q_X W) \geq D(P_X V||Q_X V)$ for every pair of probability measures $P_X$ and $Q_X$ on $(\X,\F)$. The next theorem presents the $\chi^2$-divergence characterization of $\succeq_{\textsf{\tiny ln}}$.

\begin{theorem}[$\chi^2$-Divergence Characterization of $\succeq_{\textsf{\tiny ln}}$]
\label{Thm: Chi Squared Divergence Characterization of Less Noisy}
For any Markov kernels $W:\H_1 \times \X \rightarrow [0,1]$ and $V:\H_2 \times \X \rightarrow [0,1]$ acting on the same source space, $W \succeq_{\textsf{\tiny ln}} V$ if and only if:
$$ \chi^2 \!\left(P_X W||Q_X W\right) \geq \chi^2 \!\left(P_X V||Q_X V\right) $$
for every pair of probability measures $P_X$ and $Q_X$ on $(\X,\F)$.
\end{theorem}

Theorem \ref{Thm: Chi Squared Divergence Characterization of Less Noisy} is proved in subsection \ref{Characterization using chi squared divergence}.

\subsection{Less noisy domination by symmetric channels}
\label{Less Noisy Domination by Symmetric Channels}

Our remaining results are all concerned with less noisy (and degraded) domination by $q$-ary symmetric channels. Suppose we are given a $q$-ary symmetric channel $W_{\delta} \in \R^{q \times q}_{\textsf{sto}}$ with $\delta \in [0,1]$, and another channel $V \in \R^{q \times q}_{\textsf{sto}}$ with common input and output alphabets. Then, the next result provides a sufficient condition for when $W_{\delta} \succeq_{\textsf{\tiny deg}} V$.

\begin{theorem}[Sufficient Condition for Degradation by Symmetric Channels]
\label{Thm: Sufficient Condition for Degradation by Symmetric Channels}
Given a channel $V \in \R^{q \times q}_{\textsf{sto}}$ with $q \geq 2$ and minimum probability $\nu = \min\left\{[V]_{i,j}:1 \leq i,j \leq q\right\}$, we have:
$$ 0 \leq \delta \leq \frac{\nu}{1-(q-1)\nu + \frac{\nu}{q-1}} \enspace \Rightarrow \enspace W_{\delta} \succeq_{\textsf{\tiny deg}} V . $$ 
\end{theorem}

Theorem \ref{Thm: Sufficient Condition for Degradation by Symmetric Channels} is proved in section \ref{Sufficient Conditions for Degradation over General Channels}. We note that the sufficient condition in Theorem \ref{Thm: Sufficient Condition for Degradation by Symmetric Channels} is tight as there exist channels $V$ that violate $W_{\delta} \succeq_{\textsf{\tiny deg}} V$ when $\delta > \nu/(1-(q-1)\nu + \frac{\nu}{q-1})$. Furthermore, Theorem \ref{Thm: Sufficient Condition for Degradation by Symmetric Channels} also provides a sufficient condition for $W_{\delta} \succeq_{\textsf{\tiny ln}} V$ due to Proposition \ref{Prop: Relations between Channel Preorders}.

\subsection{Structure of additive noise channels}
\label{Structure of Additive Noise Channels}

Our next major result is concerned with understanding when $q$-ary symmetric channels operating on an Abelian group $(\X,\oplus)$ dominate other additive noise channels on $(\X,\oplus)$, which are defined in \eqref{Eq: Additive Noise Channel}, in the less noisy and degraded senses. Given a symmetric channel $W_{\delta} \in \R^{q \times q}_{\textsf{sto}}$ with $\delta \in [0,1]$, we define the \textit{additive less noisy domination region} of $W_{\delta}$ as:
\begin{equation}
\label{Eq: Additive Domination Region for Symmetric Channel}
\L_{W_{\delta}}^{\textsf{add}} \triangleq \left\{v \in \P_q : W_{\delta} = \textsf{\small circ}_{\X}\!\left(\w\right) \succeq_{\textsf{\tiny ln}} \textsf{\small circ}_{\X}\!\left(v\right)\right\}
\end{equation}
which is the set of all noise pmfs whose corresponding channel transition probability matrices are dominated by $W_{\delta}$ in the less noisy sense. Likewise, we define the \textit{additive degradation region} of $W_{\delta}$ as:
\begin{equation}
\label{Eq: Additive Degradation Region for Symmetric Channel}
\D_{W_{\delta}}^{\textsf{add}} \triangleq \left\{v \in \P_q : W_{\delta} = \textsf{\small circ}_{\X}\!\left(\w\right) \succeq_{\textsf{\tiny deg}} \textsf{\small circ}_{\X}\!\left(v\right)\right\}
\end{equation}
which is the set of all noise pmfs whose corresponding channel transition probability matrices are degraded versions of $W_{\delta}$. The next theorem exactly characterizes $\D_{W_{\delta}}^{\textsf{add}}$, and ``bounds'' $\L_{W_{\delta}}^{\textsf{add}}$ in a set theoretic sense. 

\begin{theorem}[Additive Less Noisy Domination and Degradation Regions for Symmetric Channels]
\label{Thm: Additive Less Noisy Domination and Degradation Regions for Symmetric Channels}
Given a symmetric channel $W_{\delta} = \textsf{\small circ}_{\X}\!\left(\w\right) \in \R^{q \times q}_{\textsf{sto}}$ with $\delta \in \big[0,\frac{q-1}{q}\big]$ and $q \geq 2$, we have:
\begin{align*}
\D_{W_{\delta}}^{\textsf{add}} & = \textsf{\small conv}\left(\left\{\w P_q^k : k \in [q]\right\}\right) \\
& \subseteq \textsf{\small conv}\left(\left\{\w P_q^k : k \in [q]\right\} \cup \left\{w_{\gamma} P_q^k : k \in [q]\right\}\right) \\
& \subseteq \L_{W_{\delta}}^{\textsf{add}} \subseteq \left\{v \in \P_q : \left\|v - \unif\right\|_{\ell^2} \leq \left\|\w - \unif\right\|_{\ell^2} \right\}
\end{align*}
where the first set inclusion is strict for $\delta \in \big(0,\frac{q-1}{q}\big)$ and $q \geq 3$, $P_q$ denotes the generator cyclic permutation matrix as defined in \eqref{Eq: Generator Cyclic Permutation Matrix}, $\unif$ denotes the uniform pmf, $\left\|\cdot\right\|_{\ell^2}$ is the Euclidean $\ell^2$-norm, and:
$$ \gamma = \frac{1-\delta}{1-\delta+\frac{\delta}{(q-1)^2}} . $$
Furthermore, $\L_{W_{\delta}}^{\textsf{add}}$ is a closed and convex set that is invariant under the permutations $\left\{P_x \in \R^{q \times q} : x \in \X\right\}$ defined in \eqref{Eq: Permutation Representation} corresponding to the underlying Abelian group $(\X,\oplus)$ (i.e. $v \in \L_{W_{\delta}}^{\textsf{add}} \Rightarrow v P_x \in \L_{W_{\delta}}^{\textsf{add}}$ for every $x \in \X$).
\end{theorem}

Theorem \ref{Thm: Additive Less Noisy Domination and Degradation Regions for Symmetric Channels} is a compilation of several results. As explained at the very end of subsection \ref{Sufficient Conditions}, Proposition \ref{Prop: Additive Less Noisy Domination and Degradation Regions} (in subsection \ref{Less Noisy Domination and Degradation Regions for Additive
Noise Channels}), Corollary \ref{Cor: Degradation Region of Symmetric Channel} (in subsection \ref{Less Noisy Domination and Degradation Regions for Symmetric Channels}), part 1 of Proposition \ref{Prop: Necessary Conditions for Less Noisy Domination over Additive Noise Channels} (in subsection \ref{Necessary Conditions}), and Proposition \ref{Prop: Less Noisy Domination by Symmetric Channels} (in subsection \ref{Sufficient Conditions}) make up Theorem \ref{Thm: Additive Less Noisy Domination and Degradation Regions for Symmetric Channels}. We remark that according to numerical evidence, the second and third set inclusions in Theorem \ref{Thm: Additive Less Noisy Domination and Degradation Regions for Symmetric Channels} appear to be strict, and $\L_{W_{\delta}}^{\textsf{add}}$ seems to be a strictly convex set. The content of Theorem \ref{Thm: Additive Less Noisy Domination and Degradation Regions for Symmetric Channels} and these observations are illustrated in Figure \ref{Figure: Additive Noise Channel Domination Regions}, which portrays the probability simplex of noise pmfs for $q = 3$ and the pertinent regions which capture less noisy domination and degradation by a $q$-ary symmetric channel. 

\begin{figure}[!t] 
\centering
\includegraphics[trim = 18mm 145mm 38mm 16mm, width=0.6\linewidth]{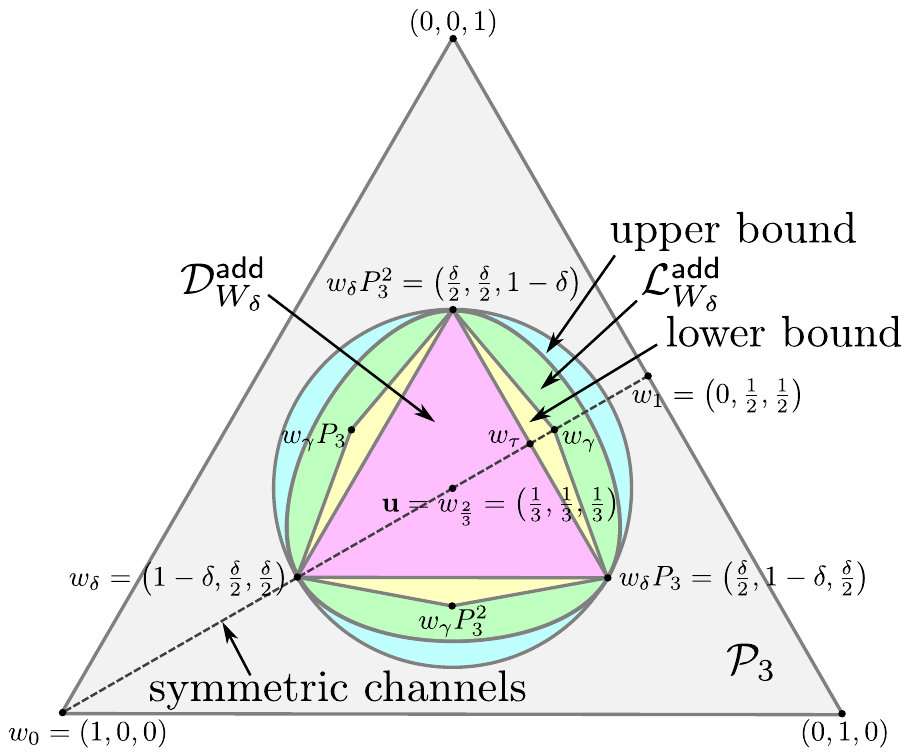} 
\caption{Illustration of the additive less noisy domination region and additive degradation region for a $q$-ary symmetric channel when $q = 3$ and $\delta \in \left(0,2/3\right)$: The gray triangle denotes the probability simplex of noise pmfs $\P_3$. The dotted line denotes the parametrized family of noise pmfs of $3$-ary symmetric channels $\left\{\w \in \P_3:\delta \in [0,1]\right\}$; its noteworthy points are $w_{0}$ (corner of simplex, $W_0$ is less noisy than every channel), $w_{\delta}$ for some fixed $\delta \in \left(0,2/3\right)$ (noise pmf of $3$-ary symmetric channel $W_{\delta}$ under consideration), $w_{2/3} = \unif$ (uniform pmf, $W_{2/3}$ is more noisy than every channel), $w_{\tau}$ with $\tau = 1-\left(\delta/2\right)$ ($W_{\tau}$ is the extremal symmetric channel that is degraded by $W_{\delta}$), $w_{\gamma}$ with $\gamma = (1-\delta)/(1-\delta+(\delta/4))$ ($W_{\gamma}$ is a $3$-ary symmetric channel that is not degraded by $W_{\delta}$ but $W_{\delta} \succeq_{\textsf{\tiny ln}} W_{\gamma}$), and $w_{1}$ (edge of simplex). The magenta triangle denotes the additive degradation region $\textsf{\scriptsize conv}\!\left(\left\{\w,\w P_3,\w P_3^2\right\}\right)$ of $W_{\delta}$. The green convex region denotes the additive less noisy domination region of $W_{\delta}$, and the yellow region $\textsf{\scriptsize conv}\!\left(\left\{\w,\w P_3,\w P_3^2,w_{\gamma},w_{\gamma} P_3,w_{\gamma} P_3^2\right\}\right)$ is its lower bound while the circular cyan region $\left\{v \in \mathcal{P}_3 : \left\|v - \textbf{u}\right\|_{\ell^2} \leq \left\|w_{\delta} - \textbf{u}\right\|_{\ell^2}\right\}$ (which is a hypersphere for general $q \geq 3$) is its upper bound. Note that we do not need to specify the underlying group because there is only one group of order $3$.}
\label{Figure: Additive Noise Channel Domination Regions}
\end{figure}

\subsection{Comparison of Dirichlet forms}
\label{Comparison of Dirichlet Forms}

As mentioned in subsection \ref{Main Question and Motivation}, one of the reasons we study $q$-ary symmetric channels and prove Theorems \ref{Thm: Sufficient Condition for Degradation by Symmetric Channels} and \ref{Thm: Additive Less Noisy Domination and Degradation Regions for Symmetric Channels} is because less noisy domination implies useful bounds between Dirichlet forms. Recall that the $q$-ary symmetric channel $W_{\delta} \in \R^{q \times q}_{\textsf{sto}}$ with $\delta \in [0,1]$ has uniform stationary distribution $\unif \in \P_q$ (see part 3 of Proposition \ref{Prop: Properties of Symmetric Channel Matrices}). For any channel $V \in \R^{q \times q}_{\textsf{sto}}$ that is doubly stochastic and has uniform stationary distribution, we may define a corresponding \textit{Dirichlet form}:
\begin{equation}
\forall f \in \R^q, \enspace \Dirichlet_V\left(f,f\right) = \frac{1}{q} f^T \left(I_q - V\right) f
\end{equation}
where $f = \left[f_1 \cdots f_q\right]^T \in \R^q$ are column vectors, and $I_q \in \R^{q \times q}$ denotes the $q \times q$ identity matrix (as shown in \cite{LogSobolevInequalitiesDiaconis} or \cite{AnalysisofMixingTimes}). Our final theorem portrays that $W_{\delta} \succeq_{\textsf{\tiny ln}} V$ implies that the Dirichlet form corresponding to $V$ dominates the Dirichlet form corresponding to $W_{\delta}$ pointwise. The Dirichlet form corresponding to $W_{\delta}$ is in fact a scaled version of the so called \textit{standard Dirichlet form}:
\begin{equation}
\label{Eq: Standard Dirichlet Form Quadratic Form}
\forall f \in \R^q, \enspace \Dirichlet_{\textsf{std}}\left(f,f\right) \triangleq \VAR_{\unif}(f) = \frac{1}{q}\sum_{k = 1}^{q}{f_k^2} - \left(\frac{1}{q}\sum_{k = 1}^{q}{f_k}\right)^{\!2}
\end{equation}
which is the Dirichlet form corresponding to the $q$-ary symmetric channel $W_{(q-1)/q} = \1\unif$ with all uniform conditional pmfs. Indeed, using $I_q - W_{\delta} = \frac{q\delta}{q-1}(I_q - \1\unif)$, we have:
\begin{equation}
\label{Eq: Symmetric Channel Dirichlet Form Quadratic Form}
\forall f \in \R^q, \enspace \Dirichlet_{W_{\delta}}\left(f,f\right) = \frac{q\delta}{q-1}\Dirichlet_{\textsf{std}}\left(f,f\right) .
\end{equation}
The standard Dirichlet form is the usual choice for Dirichlet form comparison because its logarithmic Sobolev constant has been precisely computed in \cite[Appendix, Theorem A.1]{LogSobolevInequalitiesDiaconis}. So, we present Theorem \ref{Thm: Domination of Dirichlet Forms} using $\Dirichlet_{\textsf{std}}$ rather than $\Dirichlet_{W_{\delta}}$.

\begin{theorem}[Domination of Dirichlet Forms]
\label{Thm: Domination of Dirichlet Forms}
Given the doubly stochastic channels $W_{\delta} \in \R^{q \times q}_{\textsf{sto}}$ with $\delta \in \big[0,\frac{q-1}{q}\big]$ and $V \in \R^{q \times q}_{\textsf{sto}}$, if $W_{\delta} \succeq_{\textsf{\tiny ln}} V$, then:
$$ \forall f \in \R^q, \enspace \Dirichlet_{V}\left(f,f\right) \geq \frac{q\delta}{q-1}\Dirichlet_{\textsf{std}}\left(f,f\right) . $$ 
\end{theorem}   

An extension of Theorem \ref{Thm: Domination of Dirichlet Forms} is proved in section \ref{Less Noisy Domination and Logarithmic Sobolev Inequalities}. The domination of Dirichlet forms shown in Theorem \ref{Thm: Domination of Dirichlet Forms} has several useful consequences. A major consequence is that we can immediately establish Poincar\'{e} (spectral gap) inequalities and logarithmic Sobolev inequalities (LSIs) for the channel $V$ using the corresponding inequalities for $q$-ary symmetric channels. For example, the LSI for $W_{\delta} \in \R^{q \times q}_{\textsf{sto}}$ with $q > 2$ is:
\begin{equation}
\label{Eq: LSI for Symmetric Channel}
D\left(f^2 \unif || \unif\right) \leq \frac{(q-1)\log(q-1)}{(q-2)\delta} \, \Dirichlet_{W_{\delta}}\left(f,f\right)
\end{equation}
for all $f \in \R^{q}$ such that $\sum_{k = 1}^{q}{f_k^2} = q$, where we use \eqref{Eq: Log-Sobolev Inequality} and the logarithmic Sobolev constant computed in part 1 of Proposition \ref{Prop: Constants of Symmetric Channels}. As shown in Appendix \ref{App: Proofs of Propositions}, \eqref{Eq: LSI for Symmetric Channel} is easily established using the known logarithmic Sobolev constant corresponding to the standard Dirichlet form. Using the LSI for $V$ that follows from \eqref{Eq: LSI for Symmetric Channel} and Theorem \ref{Thm: Domination of Dirichlet Forms}, we immediately obtain guarantees on the convergence rate and \textit{hypercontractivity} properties of the associated \textit{Markov semigroup} $\left\{\exp(-t(I_q - V)) : t \geq 0\right\}$. We refer readers to \cite{LogSobolevInequalitiesDiaconis} and \cite{AnalysisofMixingTimes} for comprehensive accounts of such topics.   

\subsection{Outline}
\label{Outline}

We briefly outline the content of the ensuing sections. In section \ref{Less Noisy Domination and Degradation Regions}, we study the structure of less noisy domination and degradation regions of channels. In section \ref{Equivalent Characterizations of Less Noisy Preorder}, we prove Theorem \ref{Thm: Chi Squared Divergence Characterization of Less Noisy} and present some other equivalent characterizations of $\succeq_{\textsf{\tiny ln}}$. We then derive several necessary and sufficient conditions for less noisy domination among additive noise channels in section \ref{Conditions for Less Noisy Domination over Additive Noise Channels}, which together with the results of section \ref{Less Noisy Domination and Degradation Regions}, culminates in a proof of Theorem \ref{Thm: Additive Less Noisy Domination and Degradation Regions for Symmetric Channels}. Section \ref{Sufficient Conditions for Degradation over General Channels} provides a proof of Theorem \ref{Thm: Sufficient Condition for Degradation by Symmetric Channels}, and section \ref{Less Noisy Domination and Logarithmic Sobolev Inequalities} introduces LSIs and proves an extension of Theorem \ref{Thm: Domination of Dirichlet Forms}. Finally, we conclude our discussion in section \ref{Conclusion}.

\section{Less noisy domination and degradation regions}
\label{Less Noisy Domination and Degradation Regions}

In this section, we focus on understanding the ``geometric'' aspects of less noisy domination and degradation by channels. We begin by deriving some simple characteristics of the sets of channels that are dominated by some fixed channel in the less noisy and degraded senses. We then specialize our results for additive noise channels, and this culminates in a complete characterization of $\D_{W_{\delta}}^{\textsf{add}}$ and derivations of certain properties of $\L_{W_{\delta}}^{\textsf{add}}$ presented in Theorem \ref{Thm: Additive Less Noisy Domination and Degradation Regions for Symmetric Channels}. 

Let $W \in \R^{q \times r}_{\textsf{sto}}$ be a fixed channel with $q,r \in \N$, and define its \textit{less noisy domination region}:
\begin{equation}
\label{Eq: Domination Region}
\L_{W} \triangleq \left\{V \in \R^{q \times r}_{\textsf{sto}} : W \succeq_{\textsf{\tiny ln}} V\right\}
\end{equation}
as the set of all channels on the same input and output alphabets that are dominated by $W$ in the less noisy sense. Moreover, we define the \textit{degradation region} of $W$:
\begin{equation}
\label{Eq: Degradation Region}
\D_{W} \triangleq \left\{V \in \R^{q \times r}_{\textsf{sto}} : W \succeq_{\textsf{\tiny deg}} V\right\}
\end{equation}
as the set of all channels on the same input and output alphabets that are degraded versions of $W$. Then, $\L_{W}$ and $\D_W$ satisfy the properties delineated below.

\begin{proposition}[Less Noisy Domination and Degradation Regions]
\label{Prop: Less Noisy Domination and Degradation Regions}
Given the channel $W \in \R^{q \times r}_{\textsf{sto}}$, its less noisy domination region $\L_W$ and its degradation region $\D_W$ are non-empty, closed, convex, and output alphabet permutation symmetric (i.e. $V \in \L_W \Rightarrow VP \in \L_W$ and $V \in \D_W \Rightarrow VP \in \D_W$ for every permutation matrix $P \in \R^{r \times r}$).
\end{proposition} 

\begin{proof} ~\newline
\textbf{Non-Emptiness of $\L_W$ and $\D_W$:} $W \succeq_{\textsf{\tiny ln}} W \Rightarrow W \in \L_W$, and $W \succeq_{\textsf{\tiny deg}} W \Rightarrow W \in \D_W$. So, $\L_W$ and $\D_W$ are non-empty. \\
\textbf{Closure of $\L_W$:} Fix any two pmfs $P_X,Q_X \in \P_{q}$, and consider a sequence of channels $V_k \in \L_W$ such that $V_k \rightarrow V \in \R^{q \times r}_{\textsf{sto}}$ (with respect to the Frobenius norm). Then, we also have $P_X V_k \rightarrow P_X V $ and $Q_X V_k \rightarrow Q_X V$ (with respect to the $\ell^2$-norm). Hence, we get: 
\begin{align*}
D\left(P_X V||Q_X V\right) & \leq \liminf_{k \rightarrow \infty}{D\left(P_X V_k||Q_X V_k\right)} \\
& \leq D\left(P_X W||Q_X W\right) 
\end{align*}
where the first line follows from the lower semicontinuity of KL divergence \cite[Theorem 1]{LowerSemicontinuityDivergence}, \cite[Theorem 3.6, Section 3.5]{InfoTheoryNotes}, and the second line holds because $V_k \in \L_W$. This implies that for any two pmfs $P_X,Q_X \in \P_q$, the set $\mathcal{S}\left(P_X,Q_X\right) = \left\{V \in \R^{q \times r}_{\textsf{sto}} : D\left(P_X W||Q_X W\right) \geq D\left(P_X V||Q_X V\right)\right\}$ is actually closed. Using Proposition \ref{Prop: KL Divergence Characterization of Less Noisy}, we have that:
$$ \L_W = \bigcap_{P_X,Q_X \in \P_q}{\mathcal{S}\left(P_X,Q_X\right)} . $$ 
So, $\L_W$ is closed since it is an intersection of closed sets \cite{Rudin}. \\
\textbf{Closure of $\D_W$:} Consider a sequence of channels $V_k \in \D_W$ such that $V_k \rightarrow V \in \R^{q \times r}_{\textsf{sto}}$. Since each $V_k = W A_k$ for some channel $A_k \in \R^{r \times r}_{\textsf{sto}}$ belonging to the compact set $\R^{r \times r}_{\textsf{sto}}$, there exists a subsequence $A_{k_m}$ that converges by (sequential) compactness \cite{Rudin}: $A_{k_m} \rightarrow A \in \R^{r \times r}_{\textsf{sto}}$. Hence, $V \in \D_W$ since $V_{k_m} = W A_{k_m} \rightarrow W A = V$, and $\D_W$ is a closed set. \\
\textbf{Convexity of $\L_W$:} Suppose $V_1,V_2 \in \L_W$, and let $\lambda \in [0,1]$ and $\bar{\lambda} = 1 - \lambda$. Then, for every $P_X,Q_X \in \P_q$, we have:
$$ D(P_X W||Q_X W) \geq D(P_X (\lambda V_1 + \bar{\lambda} V_2) || Q_X (\lambda V_1 + \bar{\lambda} V_2)) $$
by the convexity of KL divergence. Hence, $\L_W$ is convex. \\
\textbf{Convexity of $\D_W$:} If $V_1,V_2 \in \D_W$ so that $V_1 = W A_1$ and $V_2 = W A_2$ for some $A_1,A_2 \in \R^{r \times r}_{\textsf{sto}}$, then $\lambda V_1 + \bar{\lambda} V_2 = W(\lambda A_1 + \bar{\lambda} A_2) \in \D_W$ for all $\lambda \in [0,1]$, and $\D_W$ is convex. \\
\textbf{Symmetry of $\L_W$:} This is obvious from Proposition \ref{Prop: KL Divergence Characterization of Less Noisy} because KL divergence is invariant to permutations of its input pmfs. \\
\textbf{Symmetry of $\D_W$:} Given $V \in \D_W$ so that $V = WA$ for some $A \in \R^{r \times r}_{\textsf{sto}}$, we have that $VP = WAP \in \D_W$ for every permutation matrix $P \in \R^{r \times r}$. This completes the proof.
\end{proof}

While the channels in $\L_W$ and $\D_W$ all have the same output alphabet as $W$, as defined in \eqref{Eq: Domination Region} and \eqref{Eq: Degradation Region}, we may extend the output alphabet of $W$ by adding zero probability letters. So, separate less noisy domination and degradation regions can be defined for each output alphabet size that is at least as large as the original output alphabet size of $W$. 

\subsection{Less noisy domination and degradation regions for additive noise channels}
\label{Less Noisy Domination and Degradation Regions for Additive Noise Channels}

Often in information theory, we are concerned with additive noise channels on an Abelian group $(\X,\oplus)$ with $\X = [q]$ and $q \in \N$, as defined in \eqref{Eq: Additive Noise Channel}. Such channels are completely defined by a noise pmf $P_Z \in \P_q$ with corresponding channel transition probability matrix $\textsf{\small circ}_{\X}\!\left(P_Z\right) \in \R^{q \times q}_{\textsf{sto}}$. Suppose $W = \textsf{\small circ}_{\X}\!\left(w\right) \in \R^{q \times q}_{\textsf{sto}}$ is an additive noise channel with noise pmf $w \in \P_q$. Then, we are often only interested in the set of additive noise channels that are dominated by $W$. We define the \textit{additive less noisy domination region} of $W$:
\begin{equation}
\label{Eq: Additive Domination Region}
\L_{W}^{\textsf{add}} \triangleq \left\{v \in \P_q : W \succeq_{\textsf{\tiny ln}} \textsf{\small circ}_{\X}\!\left(v\right)\right\}
\end{equation}
as the set of all noise pmfs whose corresponding channel transition matrices are dominated by $W$ in the less noisy sense. Likewise, we define the \textit{additive degradation region} of $W$:
\begin{equation}
\label{Eq: Additive Degradation Region}
\D_{W}^{\textsf{add}} \triangleq \left\{v \in \P_q : W \succeq_{\textsf{\tiny deg}} \textsf{\small circ}_{\X}\!\left(v\right)\right\}
\end{equation}
as the set of all noise pmfs whose corresponding channel transition matrices are degraded versions of $W$. (These definitions generalize \eqref{Eq: Additive Domination Region for Symmetric Channel} and \eqref{Eq: Additive Degradation Region for Symmetric Channel}, and can also hold for any non-additive noise channel $W$.) The next proposition illustrates certain properties of $\L_{W}^{\textsf{add}}$ and explicitly characterizes $\D_{W}^{\textsf{add}}$.

\begin{proposition}[Additive Less Noisy Domination and Degradation Regions]
\label{Prop: Additive Less Noisy Domination and Degradation Regions}
Given the additive noise channel $W = \textsf{\small circ}_{\X}(w) \in \R^{q \times q}_{\textsf{sto}}$ with noise pmf $w \in \P_q$, we have:
\begin{enumerate}
\item $\L_W^{\textsf{add}}$ and $\D_W^{\textsf{add}}$ are non-empty, closed, convex, and invariant under the permutations $\left\{P_x \in \R^{q \times q}: x \in \X\right\}$ defined in \eqref{Eq: Permutation Representation} (i.e. $v \in \L_{W}^{\textsf{add}} \Rightarrow v P_x \in \L_{W}^{\textsf{add}}$ and $v \in \D_{W}^{\textsf{add}} \Rightarrow v P_x \in \D_{W}^{\textsf{add}}$ for every $x \in \X$). 
\item $\D_W^{\textsf{add}} = \textsf{\small conv}\left(\left\{w P_x : x \in \X\right\}\right) = \left\{v \in \P_q : w \succeq_{\text{\tiny $\X$}} v\right\}$, where $\succeq_{\text{\tiny $\X$}}$ denotes the group majorization preorder as defined in Appendix \ref{App: Basics of Majorization Theory}.
\end{enumerate}
\end{proposition} 

To prove Proposition \ref{Prop: Additive Less Noisy Domination and Degradation Regions}, we will need the following lemma.

\begin{lemma}[Additive Noise Channel Degradation]
\label{Lemma: Additive Noise Channel Degradation}
Given two additive noise channels $W = \textsf{\small circ}_{\X}(w) \in \R^{q \times q}_{\textsf{sto}}$ and $V = \textsf{\small circ}_{\X}(v) \in \R^{q \times q}_{\textsf{sto}}$ with noise pmfs $w,v \in \P_q$, $W \succeq_{\textsf{\tiny deg}} V$ if and only if $V = W \textsf{\small circ}_{\X}(z) = \textsf{\small circ}_{\X}(z) W$ for some $z \in \P_q$ (i.e. for additive noise channels $W \succeq_{\textsf{\tiny deg}} V$, the channel that degrades $W$ to produce $V$ is also an additive noise channel without loss of generality).  
\end{lemma}

\begin{proof} Since $\X$-circulant matrices commute, we must have $W \textsf{\small circ}_{\X}(z) = \textsf{\small circ}_{\X}(z) W$ for every $z \in \P_q$. Furthermore, $V = W \textsf{\small circ}_{\X}(z)$ for some $z \in \P_q$ implies that $W \succeq_{\textsf{\tiny deg}} V$ by Definition \ref{Def: Degradation Preorder}. So, it suffices to prove that $W \succeq_{\textsf{\tiny deg}} V$ implies $V = W \textsf{\small circ}_{\X}(z)$ for some $z \in \P_q$. By Definition \ref{Def: Degradation Preorder}, $W \succeq_{\textsf{\tiny deg}} V$ implies that $V = W R$ for some doubly stochastic channel $R \in \R^{q \times q}_{\textsf{sto}}$ (as $V$ and $W$ are doubly stochastic). Let $r$ with $r^T \in \P_q$ be the first column of $R$, and $s = W r$ with $s^T \in \P_q$ be the first column of $V$. Then, it is straightforward to verify using \eqref{Eq: Group Circulant Matrix Row-Column Decomposition} that:
\begin{align*}
V & = \left[
\begin{array}{ccccc} 
s & P_1 s & P_2 s & \cdots & P_{q-1} s
\end{array}
\right] \\
& = \left[
\begin{array}{ccccc} 
Wr & P_1 Wr & P_2 Wr & \cdots & P_{q-1} Wr
\end{array}
\right] \\
& = W \left[
\begin{array}{ccccc} 
r & P_1 r & P_2 r & \cdots & P_{q-1} r
\end{array}
\right]
\end{align*}
where the third equality holds because $\left\{P_x : x \in \X\right\}$ are $\X$-circulant matrices which commute with $W$. Hence, $V$ is the product of $W$ and an $\X$-circulant stochastic matrix, i.e. $V = W \textsf{\small circ}_{\X}(z)$ for some $z \in \P_q$. This concludes the proof.
\end{proof}

We emphasize that in Lemma \ref{Lemma: Additive Noise Channel Degradation}, the channel that degrades $W$ to produce $V$ is only an additive noise channel without loss of generality. We can certainly have $V = WR$ with a non-additive noise channel $R$. Consider for instance, $V = W = \1 \1^T /q$, where every doubly stochastic matrix $R$ satisfies $V = WR$. However, when we consider $V = WR$ with an additive noise channel $R$, $V$ corresponds to the channel $W$ with an additional independent additive noise term associated with $R$. We now prove Proposition \ref{Prop: Additive Less Noisy Domination and Degradation Regions}.

\renewcommand{\proofname}{Proof of Proposition \ref{Prop: Additive Less Noisy Domination and Degradation Regions}}

\begin{proof} ~\newline
\textbf{Part 1:} Non-emptiness, closure, and convexity of $\L_W^{\textsf{add}}$ and $\D_W^{\textsf{add}}$ can be proved in exactly the same way as in Proposition \ref{Prop: Less Noisy Domination and Degradation Regions}, with the additional observation that the set of $\X$-circulant matrices is closed and convex. Moreover, for every $x \in \X$:
\begin{alignat}{5}
W & \succeq_{\textsf{\tiny ln}} & \, W P_x = \textsf{\small circ}_{\X}\!\left(w P_x\right) & \succeq_{\textsf{\tiny ln}} & \, W \nonumber \\
W & \succeq_{\textsf{\tiny deg}} & \, W P_x = \textsf{\small circ}_{\X}\!\left(w P_x\right) & \succeq_{\textsf{\tiny deg}} & \, W \nonumber
\end{alignat}
where the equalities follow from \eqref{Eq: Group Circulant Matrix Row-Column Decomposition}. These inequalities and the transitive properties of $\succeq_{\textsf{\tiny ln}}$ and $\succeq_{\textsf{\tiny deg}}$ yield the invariance of $\L_{W}^{\textsf{add}}$ and $\D_{W}^{\textsf{add}}$ with respect to $\left\{P_x \in \R^{q \times q} : x \in \X\right\}$. \\
\textbf{Part 2:} Lemma \ref{Lemma: Additive Noise Channel Degradation} is equivalent to the fact that $v \in \D_W^{\textsf{add}}$ if and only if $\textsf{\small circ}_{\X}(v) = \textsf{\small circ}_{\X}(w) \, \textsf{\small circ}_{\X}(z)$ for some $z \in \P_q$. This implies that $v \in \D_W^{\textsf{add}}$ if and only if $v = w \, \textsf{\small circ}_{\X}(z)$ for some $z \in \P_q$ (due to \eqref{Eq: Group Circulant Matrix Row-Column Decomposition} and the fact that $\X$-circulant matrices commute). Applying Proposition \ref{Prop: Group Majorization} from Appendix \ref{App: Basics of Majorization Theory} completes the proof.
\end{proof}

\renewcommand{\proofname}{Proof}

We remark that part 1 of Proposition \ref{Prop: Additive Less Noisy Domination and Degradation Regions} does not require $W$ to be an additive noise channel. The proofs of closure, convexity, and invariance with respect to $\left\{P_x \in \R^{q \times q} : x \in \X\right\}$ hold for general $W \in \R^{q \times q}_{\textsf{sto}}$. Moreover, $\L_W^{\textsf{add}}$ and $\D_W^{\textsf{add}}$ are non-empty because $\unif \in \L_W^{\textsf{add}}$ and $\unif \in \D_W^{\textsf{add}}$. 

\subsection{Less noisy domination and degradation regions for symmetric channels}
\label{Less Noisy Domination and Degradation Regions for Symmetric Channels}

Since $q$-ary symmetric channels for $q \in \N$ are additive noise channels, Proposition \ref{Prop: Additive Less Noisy Domination and Degradation Regions} holds for symmetric channels. In this subsection, we deduce some simple results that are unique to symmetric channels. The first of these is a specialization of part 2 of Proposition \ref{Prop: Additive Less Noisy Domination and Degradation Regions} which states that the additive degradation region of a symmetric channel can be characterized by traditional majorization instead of group majorization.

\begin{corollary}[Degradation Region of Symmetric Channel]
\label{Cor: Degradation Region of Symmetric Channel}
The $q$-ary symmetric channel $W_{\delta} = \textsf{\small circ}_{\X}(\w) \in \R^{q \times q}_{\textsf{sto}}$ for $\delta \in \left[0,1\right]$ has additive degradation region:
$$ \D_{W_{\delta}}^{\textsf{add}} = \left\{v \in \P_q : \w \succeq_{\textsf{\tiny maj}} v\right\} = \textsf{\small conv}\left(\left\{\w P_q^k : k \in [q]\right\}\right) $$
where $\succeq_{\textsf{\tiny maj}}$ denotes the majorization preorder defined in Appendix \ref{App: Basics of Majorization Theory}, and $P_q \in \R^{q \times q}$ is defined in \eqref{Eq: Generator Cyclic Permutation Matrix}.
\end{corollary}

\begin{proof} From part 2 of Proposition \ref{Prop: Additive Less Noisy Domination and Degradation Regions}, we have that:
\begin{align*}
\D_{W_{\delta}}^{\textsf{add}} & = \textsf{\small conv}\left(\left\{\w P_x : x \in \X\right\}\right) = \textsf{\small conv}\left(\left\{\w P_q^k : k \in [q]\right\}\right) \\
& = \textsf{\small conv}\left(\left\{\w P : P \in \R^{q \times q} \enspace \text{is a permutation matrix}\right\}\right) \\
& = \left\{v \in \P_q : w \succeq_{\textsf{\tiny maj}} v\right\}
\end{align*}
where the second and third equalities hold regardless of the choice of group $(\X,\oplus)$, because the sets of all cyclic or regular permutations of $\w = \left(1-\delta,\delta/(q-1),\dots,\delta/(q-1)\right)$ equal $\left\{\w P_x : x \in \X\right\}$. The final equality follows from the definition of majorization in Appendix \ref{App: Basics of Majorization Theory}.
\end{proof}

With this geometric characterization of the additive degradation region, it is straightforward to find the extremal symmetric channel $W_{\tau}$ that is a degraded version of $W_{\delta}$ for some fixed $\delta \in [0,1]\backslash\big\{\frac{q-1}{q}\big\}$. Indeed, we compute $\tau$ by using the fact that the noise pmf $w_{\tau} \in \textsf{\small conv}\!\left(\left\{\w P_q^{k}: k = 1,\dots,q-1\right\}\right)$:
\begin{equation}
\label{Eq: Extremal Symmetric Channel Equation}
w_{\tau} = \sum_{i = 1}^{q-1}{\lambda_i \w P_q^{i}}
\end{equation}
for some $\lambda_1,\dots,\lambda_{q-1} \in [0,1]$ such that $\lambda_1 + \cdots + \lambda_{q-1} = 1$. Solving \eqref{Eq: Extremal Symmetric Channel Equation} for $\tau$ and $\lambda_1,\dots,\lambda_{q-1}$ yields:
\begin{equation}
\label{Eq: Extremal Degraded Channel}
\tau = 1 - \frac{\delta}{q-1} 
\end{equation}
and $\lambda_1 = \cdots = \lambda_{q-1} = \frac{1}{q-1}$, which means that:
\begin{equation}
w_{\tau} = \frac{1}{q-1}\sum_{i = 1}^{q-1}{\w P_q^{i}} .
\end{equation}
This is illustrated in Figure \ref{Figure: Additive Noise Channel Domination Regions} for the case where $\delta \in \big(0,\frac{q-1}{q}\big)$ and $\tau > \frac{q-1}{q} > \delta$. For $\delta \in \big(0,\frac{q-1}{q}\big)$, the symmetric channels that are degraded versions of $W_{\delta}$ are $\left\{ W_{\gamma} : \gamma \in [\delta,\tau]\right\}$. In particular, for such $\gamma \in [\delta,\tau]$, $W_{\gamma} = W_{\delta} W_{\beta}$ with $\beta = (\gamma - \delta)/(1 - \delta - \frac{\delta}{q-1})$ using the proof of part 5 of Proposition \ref{Prop: Properties of Symmetric Channel Matrices} in Appendix \ref{App: Proofs of Propositions}.

In the spirit of comparing symmetric and erasure channels as done in \cite{BCCapacityNairShamai} for the binary input case, our next result shows that a $q$-ary symmetric channel can never be less noisy than a $q$-ary erasure channel. 

\begin{proposition}[Symmetric Channel $\not\succeq_{\textsf{\tiny ln}}$ Erasure Channel]
\label{Prop: Symmetric Channel not Less Noisy than Erasure Channel}
For $q \in \N\backslash\{1\}$, given a $q$-ary erasure channel $E_{\epsilon} \in \R^{q \times (q+1)}_{\textsf{sto}}$ with erasure probability $\epsilon \in (0,1)$, there does not exist $\delta \in (0,1)$ such that the corresponding $q$-ary symmetric channel $W_{\delta} \in \R^{q \times q}_{\textsf{sto}}$ on the same input alphabet satisfies $W_{\delta} \succeq_{\textsf{\tiny ln}} E_{\epsilon}$. 
\end{proposition}

\begin{proof}
For a $q$-ary erasure channel $E_{\epsilon}$ with $\epsilon \in (0,1)$, we always have $D(\unif E_{\epsilon}|| \Delta_{0} E_{\epsilon}) = +\infty$ for $\unif,\Delta_{0} = (1,0,\dots,0)$ $\in \P_q$. On the other hand, for any $q$-ary symmetric channel $W_{\delta}$ with $\delta \in (0,1)$, we have $D(P_X W_{\delta} || Q_X W_{\delta}) < +\infty$ for every $P_X,Q_X \in \P_q$. Thus, $W_{\delta} \not\succeq_{\textsf{\tiny ln}} E_{\epsilon}$ for any $\delta \in (0,1)$.
\end{proof}

In fact, the argument for Proposition \ref{Prop: Symmetric Channel not Less Noisy than Erasure Channel} conveys that a symmetric channel $W_{\delta} \in \R^{q \times q}_{\textsf{sto}}$ with $\delta \in (0,1)$ satisfies $W_{\delta} \succeq_{\textsf{\tiny ln}} V$ for some channel $V \in \R^{q \times r}_{\textsf{sto}}$ only if $D(P_X V || Q_X V) < +\infty$ for every $P_X,Q_X \in \P_q$. Typically, we are only interested in studying $q$-ary symmetric channels with $q \geq 2$ and $\delta \in \big(0,\frac{q-1}{q}\big)$. For example, the BSC with crossover probability $\delta$ is usually studied for $\delta \in \left(0,\frac{1}{2}\right)$. Indeed, the less noisy domination characteristics of the extremal $q$-ary symmetric channels with $\delta = 0$ or $\delta = \frac{q-1}{q}$ are quite elementary. Given $q \geq 2$, $W_{0} = I_q \in \R^{q \times q}_{\textsf{sto}}$ satisfies $W_{0} \succeq_{\textsf{\tiny ln}} V$, and $W_{(q-1)/q} = \1 \unif \in \R^{q \times q}_{\textsf{sto}}$ satisfies $V \succeq_{\textsf{\tiny ln}} W_{(q-1)/q}$, for every channel $V \in \R^{q \times r}_{\textsf{sto}}$ on a common input alphabet. For the sake of completeness, we also note that for $q \geq 2$, the extremal $q$-ary erasure channels $E_{0} \in \R^{q \times (q+1)}_{\textsf{sto}}$ and $E_{1} \in \R^{q \times (q+1)}_{\textsf{sto}}$, with $\epsilon = 0$ and $\epsilon = 1$ respectively, satisfy $E_{0} \succeq_{\textsf{\tiny ln}} V$ and $V \succeq_{\textsf{\tiny ln}} E_{1}$ for every channel $V \in \R^{q \times r}_{\textsf{sto}}$ on a common input alphabet.

The result that the $q$-ary symmetric channel with uniform noise pmf $W_{(q-1)/q}$ is more noisy than every channel on the same input alphabet has an analogue concerning additive white Gaussian noise (AWGN) channels. Consider all additive noise channels of the form:
\begin{equation}
Y = X + Z
\end{equation}
where $X,Y \in \R$, the input $X$ is uncorrelated with the additive noise $Z$: $\E\left[X Z\right] = 0$, and the noise $Z$ has power constraint $\E\left[Z^2\right] \leq \sigma_Z^2$ for some fixed $\sigma_Z > 0$. Let $X = X_{\textsf{g}} \sim \mathcal{N}(0,\sigma_X^2)$ (Gaussian distribution with mean $0$ and variance $\sigma_X^2$) for some $\sigma_X > 0$. Then, we have:
\begin{equation}
I\left(X_{\textsf{g}};X_{\textsf{g}} + Z\right) \geq I\left(X_{\textsf{g}};X_{\textsf{g}} + Z_{\textsf{g}}\right)
\end{equation} 
where $Z_{\textsf{g}} \sim \mathcal{N}(0,\sigma_Z^2)$, $Z_{\textsf{g}}$ is independent of $X_{\textsf{g}}$, and equality occurs if and only if $Z = Z_{\textsf{g}}$ in distribution \cite[Section 4.7]{InfoTheoryNotes}. This states that Gaussian noise is the ``worst case additive noise'' for a Gaussian source. Hence, the AWGN channel is \textit{not more capable} than any other additive noise channel with the same constraints. As a result, the AWGN channel is \textit{not less noisy} than any other additive noise channel with the same constraints (using Proposition \ref{Prop: Relations between Channel Preorders}).

\section{Equivalent characterizations of less noisy preorder}
\label{Equivalent Characterizations of Less Noisy Preorder}

Having studied the structure of less noisy domination and degradation regions of channels, we now consider the problem of verifying whether a channel $W$ is less noisy than another channel $V$. Since using Definition \ref{Def: Less Noisy Preorder} or Proposition \ref{Prop: KL Divergence Characterization of Less Noisy} directly is difficult, we often start by checking whether $V$ is a degraded version of $W$. When this fails, we typically resort to verifying van Dijk's condition in Proposition \ref{Prop: van Dijk Characterization of Less Noisy}, cf. \cite[Theorem 2]{vanDijk}. In this section, we prove the equivalent characterization of the less noisy preorder in Theorem \ref{Thm: Chi Squared Divergence Characterization of Less Noisy}, and then present some useful corollaries of van Dijk's condition.

\subsection{Characterization using $\chi^2$-divergence}
\label{Characterization using chi squared divergence}

Recall the general measure theoretic setup and the definition of $\chi^2$-divergence from subsection \ref{Chi Squared Divergence Characterization of the Less Noisy Preorder}. It is well-known that KL divergence is locally approximated by $\chi^2$-divergence, e.g. \cite[Section 4.2]{InfoTheoryNotes}. While this approximation sometimes fails globally,
cf. \cite{Anantharam2014Hypercontractivity}, the following notable result was first shown by Ahlswede and G\'{a}cs in the discrete case in \cite{AhlswedeGacsHypercontraction}, and then extended to general alphabets in \cite[Theorem 3]{GraphSDPI}:
\begin{equation}
\label{Eq: KL and Chi-Squared Contraction Coefficient}
\eta_{\textsf{\tiny KL}}\!\left(W\right) = \eta_{\chi^2}\!\left(W\right) \triangleq \sup_{\substack{P_X,Q_X \\ 0<\chi^2(P_X||Q_X)<+\infty}}{\frac{\chi^2\left(P_X W||Q_X W\right)}{\chi^2\left(P_X||Q_X\right)}} 
\end{equation}
for any Markov kernel $W : \H_1 \times \X \rightarrow [0,1]$, where $\eta_{\textsf{\tiny KL}}\!\left(W\right)$ is defined as in \eqref{Eq: KL Contraction Coefficient}, $\eta_{\chi^2}\!\left(W\right)$ is the contraction coefficient for $\chi^2$-divergence, and the suprema in $\eta_{\textsf{\tiny KL}}\!\left(W\right)$ and $\eta_{\chi^2}\!\left(W\right)$ are taken over all probability measures $P_X$ and $Q_X$ on $(\X,\F)$. Since $\eta_{\textsf{\tiny KL}}$ characterizes less noisy domination with respect to an erasure channel as mentioned in subsection \ref{Main Question and Motivation}, \eqref{Eq: KL and Chi-Squared Contraction Coefficient} portrays that $\eta_{\chi^2}$ also characterizes this. We will now prove Theorem \ref{Thm: Chi Squared Divergence Characterization of Less Noisy} from subsection \ref{Chi Squared Divergence Characterization of the Less Noisy Preorder}, which generalizes \eqref{Eq: KL and Chi-Squared Contraction Coefficient} and illustrates that $\chi^2$-divergence actually characterizes less noisy domination by an arbitrary channel.

\renewcommand{\proofname}{Proof of Theorem \ref{Thm: Chi Squared Divergence Characterization of Less Noisy}}

\begin{proof}
In order to prove the forward direction, we recall the local approximation of KL divergence using $\chi^2$-divergence from \cite[Proposition 4.2]{InfoTheoryNotes}, which states that for any two probability measures $P_X$ and $Q_X$ on $(\X,\F)$:
\begin{equation}
\label{Eq: Local Approximation}
\lim_{\lambda \rightarrow 0^{+}}{\frac{2}{\lambda^2} D\!\left(\lambda P_X + \bar{\lambda} Q_X||Q_X\right)} = \chi^2 \!\left(P_X||Q_X\right)
\end{equation}
where $\bar{\lambda} = 1 - \lambda$ for $\lambda \in (0,1)$, and both sides of \eqref{Eq: Local Approximation} are finite or infinite together. Then, we observe that for any two probability measures $P_X$ and $Q_X$, and any $\lambda \in [0,1]$, we have:
$$ D\!\left(\lambda P_X W \! + \! \bar{\lambda} Q_X W||Q_X W\right) \geq D\!\left(\lambda P_X V \! + \! \bar{\lambda} Q_X V||Q_X V\right) $$
since $W \succeq_{\textsf{\tiny ln}} V$. Scaling this inequality by $\frac{2}{\lambda^2}$ and letting $\lambda \rightarrow 0$ produces:
$$ \chi^2 \!\left(P_X W||Q_X W\right) \geq \chi^2 \!\left(P_X V||Q_X V\right) $$
as shown in \eqref{Eq: Local Approximation}. This proves the forward direction.

To establish the converse direction, we recall an integral representation of KL divergence using $\chi^2$-divergence presented in \cite[Appendix A.2]{GraphSDPI} (which can be distilled from the argument in \cite[Theorem 1]{ChoiRuskaiSeneta94}):\footnote{Note that \cite[Equation (78)]{GraphSDPI}, and hence \cite[Equation (7)]{SymmetricChannelDominationConf}, are missing factors of $\frac{1}{t+1}$ inside the integrals.}
\begin{equation}
\label{Eq: Integral Representation}
D\!\left(P_X || Q_X\right) = \int_{0}^{\infty}{\frac{\chi^2 \!\left(P_X||Q_X^t\right)}{t+1} \, dt} 
\end{equation}
for any two probability measures $P_X$ and $Q_X$ on $(\X,\F)$, where $Q_X^t = \frac{t}{1+t} P_X + \frac{1}{t+1} Q_X$ for $t \in [0,\infty)$, and both sides of \eqref{Eq: Integral Representation} are finite or infinite together (as a close inspection of the proof in \cite[Appendix A.2]{GraphSDPI} reveals). Hence, for every $P_X$ and $Q_X$, we have by assumption: 
$$ \chi^2 \!\left(P_X W||Q_X^t W\right) \geq \chi^2 \!\left(P_X V||Q_X^t V\right) $$
which implies that:
\begin{align*}
\int_{0}^{\infty}{\frac{\chi^2 \!\left(P_X W||Q_X^t W\right)}{t+1} \, dt} & \geq \int_{0}^{\infty}{\frac{\chi^2 \!\left(P_X V||Q_X^t V\right)}{t+1} \, dt} \\
\Rightarrow \enspace D\!\left(P_X W || Q_X W\right) & \geq D\!\left(P_X V || Q_X V\right) .
\end{align*}
Hence, $W \succeq_{\textsf{\tiny ln}} V$, which completes the proof.
\end{proof}

\renewcommand{\proofname}{Proof}

\subsection{Characterizations via the L\"{o}wner partial order and spectral radius}
\label{Characterizations via the Loewner Partial Order and Spectral Radius}

We will use the finite alphabet setup of subsection \ref{Channel Preorders in Information Theory} for the remaining discussion in this paper. In the finite alphabet setting, Theorem \ref{Thm: Chi Squared Divergence Characterization of Less Noisy} states that $W \in \R^{q \times r}_{\textsf{sto}}$ is less noisy than $V \in \R^{q \times s}_{\textsf{sto}}$ if and only if for every $P_X,Q_X \in \P_q$:
\begin{equation}
\label{Eq: Finite Chi Squared Condition}
\chi^2 \!\left(P_X W||Q_X W\right) \geq \chi^2 \!\left(P_X V||Q_X V\right) .
\end{equation}
This characterization has the flavor of a L\"{o}wner partial order condition. Indeed, it is straightforward to verify that for any $P_X \in \P_q$ and $Q_X \in \P_q^{\circ}$, we can write their $\chi^2$-divergence as:
\begin{equation}
\chi^2 \!\left(P_X||Q_X\right) = J_X \textsf{\small diag}\!\left(Q_X\right)^{-1} J_X^T . 
\end{equation}
where $J_X = P_X - Q_X$. Hence, we can express \eqref{Eq: Finite Chi Squared Condition} as:
\begin{equation}
J_X W \textsf{\small diag}\!\left(Q_X W\right)^{-1} W^T J_X^T \geq J_X V \textsf{\small diag}\!\left(Q_X V\right)^{-1} V^T J_X^T
\end{equation}
for every $J_X = P_X - Q_X$ such that $P_X \in \P_q$ and $Q_X \in \P_q^{\circ}$. This suggests that \eqref{Eq: Finite Chi Squared Condition} is equivalent to:
\begin{equation}
\label{Eq: Loewner Condition Intuition}
W \textsf{\small diag}\!\left(Q_X W\right)^{-1} W^T \succeq_{\textsf{\tiny PSD}} V \textsf{\small diag}\!\left(Q_X V\right)^{-1} V^T
\end{equation}
for every $Q_X \in \P_q^{\circ}$. It turns out that \eqref{Eq: Loewner Condition Intuition} indeed characterizes $\succeq_{\textsf{\tiny ln}}$, and this is straightforward to prove directly. The next proposition illustrates that \eqref{Eq: Loewner Condition Intuition} also follows as a corollary of van Dijk's characterization in Proposition \ref{Prop: van Dijk Characterization of Less Noisy}, and presents an equivalent spectral characterization of $\succeq_{\textsf{\tiny ln}}$. 

\begin{proposition}[L\"{o}wner and Spectral Characterizations of $\succeq_{\textsf{\tiny ln}}$]
\label{Prop: Loewner and Spectral Characterizations of Less Noisy}
For any pair of channels $W \in \R^{q \times r}_{\textsf{sto}}$ and $V \in \R^{q \times s}_{\textsf{sto}}$ on the same input alphabet $[q]$, the following are equivalent:
\begin{enumerate}
\item $W \succeq_{\textsf{\tiny ln}} V$.
\item For every $P_X \in \P_q^{\circ}$, we have:
$$ W \textsf{\small diag}\!\left(P_X W\right)^{-1} W^T \succeq_{\textsf{\tiny PSD}} V \textsf{\small diag}\!\left(P_X V\right)^{-1} V^T . $$
\item For every $P_X \in \P_q^{\circ}$, we have $\Range\big(V \textsf{\small diag}\!\left(P_X V\right)^{-1} V^T\big)$ $\subseteq \Range\big(W \textsf{\small diag}\!\left(P_X W\right)^{-1} W^T\big)$ and:\footnote{Note that \cite[Theorem 1 part 4]{SymmetricChannelDominationConf} neglected to mention the inclusion relation $\Range\big(V \textsf{\scriptsize diag}\!\left(P_X V\right)^{-1} V^T\big) \subseteq \Range\big(W \textsf{\scriptsize diag}\!\left(P_X W\right)^{-1} W^T\big)$.} 
$$ \rho\!\left(\!\left(W \textsf{\small diag}\!\left(P_X W\right)^{-1} W^T\right)^{\dagger} \! V \textsf{\small diag}\!\left(P_X V\right)^{-1} V^{T}\right) = 1 . $$
\end{enumerate}
\end{proposition}

\begin{proof} 
($1 \Leftrightarrow 2$) Recall the functional $F: \P_q \rightarrow \R, F(P_X) = I(P_X,W_{Y|X}) - I(P_X,V_{Y|X})$ defined in Proposition \ref{Prop: van Dijk Characterization of Less Noisy}, cf. \cite[Theorem 2]{vanDijk}. Since $F : \P_q \rightarrow \R$ is continuous on its domain $\P_q$, and twice differentiable on $\P_q^{\circ}$, $F$ is concave if and only if its Hessian is negative semidefinite for every $P_X \in \P_q^{\circ}$ (i.e. $-\nabla^2 F \left(P_X\right) \succeq_{\textsf{\tiny PSD}} 0$ for every $P_X \in \P_q^{\circ}$) \cite[Section 3.1.4]{ConvexOptimization}. The Hessian matrix of $F$, $\nabla^2 F : \P_q^{\circ} \rightarrow \R^{q \times q}_{\textsf{sym}}$, is defined entry-wise for every $x,x^{\prime} \in [q]$ as:
$$ \left[\nabla^2 F (P_X)\right]_{x,x^{\prime}} = \frac{\partial^2 F}{\partial P_X(x) \partial P_X(x^{\prime})}\left(P_X\right) $$
where we index the matrix $\nabla^2 F (P_X)$ starting at $0$ rather than $1$. Furthermore, a straightforward calculation shows that:
$$ \nabla^2 F (P_X) = V \textsf{\small diag}\!\left(P_X V\right)^{-1} V^T - W \textsf{\small diag}\!\left(P_X W\right)^{-1} W^T $$
for every $P_X \in \P_q^{\circ}$. (Note that the matrix inverses here are well-defined because $P_X \in \P_q^{\circ}$). Therefore, $F$ is concave if and only if for every $P_X \in \P_q^{\circ}$:
$$ W \textsf{\small diag}\!\left(P_X W\right)^{-1} W^T \succeq_{\textsf{\tiny PSD}} V \textsf{\small diag}\!\left(P_X V\right)^{-1} V^T . $$
This establishes the equivalence between parts 1 and 2 due to van Dijk's characterization of $\succeq_{\textsf{\tiny ln}}$ in Proposition \ref{Prop: van Dijk Characterization of Less Noisy}.

($2 \Leftrightarrow 3$) We now derive the spectral characterization of $\succeq_{\textsf{\tiny ln}}$ using part 2. Recall the well-known fact (see \cite[Theorem 1 parts (a),(f)]{PSDLoewnerOrder} and \cite[Theorem 7.7.3 (a)]{BasicMatrixAnalysis}):\\
\textit{Given positive semidefinite matrices $A,B \in \R^{q \times q}_{\succeq 0}$, $A \succeq_{\textsf{\tiny PSD}} B$ if and only if $\Range(B) \subseteq \Range(A)$ and $\rho\left(A^{\dagger}B\right) \leq 1$.}\\
Since $W \textsf{\small diag}\!\left(P_X W\right)^{-1} W^T$ and $V \textsf{\small diag}\!\left(P_X V\right)^{-1} V^T$ are positive semidefinite for every $P_X \in \P_q^{\circ}$, applying this fact shows that part 2 holds if and only if for every $P_X \in \P_q^{\circ}$, we have $\Range\big(V \textsf{\small diag}\!\left(P_X V\right)^{-1} V^T\big) \subseteq \Range\big(W \textsf{\small diag}\!\left(P_X W\right)^{-1} W^T\big)$ and: 
$$ \rho\left(\left(W \textsf{\small diag}\!\left(P_X W\right)^{-1} W^T\right)^{\dagger} V \textsf{\small diag}\!\left(P_X V\right)^{-1} V^{T}\right) \leq 1 . $$
To prove that this inequality is an equality, for any $P_X \in \P_q^{\circ}$, let $A = W \textsf{\small diag}\!\left(P_X W\right)^{-1} W^T$ and $B = V \textsf{\small diag}\!\left(P_X V\right)^{-1} \cdot$ $V^T$. It suffices to prove that: $\Range(B) \subseteq \Range(A)$ and $\rho\left(A^{\dagger}B\right) \leq 1$ if and only if $\Range(B) \subseteq \Range(A)$ and $\rho\left(A^{\dagger}B\right) = 1$. The converse direction is trivial, so we only establish the forward direction. Observe that $P_X A = \1^T$ and $P_X B = \1^T$. This implies that $\1^T A^{\dagger} B = P_X (A A^{\dagger}) B = P_X B = \1^T$, where $(A A^{\dagger}) B = B$ because $\Range(B) \subseteq \Range(A)$ and $A A^{\dagger}$ is the orthogonal projection matrix onto $\Range(A)$. Since $\rho\left(A^{\dagger}B\right) \leq 1$ and $A^{\dagger} B$ has an eigenvalue of $1$, we have $\rho\left(A^{\dagger}B\right) = 1$. Thus, we have proved that part 2 holds if and only if for every $P_X \in \P_q^{\circ}$, we have $\Range\big(V \textsf{\small diag}\!\left(P_X V\right)^{-1} V^T\big) \subseteq \Range\big(W \textsf{\small diag}\!\left(P_X W\right)^{-1} W^T\big)$ and:
$$ \rho\left(\left(W \textsf{\small diag}\!\left(P_X W\right)^{-1} W^T\right)^{\dagger} V \textsf{\small diag}\!\left(P_X V\right)^{-1} V^{T}\right) = 1 . $$ 
This completes the proof.
\end{proof}

The L\"{o}wner characterization of $\succeq_{\textsf{\tiny ln}}$ in part 2 of Proposition \ref{Prop: Loewner and Spectral Characterizations of Less Noisy} will be useful for proving some of our ensuing results. We remark that the equivalence between parts 1 and 2 can be derived by considering several other functionals. For instance, for any fixed pmf $Q_X \in \P_q^{\circ}$, we may consider the functional $F_2:\P_q \rightarrow \R$ defined by:
\begin{equation}
F_2(P_X) = D\!\left(P_X W || Q_X W\right) - D\!\left(P_X V || Q_X V\right)
\end{equation}
which has Hessian matrix, $\nabla^2 F_2 : \P_q^{\circ} \rightarrow \R^{q \times q}_{\textsf{sym}}$, $\nabla^2 F_2 (P_X) = W \textsf{\small diag}\!\left(P_X W\right)^{-1} W^T - V \textsf{\small diag}\!\left(P_X V\right)^{-1} V^T$, that does not depend on $Q_X$. Much like van Dijk's functional $F$, $F_2$ is \textit{convex} (for all $Q_X \in \P_q^{\circ}$) if and only if $W \succeq_{\textsf{\tiny ln}} V$. This is reminiscent of Ahlswede and G\'{a}cs' technique to prove \eqref{Eq: KL and Chi-Squared Contraction Coefficient}, where the convexity of a similar functional is established \cite{AhlswedeGacsHypercontraction}. 

As another example, for any fixed pmf $Q_X \in \P_q^{\circ}$, consider the functional $F_3:\P_q \rightarrow \R$ defined by:
\begin{equation}
F_3(P_X) = \chi^2 \!\left(P_X W || Q_X W\right) - \chi^2 \!\left(P_X V || Q_X V\right)
\end{equation}
which has Hessian matrix, $\nabla^2 F_3 : \P_q^{\circ} \rightarrow \R^{q \times q}_{\textsf{sym}}$, $\nabla^2 F_3 (P_X) = 2 \, W \textsf{\small diag}\!\left(Q_X W\right)^{-1} W^T - 2 \, V \textsf{\small diag}\!\left(Q_X V\right)^{-1} V^T$, that does not depend on $P_X$. Much like $F$ and $F_2$, $F_3$ is \textit{convex} for all $Q_X \in \P_q^{\circ}$ if and only if $W \succeq_{\textsf{\tiny ln}} V$.

Finally, we also mention some specializations of the spectral radius condition in part 3 of Proposition \ref{Prop: Loewner and Spectral Characterizations of Less Noisy}. If $q \geq r$ and $W$ has full column rank, the expression for spectral radius in the proposition statement can be simplified to:
\begin{equation}
\label{Eq: Simpler Spectral Condition}
\rho\left((W^{\dagger})^T \textsf{\small diag}\!\left(P_X W\right) W^{\dagger} V \textsf{\small diag}\!\left(P_X V\right)^{-1} V^{T}\right) = 1
\end{equation}
using basic properties of the Moore-Penrose pseudoinverse. Moreover, if $q = r$ and $W$ is non-singular, then the Moore-Penrose pseudoinverses in \eqref{Eq: Simpler Spectral Condition} can be written as inverses, and the inclusion relation between the ranges in part 3 of Proposition \ref{Prop: Loewner and Spectral Characterizations of Less Noisy} is trivially satisfied (and can be omitted from the proposition statement). We have used the spectral radius condition in this latter setting to (numerically) compute the additive less noisy domination region in Figure \ref{Figure: Additive Noise Channel Domination Regions}.

\section{Conditions for less noisy domination over additive noise channels}
\label{Conditions for Less Noisy Domination over Additive Noise Channels}

We now turn our attention to deriving several conditions for determining when $q$-ary symmetric channels are less noisy than other channels. Our interest in $q$-ary symmetric channels arises from their analytical tractability; Proposition \ref{Prop: Properties of Symmetric Channel Matrices} from subsection \ref{Symmetric Channels and their Properties}, Proposition \ref{Prop: Constants of Symmetric Channels} from section \ref{Less Noisy Domination and Logarithmic Sobolev Inequalities}, and \cite[Theorem 4.5.2]{Gallager} (which conveys that $q$-ary symmetric channels have uniform capacity achieving input distributions) serve as illustrations of this tractability. We focus on additive noise channels in this section, and on general channels in the next section.  

\subsection{Necessary conditions}
\label{Necessary Conditions}

We first present some straightforward necessary conditions for when an additive noise channel $W \in \R^{q \times q}_{\textsf{sto}}$ with $q \in \N$ is less noisy than another additive noise channel $V \in \R^{q \times q}_{\textsf{sto}}$ on an Abelian group $(\X,\oplus)$. These conditions can obviously be specialized for less noisy domination by symmetric channels. 

\begin{proposition}[Necessary Conditions for $\succeq_{\textsf{\tiny ln}}$ Domination over Additive Noise Channels]
\label{Prop: Necessary Conditions for Less Noisy Domination over Additive Noise Channels}
Suppose $W = \textsf{\small circ}_{\X}(w)$ and $V = \textsf{\small circ}_{\X}(v)$ are additive noise channels with noise pmfs $w,v \in \P_q$ such that $W \succeq_{\textsf{\tiny ln}} V$. Then, the following are true:
\begin{enumerate}
\item (Circle Condition) $\left\|w - \unif\right\|_{\ell^2} \geq \left\|v - \unif\right\|_{\ell^2}$. 
\item (Contraction Condition) $\eta_{\textsf{\tiny KL}}\!\left(W\right) \geq \eta_{\textsf{\tiny KL}}\!\left(V\right)$.
\item (Entropy Condition) $H\left(v\right) \geq H\left(w\right)$, where $H:\P_q \rightarrow \R^{+}$ is the Shannon entropy function. 
\end{enumerate}
\end{proposition}

\begin{proof} ~\newline
\textbf{Part 1:} Letting $P_X = \left(1,0,\dots,0\right)$ and $Q_X = \unif$ in the $\chi^2$-divergence characterization of $\succeq_{\textsf{\tiny ln}}$ in Theorem \ref{Thm: Chi Squared Divergence Characterization of Less Noisy} produces:
$$ q \left\|w - \unif\right\|_{\ell^2}^2 = \chi^2\left(w||\unif\right) \geq \chi^2\left(v||\unif\right) = q \left\|v - \unif\right\|_{\ell^2}^2 $$
since $\unif W = \unif V = \unif$, and $P_X W = w$ and $P_X V = v$ using \eqref{Eq: Group Circulant Matrix Row-Column Decomposition}. (This result can alternatively be proved using part 2 of Proposition \ref{Prop: Loewner and Spectral Characterizations of Less Noisy} and Fourier analysis.) \\
\textbf{Part 2:} This easily follows from Proposition \ref{Prop: KL Divergence Characterization of Less Noisy} and \eqref{Eq: KL Contraction Coefficient}. \\
\textbf{Part 3:} Letting $P_X = \left(1,0,\dots,0\right)$ and $Q_X = \unif$ in the KL divergence characterization of $\succeq_{\textsf{\tiny ln}}$ in Proposition \ref{Prop: KL Divergence Characterization of Less Noisy} produces:
$$ \log\left(q\right) - H\left(w\right) = D\left(w||\unif\right) \geq D\left(v||\unif\right) = \log\left(q\right) - H\left(v\right) $$
via the same reasoning as part 1. This completes the proof.
\end{proof}

We remark that the aforementioned necessary conditions have many generalizations. Firstly, if $W,V \in \R^{q \times q}_{\textsf{sto}}$ are doubly stochastic matrices, then the generalized circle condition holds:
\begin{equation}
\label{Eq: General Circle Condition}
\left\|W - W_{\frac{q-1}{q}}\right\|_{\textsf{Fro}} \geq \left\|V - W_{\frac{q-1}{q}}\right\|_{\textsf{Fro}}
\end{equation} 
where $W_{(q-1)/q} = \1 \unif$ is the $q$-ary symmetric channel whose conditional pmfs are all uniform, and $\left\|\cdot\right\|_{\textsf{Fro}}$ denotes the Frobenius norm. Indeed, letting $P_X = \Delta_{x} = \left(0,\dots,1,\dots,0\right)$ for $x \in [q]$, which has unity in the $(x+1)$th position, in the proof of part 1 and then adding the inequalities corresponding to every $x \in [q]$ produces \eqref{Eq: General Circle Condition}. Secondly, the contraction condition in Proposition \ref{Prop: Necessary Conditions for Less Noisy Domination over Additive Noise Channels} actually holds for any pair of general channels $W \in \R^{q \times r}_{\textsf{sto}}$ and $V \in \R^{q \times s}_{\textsf{sto}}$ on a common input alphabet (not necessarily additive noise channels). Moreover, we can start with Theorem \ref{Thm: Chi Squared Divergence Characterization of Less Noisy} and take the suprema of the ratios in $\chi^2\left(P_X W||Q_X W\right)/\chi^2\left(P_X||Q_X\right) \geq \chi^2\left(P_X V||Q_X V\right)/\chi^2\left(P_X||Q_X\right)$ over all $P_X$ ($\neq Q_X$) to get:
\begin{equation}
\rho_{\textsf{max}}\!\left(Q_X,W\right) \geq \rho_{\textsf{max}}\!\left(Q_X,V\right)
\end{equation}
for any $Q_X \in \P_q$, where $\rho_{\textsf{max}}\!\left(\cdot\right)$ denotes \textit{maximal correlation} which is defined later in part 3 of Proposition \ref{Prop: Constants of Symmetric Channels}, cf. \cite{RenyiCorrelation}, and we use \cite[Theorem 3]{BoundsbetweenContractionCoefficients} (or the results of \cite{ChiSquaredContractionMaximalCorrelation}). A similar result also holds for the contraction coefficient for KL divergence with fixed input pmf (see e.g. \cite[Definition 1]{BoundsbetweenContractionCoefficients} for a definition). 

\subsection{Sufficient conditions}
\label{Sufficient Conditions}

We next portray a sufficient condition for when an additive noise channel $V \in \R^{q \times q}_{\textsf{sto}}$ is a degraded version of a symmetric channel $W_{\delta} \in \R^{q \times q}_{\textsf{sto}}$. By Proposition \ref{Prop: Relations between Channel Preorders}, this is also a sufficient condition for $W_{\delta} \succeq_{\textsf{\tiny ln}} V$.  

\begin{proposition}[Degradation by Symmetric Channels]
\label{Prop: Degradation by Symmetric Channels}
Given an additive noise channel $V = \textsf{\small circ}_{\X}(v)$ with noise pmf $v \in \P_q$ and minimum probability $\tau = \min\{\left[V\right]_{i,j}: 1 \leq i,j \leq q\}$, we have:
$$ 0 \leq \delta \leq \left(q-1\right)\tau \enspace \Rightarrow \enspace W_{\delta} \succeq_{\textsf{\tiny deg}} V $$
where $W_{\delta} \in \R^{q \times q}_{\textsf{sto}}$ is a $q$-ary symmetric channel.
\end{proposition}

\begin{proof}
Using Corollary \ref{Cor: Degradation Region of Symmetric Channel}, it suffices to prove that the noise pmf $w_{(q-1)\tau} \succeq_{\textsf{\tiny maj}} v$. Since $0 \leq \tau \leq \frac{1}{q}$, we must have $0 \leq (q-1)\tau \leq \frac{q-1}{q}$. So, all entries of $w_{(q-1)\tau}$, except (possibly) the first, are equal to its minimum entry of $\tau$. As $v \geq \tau$ (entry-wise), $w_{(q-1)\tau} \succeq_{\textsf{\tiny maj}} v$ because the conditions of part 3 in Proposition \ref{Prop: Majorization} in Appendix \ref{App: Basics of Majorization Theory} are satisfied.
\end{proof}

It is compelling to find a sufficient condition for $W_{\delta} \succeq_{\textsf{\tiny ln}} V$ that does not simply ensure $W_{\delta} \succeq_{\textsf{\tiny deg}} V$ (such as Proposition \ref{Prop: Degradation by Symmetric Channels} and Theorem \ref{Thm: Sufficient Condition for Degradation by Symmetric Channels}). The ensuing proposition elucidates such a sufficient condition for additive noise channels. The general strategy for finding such a condition for additive noise channels is to identify a noise pmf that belongs to $\L_{W_{\delta}}^{\textsf{add}} \backslash \D_{W_{\delta}}^{\textsf{add}}$. One can then use Proposition \ref{Prop: Additive Less Noisy Domination and Degradation Regions} to explicitly construct a set of noise pmfs that is a subset of $\L_{W_{\delta}}^{\textsf{add}}$ but strictly includes $\D_{W_{\delta}}^{\textsf{add}}$. The proof of Proposition \ref{Prop: Less Noisy Domination by Symmetric Channels} finds such a noise pmf (that corresponds to a $q$-ary symmetric channel).

\begin{proposition}[Less Noisy Domination by Symmetric Channels]
\label{Prop: Less Noisy Domination by Symmetric Channels}
Given an additive noise channel $V = \textsf{\small circ}_{\X}(v)$ with noise pmf $v \in \P_q$ and $q \geq 2$, if for $\delta \in \big[0,\frac{q-1}{q}\big]$ we have:
$$ v \in \textsf{\small conv}\left(\left\{\w P_q^k : k \in [q]\right\} \cup \left\{w_{\gamma} P_q^k : k \in [q]\right\}\right) $$ 
then $W_{\delta} \succeq_{\textsf{\tiny ln}} V$, where $P_q \in \R^{q \times q}$ is defined in \eqref{Eq: Generator Cyclic Permutation Matrix}, and:
$$ \gamma = \frac{1-\delta}{1-\delta + \frac{\delta}{\left(q-1\right)^2}} \in \left[1-\frac{\delta}{q-1},1\right] . $$
\end{proposition}

\begin{proof}
Due to Proposition \ref{Prop: Additive Less Noisy Domination and Degradation Regions} and $\left\{w_{\gamma} P_x : x \in \X\right\} = \{w_{\gamma} P_q^k : k \in [q]\}$, it suffices to prove that $W_{\delta} \succeq_{\textsf{\tiny ln}} W_{\gamma}$. Since $\delta = 0 \Rightarrow \gamma = 1$ and $\delta = \frac{q-1}{q} \Rightarrow \gamma = \frac{q-1}{q}$, $W_{\delta} \succeq_{\textsf{\tiny ln}} W_{\gamma}$ is certainly true for $\delta \in \big\{0,\frac{q-1}{q}\big\}$. So, we assume that $\delta \in \big(0,\frac{q-1}{q}\big)$, which implies that:
$$ \gamma = \frac{1-\delta}{1-\delta + \frac{\delta}{\left(q-1\right)^2}} \in \left(\frac{q-1}{q},1\right) . $$ 
Since our goal is to show $W_{\delta} \succeq_{\textsf{\tiny ln}} W_{\gamma}$, we prove the equivalent condition in part 2 of Proposition \ref{Prop: Loewner and Spectral Characterizations of Less Noisy} that for every $P_X \in \P_q^{\circ}$:
\begin{align*}
& \, W_{\delta} \, \textsf{\small diag}\!\left(P_X W_{\delta}\right)^{-1} W_{\delta}^T \succeq_{\textsf{\tiny PSD}} W_{\gamma} \, \textsf{\small diag}\!\left(P_X W_{\gamma}\right)^{-1} W_{\gamma}^T \\
\Leftrightarrow & \, W_{\gamma}^{-1} \textsf{\small diag}\!\left(P_X W_{\gamma}\right) W_{\gamma}^{-1} \succeq_{\textsf{\tiny PSD}} W_{\delta}^{-1} \textsf{\small diag}\!\left(P_X W_{\delta}\right) W_{\delta}^{-1} \\
\Leftrightarrow & \, \textsf{\small diag}\!\left(P_X W_{\gamma}\right) \succeq_{\textsf{\tiny PSD}} W_{\gamma} W_{\delta}^{-1} \textsf{\small diag}\!\left(P_X W_{\delta}\right) W_{\delta}^{-1} W_{\gamma} \\
\Leftrightarrow & \, I_q \! \succeq_{\textsf{\tiny PSD}} \! \textsf{\small diag}\!\left(P_X W_{\gamma}\right)^{\!-\!\frac{1}{2}} \! W_{\tau} \textsf{\small diag}\!\left(P_X W_{\delta}\right) \! W_{\tau} \textsf{\small diag}\!\left(P_X W_{\gamma}\right)^{\!-\!\frac{1}{2}} \\
\Leftrightarrow & \, 1 \! \geq \! \left\| \textsf{\small diag}\!\left(P_X W_{\gamma}\right)^{\!-\!\frac{1}{2}} \! W_{\tau} \textsf{\small diag}\!\left(P_X W_{\delta}\right) \! W_{\tau} \textsf{\small diag}\!\left(P_X W_{\gamma}\right)^{\!-\!\frac{1}{2}} \! \right\|_{\textsf{op}} \\
\Leftrightarrow & \, 1 \geq \left\|\textsf{\small diag}\!\left(P_X W_{\gamma}\right)^{-\frac{1}{2}} W_{\tau} \textsf{\small diag}\!\left(P_X W_{\delta}\right)^{\frac{1}{2}} \right\|_{\textsf{op}}
\end{align*}
where the second equivalence holds because $W_{\delta}$ and $W_{\gamma}$ are symmetric and invertible (see part 4 of Proposition \ref{Prop: Properties of Symmetric Channel Matrices} and \cite[Corollary 7.7.4]{BasicMatrixAnalysis}), the third and fourth equivalences are non-singular $*$-congruences with $W_{\tau} = W_{\delta}^{-1} W_{\gamma} = W_{\gamma} W_{\delta}^{-1}$ and:
$$ \tau = \frac{\gamma-\delta}{1-\delta-\frac{\delta}{q-1}} > 0 $$
which can be computed as in the proof of Proposition \ref{Prop: Properties of Domination Factor Function} in Appendix \ref{App: Auxiliary Results}, and $\left\|\cdot\right\|_{\textsf{op}}$ denotes the spectral or operator norm.\footnote{Note that we cannot use the strict L\"{o}wner partial order $\succ_{\textsf{\tiny PSD}}$ (for $A,B \in \R^{q \times q}_{\textsf{sym}}$, $A \succ_{\textsf{\tiny PSD}} B$ if and only if $A - B$ is positive definite) for these equivalences as $\1^T W_{\gamma}^{-1} \textsf{\scriptsize diag}\!\left(P_X W_{\gamma}\right) W_{\gamma}^{-1} \1 = \1^T W_{\delta}^{-1} \textsf{\scriptsize diag}\!\left(P_X W_{\delta}\right) W_{\delta}^{-1} \1$.}

It is instructive to note that if $W_{\tau} \in \R^{q \times q}_{\textsf{sto}}$, then the \textit{divergence transition matrix} $\textsf{\small diag}\!\left(P_X W_{\gamma}\right)^{-\frac{1}{2}} W_{\tau} \textsf{\small diag}\!\left(P_X W_{\delta}\right)^{\frac{1}{2}}$ has right singular vector $\sqrt{P_X W_{\delta}}^T$ and left singular vector $\sqrt{P_X W_{\gamma}}^T$ corresponding to its maximum singular value of unity (where the square roots are applied entry-wise)\textemdash{}see \cite{BoundsbetweenContractionCoefficients} and the references therein. So, $W_{\tau} \in \R^{q \times q}_{\textsf{sto}}$ is a sufficient condition for $W_{\delta} \succeq_{\textsf{\tiny ln}} W_{\gamma}$. Since $W_{\tau} \in \R^{q \times q}_{\textsf{sto}}$ if and only if $0 \leq \tau \leq 1$ if and only if $\delta \leq \gamma \leq 1-\frac{\delta}{q-1}$, the latter condition also implies that $W_{\delta} \succeq_{\textsf{\tiny ln}} W_{\gamma}$. However, we recall from \eqref{Eq: Extremal Degraded Channel} in subsection \ref{Less Noisy Domination and Degradation Regions for Symmetric Channels} that $W_{\delta} \succeq_{\textsf{\tiny deg}} W_{\gamma}$ for $\delta \leq \gamma \leq 1-\frac{\delta}{q-1}$, while we seek some $1-\frac{\delta}{q-1} < \gamma \leq 1$ for which $W_{\delta} \succeq_{\textsf{\tiny ln}} W_{\gamma}$. When $q = 2$, we only have:
$$ \gamma = \frac{1-\delta}{1-\delta + \frac{\delta}{\left(q-1\right)^2}} = 1 - \frac{\delta}{q-1} = 1 - \delta $$
which implies that $W_{\delta} \succeq_{\textsf{\tiny deg}} W_{\gamma}$ is true for $q = 2$. On the other hand, when $q \geq 3$, it is straightforward to verify that:
$$ \gamma = \frac{1-\delta}{1-\delta + \frac{\delta}{\left(q-1\right)^2}} \in \left(1 - \frac{\delta}{q-1} , 1\right) $$ 
since $\delta \in \big(0,\frac{q-1}{q}\big)$.

From the preceding discussion, it suffices to prove for $q \geq 3$ that for every $P_X \in \P_q^{\circ}$:
$$ \left\| \textsf{\small diag}\!\left(P_X W_{\gamma}\right)^{\!-\!\frac{1}{2}} \! W_{\tau} \textsf{\small diag}\!\left(P_X W_{\delta}\right) \! W_{\tau} \textsf{\small diag}\!\left(P_X W_{\gamma}\right)^{\!-\!\frac{1}{2}} \right\|_{\textsf{op}} \leq 1 . $$
Since $\tau > 0$, and $0 \leq \tau \leq 1$ does not produce $\gamma > 1 - \frac{\delta}{q-1}$, we require that $\tau > 1$ ($\Leftrightarrow \gamma > 1 - \frac{\delta}{q-1}$) so that $W_{\tau}$ has strictly negative entries along the diagonal. Notice that:
$$ \forall x \in [q], \enspace \textsf{\small diag}\!\left(\Delta_x W_{\gamma}\right) \succeq_{\textsf{\tiny PSD}} W_{\gamma} W_{\delta}^{-1} \textsf{\small diag}\!\left(\Delta_x W_{\delta}\right) W_{\delta}^{-1} W_{\gamma} $$
where $\Delta_x = \left(0,\dots,1,\dots,0\right) \in \P_q$ denotes the Kronecker delta pmf with unity at the $(x+1)$th position, implies that:  
$$ \textsf{\small diag}\!\left(P_X W_{\gamma}\right) \succeq_{\textsf{\tiny PSD}} W_{\gamma} W_{\delta}^{-1} \textsf{\small diag}\!\left(P_X W_{\delta}\right) W_{\delta}^{-1} W_{\gamma} $$
for every $P_X \in \P_q^{\circ}$, because convex combinations preserve the L\"{o}wner relation. So, it suffices to prove that for every $x \in [q]$:
$$ \left\| \textsf{\small diag}\!\left(w_{\gamma} P_q^{x}\right)^{-\frac{1}{2}} \! W_{\tau} \textsf{\small diag}\!\left(\w P_q^{x}\right) \! W_{\tau} \textsf{\small diag}\!\left(w_{\gamma}P_q^{x}\right)^{-\frac{1}{2}} \right\|_{\textsf{op}} \leq 1 $$
where $P_q \in \R^{q \times q}$ is defined in \eqref{Eq: Generator Cyclic Permutation Matrix}, because $\Delta_x M$ extracts the $(x+1)$th row of a matrix $M \in \R^{q \times q}$. Let us define $A_x \triangleq \textsf{\small diag}\!\left(w_{\gamma} P_q^{x}\right)^{-\frac{1}{2}} \! W_{\tau} \textsf{\small diag}\!\left(\w P_q^{x}\right) \! W_{\tau} \textsf{\small diag}\!\left(w_{\gamma}P_q^{x}\right)^{-\frac{1}{2}}$ for each $x \in [q]$. Observe that for every $x \in [q]$, $A_x \in \R^{q \times q}_{\succeq 0}$ is orthogonally diagonalizable by the real spectral theorem \cite[Theorem 7.13]{LinearAlgebraAxler}, and has a strictly positive eigenvector $\sqrt{w_{\gamma}P_q^{x}}$ corresponding to the eigenvalue of unity:
$$ \forall x \in [q], \enspace \sqrt{w_{\gamma}P_q^{x}} A_x = \sqrt{w_{\gamma}P_q^{x}} $$
so that all other eigenvectors of $A_x$ have some strictly negative entries since they are orthogonal to $\sqrt{w_{\gamma}P_q^{x}}$. Suppose $A_x$ is entry-wise non-negative for every $x \in [q]$. Then, the largest eigenvalue (known as the Perron-Frobenius eigenvalue) and the spectral radius of each $A_x$ is unity by the Perron-Frobenius theorem \cite[Theorem 8.3.4]{BasicMatrixAnalysis}, which proves that $\left\| A_x \right\|_{\textsf{op}} \leq 1$ for every $x \in [q]$. Therefore, it is sufficient to prove that $A_x$ is entry-wise non-negative for every $x \in [q]$. Equivalently, we can prove that $W_{\tau} \textsf{\small diag}\!\left(\w P_q^{x}\right) W_{\tau}$ is entry-wise non-negative for every $x \in [q]$, since $\textsf{\small diag}\!\left(w_{\gamma}P_q^{x}\right)^{-\frac{1}{2}}$ scales the rows or columns of the matrix it is pre- or post-multiplied with using strictly positive scalars. 

We now show the equivalent condition below that the minimum possible entry of $W_{\tau} \textsf{\small diag}\!\left(\w P_q^{x}\right) W_{\tau}$ is non-negative:
\begin{align}
0 & \leq \min_{\substack{x \in [q]\\1 \leq i,j \leq q}}{\underbrace{\sum_{r = 1}^{q}{\left[W_{\tau}\right]_{i,r} \left[W_{\delta}\right]_{x+1,r} \left[W_{\tau}\right]_{r,j}}}_{= \enspace \left[W_{\tau} \textsf{\scriptsize diag}\left(\w P_q^{x}\right) W_{\tau}\right]_{i,j}}} \nonumber \\
& = \frac{\tau (1 - \delta)(1 - \tau)}{q-1} + \frac{\delta \tau (1 - \tau)}{(q-1)^2} + (q-2)\frac{\delta \tau^2}{(q-1)^3} .
\label{Eq: Minimum Possible Entry}
\end{align}
The above equality holds because for $i \neq j$:
$$ \frac{\delta}{q-1}\sum_{r = 1}^{q}{\underbrace{\left[W_{\tau}\right]_{i,r} \left[W_{\tau}\right]_{r,i}}_{= \, \left[W_{\tau}\right]_{i,r}^2 \, \geq \, 0}} \geq \frac{\delta}{q-1}\sum_{r = 1}^{q}{\left[W_{\tau}\right]_{i,r} \left[W_{\tau}\right]_{r,j}} $$
is clearly true (using, for example, the rearrangement inequality in \cite[Section 10.2]{HardyLittlewoodPolya}), and adding $\big(1 - \delta - \frac{\delta}{q-1}\big)\left[W_{\tau}\right]_{i,k}^2 \geq 0$ (regardless of the value of $1 \leq k \leq q$) to the left summation increases its value, while adding $\big(1 - \delta - \frac{\delta}{q-1}\big)\left[W_{\tau}\right]_{i,p} \cdot$ $\left[W_{\tau}\right]_{p,j} < 0$ (which exists for an appropriate value $1 \leq p \leq q$ as $\tau > 1$) to the right summation decreases its value. As a result, the minimum possible entry of $W_{\tau} \textsf{\small diag}\!\left(\w P_q^{x}\right) W_{\tau}$ can be achieved with $x+1 = i \neq j$ or $i \neq j = x+1$. We next substitute $\tau = (\gamma-\delta)/\big(1-\delta-\frac{\delta}{q-1}\big)$ into \eqref{Eq: Minimum Possible Entry} and simplify the resulting expression to get:
\begin{align*}
0 & \leq \left(\gamma-\delta\right) \times \\
& \quad \left(\!\!\left(\!1-\frac{\delta}{q-1}- \gamma\!\right)\!\!\left(\!1-\delta+\frac{\delta}{q-1}\!\right)\! + \!\frac{\left(q-2\right) \!\delta\! \left(\gamma-\delta\right)}{\left(q-1\right)^2}\!\right) .
\end{align*} 
The right hand side of this inequality is quadratic in $\gamma$ with roots $\gamma = \delta$ and $\gamma = \frac{1-\delta}{1-\delta + (\delta/(q-1)^2)}$. Since the coefficient of $\gamma^2$ in this quadratic is strictly negative:
$$ \underbrace{\frac{\left(q-2\right) \delta}{\left(q-1\right)^2} - \left(1-\delta+\frac{\delta}{q-1}\right)}_{\text{coefficient of} \, \gamma^2} < 0 \Leftrightarrow 1 - \delta  + \frac{\delta}{\left(q-1\right)^2} > 0 $$
the minimum possible entry of $W_{\tau} \textsf{\small diag}\!\left(\w P_q^{x}\right) W_{\tau}$ is non-negative if and only if:
$$ \delta \leq \gamma \leq \frac{1-\delta}{1-\delta + \frac{\delta}{\left(q-1\right)^2}} $$
where we use the fact that $\frac{1-\delta}{1-\delta + (\delta/(q-1)^2)} \geq 1 - \frac{\delta}{q-1} \geq \delta$. Therefore, $\gamma = \frac{1-\delta}{1-\delta + (\delta/(q-1)^2)}$ produces $W_{\delta} \succeq_{\textsf{\tiny ln}} W_{\gamma}$, which completes the proof.
\end{proof}

Heretofore we have derived results concerning less noisy domination and degradation regions in section \ref{Less Noisy Domination and Degradation Regions}, and proven several necessary and sufficient conditions for less noisy domination of additive noise channels by symmetric channels in this section. We finally have all the pieces in place to establish Theorem \ref{Thm: Additive Less Noisy Domination and Degradation Regions for Symmetric Channels} from section \ref{Main Results}. In closing this section, we indicate the pertinent results that coalesce to justify it.

\renewcommand{\proofname}{Proof of Theorem \ref{Thm: Additive Less Noisy Domination and Degradation Regions for Symmetric Channels}}

\begin{proof}
The first equality follows from Corollary \ref{Cor: Degradation Region of Symmetric Channel}. The first set inclusion is obvious, and its strictness follows from the proof of Proposition \ref{Prop: Less Noisy Domination by Symmetric Channels}. The second set inclusion follows from Proposition \ref{Prop: Less Noisy Domination by Symmetric Channels}. The third set inclusion follows from the circle condition (part 1) in Proposition \ref{Prop: Necessary Conditions for Less Noisy Domination over Additive Noise Channels}. Lastly, the properties of $\L_{W_{\delta}}^{\textsf{add}}$ are derived in Proposition \ref{Prop: Additive Less Noisy Domination and Degradation Regions}.
\end{proof}

\renewcommand{\proofname}{Proof}

\section{Sufficient conditions for degradation over general channels}
\label{Sufficient Conditions for Degradation over General Channels}

While Propositions \ref{Prop: Degradation by Symmetric Channels} and \ref{Prop: Less Noisy Domination by Symmetric Channels} present sufficient conditions for a symmetric channel $W_{\delta} \in \R^{q \times q}_{\textsf{sto}}$ to be less noisy than an additive noise channel, our more comprehensive objective is to find the maximum $\delta \in \big[0,\frac{q-1}{q}\big]$ such that $W_{\delta} \succeq_{\textsf{\tiny ln}} V$ for any given general channel $V \in \R^{q \times r}_{\textsf{sto}}$ on a common input alphabet. We may formally define this maximum $\delta$ (that characterizes the extremal symmetric channel that is less noisy than $V$) as:
\begin{equation}
\delta^{\star}\left(V\right) \triangleq \sup\left\{\delta \in \left[0,\frac{q-1}{q}\right] : W_{\delta} \succeq_{\textsf{\tiny ln}} V\right\}
\end{equation}
and for every $0 \leq \delta < \delta^{\star}\left(V\right)$, $W_{\delta} \succeq_{\textsf{\tiny ln}} V$. Alternatively, we can define a non-negative (less noisy) \textit{domination factor} function for any channel $V \in \R^{q \times r}_{\textsf{sto}}$:
\begin{equation}
\label{Eq: Domination Factor}
\mu_V\left(\delta\right) \triangleq \sup_{\substack{P_X, Q_X \in \P_q :\\0 < D\left(P_X W_{\delta}||Q_X W_{\delta}\right) < +\infty}}{\frac{D\left(P_X V||Q_X V\right)}{D\left(P_X W_{\delta}||Q_X W_{\delta}\right)}}
\end{equation}
with $\delta \in \big[0,\frac{q-1}{q}\big)$, which is analogous to the contraction coefficient for KL divergence since $\mu_V\left(0\right) \triangleq \eta_{\textsf{\tiny KL}}\!\left(V\right)$. Indeed, we may perceive $P_X W_{\delta}$ and $Q_X W_{\delta}$ in the denominator of \eqref{Eq: Domination Factor} as pmfs inside the ``shrunk'' simplex $\textsf{\small conv}(\{\w P_q^k :$ $k \in [q]\})$, and \eqref{Eq: Domination Factor} represents a contraction coefficient of $V$ where the supremum is taken over this ``shrunk'' simplex.\footnote{Pictorially, the ``shrunk'' simplex is the magenta triangle in Figure \ref{Figure: Additive Noise Channel Domination Regions} while the simplex itself is the larger gray triangle.} For simplicity, consider a channel $V \in \R^{q \times r}_{\textsf{sto}}$ that is strictly positive entry-wise, and has domination factor function $\mu_V: \big(0,\frac{q-1}{q}\big) \rightarrow \R^{+}$, where the domain excludes zero because $\mu_V$ is only interesting for $\delta \in \big(0,\frac{q-1}{q}\big)$, and this exclusion also affords us some analytical simplicity. It is shown in Proposition \ref{Prop: Properties of Domination Factor Function} of Appendix \ref{App: Auxiliary Results} that $\mu_V$ is always finite on $\big(0,\frac{q-1}{q}\big)$, continuous, convex, strictly increasing, and has a vertical asymptote at $\delta = \frac{q-1}{q}$. Since for every $P_X, Q_X \in \P_q$:
\begin{equation}
\mu_V\left(\delta\right) \, D\left(P_X W_{\delta}||Q_X W_{\delta}\right) \geq D\left(P_X V||Q_X V\right)
\end{equation} 
we have $\mu_V\left(\delta\right) \leq 1$ if and only if $W_{\delta} \succeq_{\textsf{\tiny ln}} V$. Hence, using the strictly increasing property of $\mu_V:\big(0,\frac{q-1}{q}\big) \rightarrow \R^{+}$, we can also characterize $\delta^{\star}\left(V\right)$ as:
\begin{equation}
\delta^{\star}\left(V\right) = \mu_V^{-1}\left(1\right)
\end{equation}
where $\mu_V^{-1}$ denotes the inverse function of $\mu_V$, and unity is in the range of $\mu_V$ by Theorem \ref{Thm: Sufficient Condition for Degradation by Symmetric Channels} since $V$ is strictly positive entry-wise.

We next briefly delineate how one might computationally approximate $\delta^{\star}\left(V\right)$ for a given general channel $V \in \R^{q \times r}_{\textsf{sto}}$. From part 2 of Proposition \ref{Prop: Loewner and Spectral Characterizations of Less Noisy}, it is straightforward to obtain the following minimax characterization of $\delta^{\star}\left(V\right)$:  
\begin{equation}
\label{Eq:Minimax Formulation}
\delta^{\star}\left(V\right) = \inf_{P_X \in \P_q^{\circ}}\,{\sup_{\delta \in \mathcal{S}(P_X)}{\delta}}
\end{equation}
where $\mathcal{S}(P_X) = \big\{\delta \! \in \! \big[0,\frac{q-1}{q}\big] \! : W_{\delta} \, \textsf{\small diag}\!\left(P_X W_{\delta}\right)^{\!-1} \! W_{\delta}^T \succeq_{\textsf{\tiny PSD}} V \textsf{\small diag}\!\left(P_X V\right)^{\!-1} \! V^T\big\}$. The infimum in \eqref{Eq:Minimax Formulation} can be na\"{i}vely approximated by sampling several $P_X \in \P_q^{\circ}$. The supremum in \eqref{Eq:Minimax Formulation} can be estimated by verifying collections of rational (ratio of polynomials) inequalities in $\delta$. This is because the positive semidefiniteness of a matrix is equivalent to the non-negativity of all its principal minors by \textit{Sylvester's criterion} \cite[Theorem 7.2.5]{BasicMatrixAnalysis}. Unfortunately, this procedure appears to be rather cumbersome.

Since analytically computing $\delta^{\star}\left(V\right)$ also seems intractable, we now prove Theorem \ref{Thm: Sufficient Condition for Degradation by Symmetric Channels} from section \ref{Main Results}. Theorem \ref{Thm: Sufficient Condition for Degradation by Symmetric Channels} provides a sufficient condition for $W_{\delta} \succeq_{\textsf{\tiny deg}} V$ (which implies $W_{\delta} \succeq_{\textsf{\tiny ln}} V$ using Proposition \ref{Prop: Relations between Channel Preorders}) by restricting its attention to the case where $V \in \R^{q \times q}_{\textsf{sto}}$ with $q \geq 2$. Moreover, it can be construed as a lower bound on $\delta^{\star}\left(V\right)$:
\begin{equation}
\delta^{\star}\left(V\right) \geq \frac{\nu}{1-(q-1)\nu + \frac{\nu}{q-1}} 
\end{equation}
where $\nu = \min\left\{[V]_{i,j}:1 \leq i,j \leq q\right\}$ is the minimum conditional probability in $V$.

\renewcommand{\proofname}{Proof of Theorem \ref{Thm: Sufficient Condition for Degradation by Symmetric Channels}}

\begin{proof}
Let the channel $V \in \R^{q \times q}_{\textsf{sto}}$ have the conditional pmfs $v_1,\dots,v_q \in \P_q$ as its rows:  
$$ V = \left[v_1^T \enspace v_2^T \, \cdots \enspace v_q^T \right]^T . $$
From the proof of Proposition \ref{Prop: Degradation by Symmetric Channels}, we know that $w_{(q-1)\nu} \succeq_{\textsf{\tiny maj}} v_i$ for every $i \in \left\{1,\dots,q\right\}$. Using part 1 of Proposition \ref{Prop: Majorization} in Appendix \ref{App: Basics of Majorization Theory} (and the fact that the set of all permutations of $w_{(q-1)\nu}$ is exactly the set of all cyclic permutations of $w_{(q-1)\nu}$), we can write this as:
$$ \forall i \in \left\{1,\dots,q\right\}, \enspace v_i = \sum_{j = 1}^{q}{p_{i,j} \, w_{(q-1)\nu} P_q^{j-1}} $$
where the matrix $P_q \in \R^{q \times q}$ is given in \eqref{Eq: Generator Cyclic Permutation Matrix}, and $\left\{p_{i,j} \geq 0 : 1 \leq i,j \leq q\right\}$ are the convex weights such that $\sum_{j = 1}^{q}{p_{i,j}} = 1$ for every $i \in \left\{1,\dots,q\right\}$. Defining $P \in \R^{q \times q}_{\textsf{sto}}$ entry-wise as $\left[P\right]_{i,j} = p_{i,j}$ for every $1 \leq i,j \leq q$, we can also write this equation as $V = P W_{(q-1)\nu}$.\footnote{Matrices of the form $V = P W_{(q-1)\nu}$ with $P \in \R^{q \times q}_{\textsf{sto}}$ are \textit{not} necessarily degraded versions of $W_{(q-1)\nu}$: $W_{(q-1)\nu} \not\succeq_{\textsf{\tiny deg}} V$ (although we certainly have input-output degradation: $W_{(q-1)\nu} \succeq_{\textsf{\tiny iod}} V$). As a counterexample, consider $W_{1/2}$ for $q = 3$, and $P = \left[1 \, 0 \, 0; \, 1 \, 0 \, 0; \, 0 \, 1 \, 0\right]$, where the semicolons separate the rows of the matrix. If $W_{1/2} \succeq_{\textsf{\tiny deg}} P W_{1/2}$, then there exists $A \in \R^{3 \times 3}_{\textsf{sto}}$ such that $P W_{1/2} = W_{1/2} A$. However, $A = W_{1/2}^{-1} P W_{1/2} = (1/4)\left[3 \, 0 \, 1; \, 3 \, 0 \, 1; \, -1 \, 4 \, 1 \right]$ has a strictly negative entry, which leads to a contradiction.} Observe that:
$$ P = \sum_{1 \leq j_1,\dots,j_q \leq q}{\left(\prod_{i = 1}^{q}{p_{i,j_i}}\right) E_{j_1,\dots,j_q}} $$
where $\left\{\prod_{i = 1}^{q}{p_{i,j_i}} : 1 \leq j_1,\dots,j_q \leq q \right\}$ form a product pmf of convex weights, and for every $1 \leq j_1,\dots,j_q \leq q$:
$$ E_{j_1,\dots,j_q} \triangleq \left[e_{j_1} \enspace e_{j_2} \, \cdots \enspace e_{j_q} \right]^T $$
where $e_i \in \R^q$ is the $i$th standard basis (column) vector that has unity at the $i$th entry and zero elsewhere. Hence, we get:
$$ V = \sum_{1 \leq j_1,\dots,j_q \leq q}{\left(\prod_{i = 1}^{q}{p_{i,j_i}}\right) E_{j_1,\dots,j_q} W_{(q-1)\nu}} . $$
Suppose there exists $\delta \in \big[0,\frac{q-1}{q}\big]$ such that for all $j_1,\dots,j_q \in \left\{1,\dots,q\right\}$:
$$ \exists M_{j_1,\dots,j_q} \in \R^{q \times q}_{\textsf{sto}}, \enspace E_{j_1,\dots,j_q} W_{(q-1)\nu} = W_{\delta} M_{j_1,\dots,j_q} $$ 
i.e. $W_{\delta} \succeq_{\textsf{\tiny deg}} E_{j_1,\dots,j_q} W_{(q-1)\nu}$. Then, we would have:
$$ V = W_{\delta} \underbrace{\sum_{1 \leq j_1,\dots,j_q \leq q}{\left(\prod_{i = 1}^{q}{p_{i,j_i}}\right) M_{j_1,\dots,j_q}}}_{\text{stochastic matrix}} $$
which implies that $W_{\delta} \succeq_{\textsf{\tiny deg}} V$.

We will demonstrate that for every $j_1,\dots,j_q \in \left\{1,\dots,q\right\}$, there exists $M_{j_1,\dots,j_q} \in \R^{q \times q}_{\textsf{sto}}$ such that $E_{j_1,\dots,j_q} W_{(q-1)\nu}$ $= W_{\delta} M_{j_1,\dots,j_q}$ when $0 \leq \delta \leq \nu/\big(1-(q-1)\nu + \frac{\nu}{q-1}\big)$. Since $0 \leq \nu \leq \frac{1}{q}$, the preceding inequality implies that $0 \leq \delta \leq \frac{q-1}{q}$, where $\delta = \frac{q-1}{q}$ is possible if and only if $\nu = \frac{1}{q}$. When $\nu = \frac{1}{q}$, $V = W_{(q-1)/q}$ is the channel with all uniform conditional pmfs, and $W_{(q-1)/q} \succeq_{\textsf{\tiny deg}} V$ clearly holds. Hence, we assume that $0 \leq \nu < \frac{1}{q}$ so that $0 \leq \delta < \frac{q-1}{q}$, and establish the equivalent condition that for every $j_1,\dots,j_q \in \left\{1,\dots,q\right\}$:
$$ M_{j_1,\dots,j_q} = W_{\delta}^{-1} E_{j_1,\dots,j_q} W_{(q-1)\nu} $$
is a valid stochastic matrix. Recall that $W_{\delta}^{-1} = W_{\tau}$ with $\tau = \frac{-\delta}{1-\delta-(\delta/(q-1))}$ using part 4 of Proposition \ref{Prop: Properties of Symmetric Channel Matrices}. Clearly, all the rows of each $M_{j_1,\dots,j_q}$ sum to unity. So, it remains to verify that each $M_{j_1,\dots,j_q}$ has non-negative entries. For any $j_1,\dots,j_q \in \left\{1,\dots,q\right\}$ and any $i,j \in \left\{1,\dots,q\right\}$:
$$ \left[M_{j_1,\dots,j_q}\right]_{i,j} \geq \nu\left(1 - \tau\right) + \tau\left(1-\left(q-1\right)\nu\right) $$
where the right hand side is the minimum possible entry of any $M_{j_1,\dots,j_q}$ (with equality when $j_1 > 1$ and $j_2 = j_3 = \cdots = j_q = 1$ for example) as $\tau < 0$ and $1-\left(q-1\right)\nu > \nu$. To ensure each $M_{j_1,\dots,j_q}$ is entry-wise non-negative, the minimum possible entry must satisfy:
\begin{align*}
\nu\left(1 - \tau\right) + \tau\left(1-\left(q-1\right)\nu\right) & \geq 0 \\
\Leftrightarrow \enspace \nu + \frac{\delta \nu}{1-\delta-\frac{\delta}{q-1}} - \frac{\delta\left(1-\left(q-1\right)\nu\right)}{1-\delta-\frac{\delta}{q-1}} & \geq 0 
\end{align*}
and the latter inequality is equivalent to:
$$ \delta \leq \frac{\nu}{1 - \left(q-1\right)\nu + \frac{\nu}{q-1}} . $$
This completes the proof.
\end{proof}

\renewcommand{\proofname}{Proof}

We remark that if $V = E_{2,1,\dots,1} W_{(q-1)\nu} \in \R^{q \times q}_{\textsf{sto}}$, then this proof illustrates that $W_{\delta} \succeq_{\textsf{\tiny deg}} V$ if and only if $0 \leq \delta \leq \nu/\big(1-(q-1)\nu + \frac{\nu}{q-1}\big)$. Hence, the condition in Theorem \ref{Thm: Sufficient Condition for Degradation by Symmetric Channels} is tight when no further information about $V$ is known. It is worth juxtaposing Theorem \ref{Thm: Sufficient Condition for Degradation by Symmetric Channels} and Proposition \ref{Prop: Degradation by Symmetric Channels}. The upper bounds on $\delta$ from these results satisfy:
\begin{equation} 
\label{Eq: Compare Upper Bounds}
\underbrace{\frac{\nu}{1-(q-1)\nu + \frac{\nu}{q-1}}}_{\text{upper bound in Theorem \ref{Thm: Sufficient Condition for Degradation by Symmetric Channels}}} \leq \underbrace{\left(q-1\right)\nu}_{\substack{\text{upper bound in}\\ \text{Proposition \ref{Prop: Degradation by Symmetric Channels}}}} 
\end{equation}
where we have equality if and only if $\nu = \frac{1}{q}$, and it is straightforward to verify that \eqref{Eq: Compare Upper Bounds} is equivalent to $\nu \leq \frac{1}{q}$. Moreover, assuming that $q$ is large and $\nu = o\left(1/q\right)$, the upper bound in Theorem \ref{Thm: Sufficient Condition for Degradation by Symmetric Channels} is $\nu/\!\left(1 + o\left(1\right) + o\left(1/q^2\right)\right) = \Theta\left(\nu\right)$, while the upper bound in Proposition \ref{Prop: Degradation by Symmetric Channels} is $\Theta\left(q\nu\right)$.\footnote{We use the \textit{Bachmann-Landau asymptotic notation} here. Consider the (strictly) positive functions $f:\N \rightarrow \R$ and $g:\N \rightarrow \R$. The little-$o$ notation is defined as: $f(q) = o\left(g(q)\right) \Leftrightarrow \lim_{q \rightarrow \infty}{f(q)/g(q)} = 0$. The big-$O$ notation is defined as: $f(q) = O\left(g(q)\right) \Leftrightarrow \limsup_{q \rightarrow \infty}{\left|f(q)/g(q)\right|} < +\infty$. Finally, the big-$\Theta$ notation is defined as: $f(q) = \Theta\left(g(q)\right) \Leftrightarrow 0 < \liminf_{q \rightarrow \infty}{\left|f(q)/g(q)\right|} \leq \limsup_{q \rightarrow \infty}{\left|f(q)/g(q)\right|} < +\infty$.} (Note that both bounds are $\Theta\left(1\right)$ if $\nu = \frac{1}{q}$.) Therefore, when $V \in \R^{q \times q}_{\textsf{sto}}$ is an additive noise channel, $\delta = O\left(q \nu\right)$ is enough for $W_{\delta} \succeq_{\textsf{\tiny deg}} V$, but a general channel $V \in \R^{q \times q}_{\textsf{sto}}$ requires $\delta = O\left(\nu\right)$ for such degradation. So, in order to account for $q$ different conditional pmfs in the general case (as opposed to a single conditional pmf which characterizes the channel in the additive noise case), we loose a factor of $q$ in the upper bound on $\delta$. Furthermore, we can check using simulations that $W_{\delta} \in \R^{q \times q}_{\textsf{sto}}$ is not in general less noisy than $V \in \R^{q \times q}_{\textsf{sto}}$ for $\delta = (q-1)\nu$. Indeed, counterexamples can be easily obtained by letting $V = E_{j_1,\dots,j_q} W_{\delta}$ for specific values of $1 \leq j_1,\dots,j_q \leq q$, and computationally verifying that $W_{\delta} \not\succeq_{\textsf{\tiny ln}} V + J \in \R^{q \times q}_{\textsf{sto}}$ for appropriate choices of perturbation matrices $J \in \R^{q \times q}$ with sufficiently small Frobenius norm.

We have now proved Theorems \ref{Thm: Chi Squared Divergence Characterization of Less Noisy}, \ref{Thm: Sufficient Condition for Degradation by Symmetric Channels}, and \ref{Thm: Additive Less Noisy Domination and Degradation Regions for Symmetric Channels} from section \ref{Main Results}. The next section relates our results regarding less noisy and degradation preorders to LSIs, and proves Theorem \ref{Thm: Domination of Dirichlet Forms}.

\section{Less noisy domination and logarithmic Sobolev inequalities}
\label{Less Noisy Domination and Logarithmic Sobolev Inequalities}

\textit{Logarithmic Sobolev inequalities} (LSIs) are a class of functional inequalities that shed light on several important phenomena such as concentration of measure, and ergodicity and hypercontractivity of Markov semigroups. We refer readers to \cite{Ledoux1999} and \cite{MarkovSemigroups} for a general treatment of such inequalities, and more pertinently to \cite{LogSobolevInequalitiesDiaconis} and \cite{AnalysisofMixingTimes}, which present LSIs in the context of finite state-space Markov chains. In this section, we illustrate that proving a channel $W \in \R^{q \times q}_{\textsf{sto}}$ is less noisy than a channel $V \in \R^{q \times q}_{\textsf{sto}}$ allows us to translate an LSI for $W$ to an LSI for $V$. Thus, important information about $V$ can be deduced (from its LSI) by proving $W \succeq_{\textsf{\tiny ln}} V$ for an appropriate channel $W$ (such as a $q$-ary symmetric channel) that has a known LSI.

We commence by introducing some appropriate notation and terminology associated with LSIs. For fixed input and output alphabet $\X = \Y = [q]$ with $q \in \N$, we think of a channel $W \in \R^{q \times q}_{\textsf{sto}}$ as a Markov kernel on $\X$. We assume that the ``time homogeneous'' discrete-time Markov chain defined by $W$ is \textit{irreducible}, and has unique \textit{stationary distribution} (or invariant measure) $\pi \in \P_q$ such that $\pi W = \pi$. Furthermore, we define the Hilbert space $\L^2\left(\X,\pi\right)$ of all real functions with domain $\X$ endowed with the inner product:
\begin{equation}
\forall f,g \in \L^2\left(\X,\pi\right), \enspace \left<f,g\right>_{\pi} \triangleq \sum_{x \in \X}{\pi(x)f(x)g(x)}
\end{equation}
and induced norm $\left\|\cdot\right\|_{\pi}$. We construe $W:\L^2\left(\X,\pi\right) \rightarrow \L^2\left(\X,\pi\right)$ as a conditional expectation operator that takes a function $f \in \L^2\left(\X,\pi\right)$, which we can write as a column vector $f = \left[f(0) \cdots f(q-1)\right]^T \in \R^q$, to another function $Wf \in \L^2\left(\X,\pi\right)$, which we can also write as a column vector $Wf \in \R^q$. Corresponding to the discrete-time Markov chain $W$, we may also define a continuous-time \textit{Markov semigroup}:
\begin{equation}
\forall t \geq 0, \enspace H_t \triangleq \exp\left(-t\left(I_q - W\right)\right) \in \R^{q \times q}_{\textsf{sto}}
\end{equation}
where the ``discrete-time derivative'' $W - I_q$ is the \textit{Laplacian operator} that forms the \textit{generator} of the Markov semigroup. The unique stationary distribution of this Markov semigroup is also $\pi$, and we may interpret $H_t : \L^2\left(\X,\pi\right) \rightarrow \L^2\left(\X,\pi\right)$ as a conditional expectation operator for each $t \geq 0$ as well. 

In order to present LSIs, we define the \textit{Dirichlet form} $\Dirichlet_W:\L^2\left(\X,\pi\right) \times \L^2\left(\X,\pi\right) \rightarrow \R$:
\begin{equation}
\forall f,g \in \L^2\left(\X,\pi\right), \enspace \Dirichlet_W\left(f,g\right) \triangleq \left<\left(I_q - W\right)f,g\right>_{\pi}
\end{equation}
which is used to study properties of the Markov chain $W$ and its associated Markov semigroup $\left\{H_t \in \R^{q \times q}_{\textsf{sto}} : t \geq 0\right\}$. ($\Dirichlet_W$ is technically only a Dirichlet form when $W$ is a \textit{reversible} Markov chain, i.e. $W$ is a self-adjoint operator, or equivalently, $W$ and $\pi$ satisfy the \textit{detailed balance condition} \cite[Section 2.3, p.705]{LogSobolevInequalitiesDiaconis}.) Moreover, the quadratic form defined by $\Dirichlet_W$ represents the energy of its input function, and satisfies:
\begin{equation}
\label{Eq: Symmetric Characterization of Dirichlet Form}
\forall f \in \L^2\!\left(\X,\pi\right), \, \Dirichlet_W\!\left(f,f\right) = \left<\!\!\left(I_q - \frac{W + W^{*}}{2}\right)\!\! f,f\!\right>_{\!\!\!\pi}
\end{equation}
where $W^{*}:\L^2\left(\X,\pi\right) \rightarrow \L^2\left(\X,\pi\right)$ is the adjoint operator of $W$. Finally, we introduce a particularly important Dirichlet form corresponding to the channel $W_{(q-1)/q} = \1 \unif$, which
has all uniform conditional pmfs and uniform stationary distribution $\pi = \unif$, known as the \textit{standard Dirichlet form}:
\begin{equation}
\label{Eq: Standard Dirichlet Form}
\begin{aligned}
\Dirichlet_{\textsf{std}}\left(f,g\right) & \triangleq \Dirichlet_{\1 \unif}\left(f,g\right) = \COV_{\unif}\left(f,g\right) \\
& = \sum_{x \in \X}{\frac{f(x)g(x)}{q}} - \left(\sum_{x \in \X}{\frac{f(x)}{q}}\right) \!\! \left(\sum_{x \in \X}{\frac{g(x)}{q}}\right)
\end{aligned}
\end{equation}
for any $f,g \in \L^2\left(\X,\unif\right)$. The quadratic form defined by the standard Dirichlet form is presented in \eqref{Eq: Standard Dirichlet Form Quadratic Form} in subsection \ref{Comparison of Dirichlet Forms}.

We now present the LSIs associated with the Markov chain $W$ and the Markov semigroup $\left\{H_t \in \R^{q \times q}_{\textsf{sto}} : t \geq 0\right\}$ it defines. The LSI for the Markov semigroup with constant $\alpha \in \R$ states that for every $f \in \L^2\left(\X,\pi\right)$ such that $\left\|f\right\|_{\pi} = 1$, we have:
\begin{equation}
\label{Eq: Log-Sobolev Inequality}
D\left(f^2 \pi \, || \, \pi\right) = \sum_{x \in \X}{\pi(x)f^2(x)\log\left(f^2(x)\right)} \leq \frac{1}{\alpha} \Dirichlet_W\left(f,f\right)
\end{equation}
where we construe $\mu = f^2 \pi \in \P_q$ as a pmf such that $\mu(x) = f(x)^2 \pi(x)$ for every $x \in \X$, and $f^2$ behaves like the Radon-Nikodym derivative (or density) of $\mu$ with respect to $\pi$. The largest constant $\alpha$ such that \eqref{Eq: Log-Sobolev Inequality} holds:
\begin{equation}
\label{Eq: Log-Sobolev Constant}
\alpha\!\left(W\right) \triangleq \inf_{\substack{f \in \L^2\left(\X,\pi\right):\\\left\|f\right\|_{\pi} = 1\\D(f^2 \pi || \pi) \neq 0}}{\frac{\Dirichlet_W\left(f,f\right)}{D\left(f^2 \pi \, || \, \pi\right)}} 
\end{equation}
is called the \textit{logarithmic Sobolev constant} (LSI constant) of the Markov chain $W$ (or the Markov chain $\left(W+ W^{*}\right)\!/2$). Likewise, the LSI for the discrete-time Markov chain with constant $\alpha \in \R$ states that for every $f \in \L^2\left(\X,\pi\right)$ such that $\left\|f\right\|_{\pi} = 1$, we have:
\begin{equation}
\label{Eq: Discrete Log-Sobolev Inequality}
D\left(f^2 \pi \, || \, \pi\right) \leq \frac{1}{\alpha} \Dirichlet_{W W^{*}}\!\left(f,f\right)
\end{equation}
where $\Dirichlet_{WW^{*}}:\L^2\left(\X,\pi\right) \times \L^2\left(\X,\pi\right) \rightarrow \R$ is the ``discrete'' Dirichlet form. The largest constant $\alpha$ such that \eqref{Eq: Discrete Log-Sobolev Inequality} holds is the LSI constant of the Markov chain $WW^{*}$, $\alpha\!\left(WW^{*}\right)$, and we refer to it as the \textit{discrete logarithmic Sobolev constant} of the Markov chain $W$. As we mentioned earlier, there are many useful consequences of such LSIs. For example, if \eqref{Eq: Log-Sobolev Inequality} holds with constant \eqref{Eq: Log-Sobolev Constant}, then for every pmf $\mu \in \P_q$:
\begin{equation}
\label{Eq: Continuous Ergodicity}
\forall t \geq 0, \enspace D\left(\mu H_t || \pi\right) \leq e^{-2 \alpha\left(W\right) t} D\left(\mu || \pi\right)
\end{equation}
where the exponent $2\alpha\!\left(W\right)$ can be improved to $4\alpha\!\left(W\right)$ if $W$ is reversible \cite[Theorem 3.6]{LogSobolevInequalitiesDiaconis}. This is a measure of ergodicity of the semigroup $\left\{H_t \in \R^{q \times q}_{\textsf{sto}} : t \geq 0\right\}$. Likewise, if \eqref{Eq: Discrete Log-Sobolev Inequality} holds with constant $\alpha\!\left(WW^{*}\right)$, then for every pmf $\mu \in \P_q$:
\begin{equation}
\label{Eq: Discrete Ergodicity}
\forall n \in \N, \enspace D\left(\mu W^n || \pi\right) \leq\left(1 - \alpha\!\left(WW^{*}\right)\right)^n D\left(\mu || \pi\right)
\end{equation}
as mentioned in \cite[Remark, p.725]{LogSobolevInequalitiesDiaconis} and proved in \cite{DiscreteTimeLogSobolev}. This is also a measure of ergodicity of the Markov chain $W$.

Although LSIs have many useful consequences, LSI constants are difficult to compute analytically. Fortunately, the LSI constant corresponding to $\Dirichlet_{\textsf{std}}$ has been computed in \cite[Appendix, Theorem A.1]{LogSobolevInequalitiesDiaconis}. Therefore, using the relation in \eqref{Eq: Symmetric Channel Dirichlet Form Quadratic Form}, we can compute LSI constants for $q$-ary symmetric channels as well. The next proposition collects the LSI constants for $q$-ary symmetric channels (which are irreducible for $\delta \in (0,1]$) as well as some other related quantities.  

\begin{proposition}[Constants of Symmetric Channels]
\label{Prop: Constants of Symmetric Channels}
The $q$-ary symmetric channel $W_{\delta} \in \R^{q \times q}_{\textsf{sto}}$ with $q \geq 2$ has:
\begin{enumerate}
\item LSI constant: 
$$ \alpha\!\left(W_{\delta}\right) = \left\{
     \begin{array}{ll}
		 \delta , & q = 2 \\
     \frac{(q-2)\delta}{(q-1)\log\left(q - 1\right)} , & q > 2
     \end{array} \right. $$
for $\delta \in (0,1]$.
\item discrete LSI constant:
$$ \alpha\!\left(W_{\delta} W_{\delta}^{*}\right) = \alpha\!\left(W_{\delta^{\prime}}\right) = \left\{
     \begin{array}{ll}
		 2\delta(1-\delta) , & q = 2 \\
     \frac{(q-2)\left(2q - 2 - q\delta\right)\delta}{(q-1)^2\log\left(q - 1\right)} , & q > 2      
     \end{array} \right. $$
for $\delta \in (0,1]$, where $\delta^{\prime} = \delta \big(2 - \frac{q \delta}{q-1}\big)$.
\item Hirschfeld-Gebelein-R\'{e}nyi maximal correlation corresponding to the uniform stationary distribution $\unif \in \P_q$:
$$ \rho_{\textsf{max}}\!\left(\unif,W_{\delta}\right) = \left|1-\delta-\frac{\delta}{q-1}\right| $$
for $\delta \in [0,1]$, where for any channel $W \in \R^{q \times r}_{\textsf{sto}}$ and any source pmf $P_X \in \P_q$, we define the maximal correlation between the input random variable $X \in [q]$ and the output random variable $Y \in [r]$ (with joint pmf $P_{X,Y}(x,y) = P_X(x) W_{Y|X}(y|x)$) as \cite{RenyiCorrelation}: 
$$ \rho_{\textsf{max}}\!\left(P_X,W\right) \triangleq \sup_{\substack{f:[q] \rightarrow \R, \, g:[r] \rightarrow \R\\ \E\left[f(X)\right] = \E\left[g(Y)\right] = 0 \\ \E\left[f^2(X)\right] = \E\left[g^2(Y)\right] = 1}}{\E\left[f(X)g(Y)\right]} . $$
\item contraction coefficient for KL divergence bounded by:
$$ \left(1-\delta-\frac{\delta}{q-1}\right)^2 \leq \eta_{\textsf{\tiny KL}}\!\left(W_{\delta}\right) \leq \left|1-\delta-\frac{\delta}{q-1}\right| $$
for $\delta \in \left[0,1\right]$.
\end{enumerate} 
\end{proposition}

\begin{proof} See Appendix \ref{App: Proofs of Propositions}.
\end{proof}

In view of Proposition \ref{Prop: Constants of Symmetric Channels} and the intractability of computing LSI constants for general Markov chains, we often ``compare'' a given irreducible channel $V \in \R^{q \times q}_{\textsf{sto}}$ with a $q$-ary symmetric channel $W_{\delta} \in \R^{q \times q}_{\textsf{sto}}$ to try and establish an LSI for it. We assume for the sake of simplicity that $V$ is doubly stochastic and has uniform stationary pmf (just like $q$-ary symmetric channels). Usually, such a comparison between $W_{\delta}$ and $V$ requires us to prove domination of Dirichlet forms, such as:
\begin{equation}
\label{Eq: Dirichlet Form Domination}
\forall f \in \L^2\left(\X,\unif\right), \enspace \Dirichlet_{V}\!\left(f,f\right) \geq \Dirichlet_{W_{\delta}}\!\left(f,f\right) = \frac{q\delta}{q-1} \Dirichlet_{\textsf{std}}\left(f,f\right)
\end{equation}
where we use \eqref{Eq: Symmetric Channel Dirichlet Form Quadratic Form}. Such pointwise domination results immediately produce LSIs, \eqref{Eq: Log-Sobolev Inequality} and \eqref{Eq: Discrete Log-Sobolev Inequality}, for $V$. Furthermore, they also lower bound the LSI constants of $V$; for example:
\begin{equation}
\alpha\!\left(V\right) \geq \alpha\!\left(W_{\delta}\right) .
\end{equation}
This in turn begets other results such as \eqref{Eq: Continuous Ergodicity} and \eqref{Eq: Discrete Ergodicity} for the channel $V$ (albeit with worse constants in the exponents since the LSI constants of $W_{\delta}$ are used instead of those for $V$). More general versions of Dirichlet form domination between Markov chains on different state spaces with different stationary distributions, and the resulting bounds on their LSI constants are presented in \cite[Lemmata 3.3 and 3.4]{LogSobolevInequalitiesDiaconis}. We next illustrate that the information theoretic notion of less noisy domination is a sufficient condition for various kinds of pointwise Dirichlet form domination. 

\begin{theorem}[Domination of Dirichlet Forms]
\label{Thm: Extended Domination of Dirichlet Forms}
Let $W,V \in \R^{q \times q}_{\textsf{sto}}$ be doubly stochastic channels, and $\pi = \unif$ be the uniform stationary distribution. Then, the following are true:
\begin{enumerate}
\item If $W \succeq_{\textsf{\tiny ln}} V$, then:
$$ \forall f \in \L^2\left(\X,\unif\right), \enspace \Dirichlet_{V V^{*}}\left(f,f\right) \geq \Dirichlet_{W W^{*}}\left(f,f\right) . $$
\item If $W \in \R^{q \times q}_{\succeq 0}$ is positive semidefinite, $V$ is normal (i.e. $V^TV = V V^T$), and $W \succeq_{\textsf{\tiny ln}} V$, then:
$$ \forall f \in \L^2\left(\X,\unif\right), \enspace \Dirichlet_{V}\left(f,f\right) \geq \Dirichlet_{W}\left(f,f\right) . $$
\item If $W = W_{\delta} \in \R^{q \times q}_{\textsf{sto}}$ is any $q$-ary symmetric channel with $\delta \in \big[0,\frac{q-1}{q}\big]$ and $W_{\delta} \succeq_{\textsf{\tiny ln}} V$, then:
$$ \forall f \in \L^2\left(\X,\unif\right), \enspace \Dirichlet_{V}\left(f,f\right) \geq \frac{q\delta}{q-1}\Dirichlet_{\textsf{std}}\left(f,f\right) . $$
\end{enumerate} 
\end{theorem}   

\begin{proof} ~\newline
\textbf{Part 1:} First observe that:
\begin{align*}
\forall f \in \L^2\left(\X,\unif\right), \enspace \Dirichlet_{WW^{*}}\!\left(f,f\right) & = \frac{1}{q} f^T \left(I_q - WW^{T}\right)f \\
\forall f \in \L^2\left(\X,\unif\right), \enspace \Dirichlet_{V V^{*}}\!\left(f,f\right) & = \frac{1}{q} f^T \left(I_q - VV^{T}\right)f
\end{align*}
where we use the facts that $W^T = W^{*}$ and $V^T = V^{*}$ because the stationary distribution is uniform. This implies that $\Dirichlet_{V V^{*}}\!\left(f,f\right) \geq \Dirichlet_{W W^{*}}\!\left(f,f\right)$ for every $f \in \L^2\left(\X,\unif\right)$ if and only if $I_q - V V^T \succeq_{\textsf{\tiny PSD}} I_q - W W^T$, which is true if and only if $W W^T \succeq_{\textsf{\tiny PSD}} V V^T$. Since $W \succeq_{\textsf{\tiny ln}} V$, we get $W W^T \succeq_{\textsf{\tiny PSD}} V V^T$ from part 2 of Proposition \ref{Prop: Loewner and Spectral Characterizations of Less Noisy} after letting $P_X = \unif = P_X W = P_X V$. \\
\textbf{Part 2:} Once again, we first observe using \eqref{Eq: Symmetric Characterization of Dirichlet Form} that:
\begin{align*}
\forall f \in \L^2\left(\X,\unif\right), \enspace \Dirichlet_{W}\left(f,f\right) & = \frac{1}{q} f^T \left(I_q - \frac{W + W^{T}}{2}\right)f , \\
\forall f \in \L^2\left(\X,\unif\right), \enspace \Dirichlet_V\left(f,f\right) & = \frac{1}{q} f^T \left(I_q - \frac{V + V^{T}}{2}\right)f .
\end{align*}
So, $\Dirichlet_{V}\left(f,f\right) \geq \Dirichlet_{W}\left(f,f\right)$ for every $f \in \L^2\left(\X,\unif\right)$ if and only if $\left(W + W^T\right)\!/2 \succeq_{\textsf{\tiny PSD}} \left(V + V^T\right)\!/2$. Since $W W^T \succeq_{\textsf{\tiny PSD}} V V^T$ from the proof of part 1, it is sufficient to prove that:
\begin{equation}
\label{Eq: Gramian implies Symmetric Part}
W W^T \succeq_{\textsf{\tiny PSD}} V V^T \enspace \Rightarrow \enspace \frac{W + W^T}{2} \succeq_{\textsf{\tiny PSD}} \frac{V + V^T}{2} .
\end{equation}
Lemma \ref{Lemma: Gramian Loewner Domination implies Symmetric Part Loewner Domination} in Appendix \ref{App: Auxiliary Results} establishes the claim in \eqref{Eq: Gramian implies Symmetric Part} because $W \in \R^{q \times q}_{\succeq 0}$ and $V$ is a normal matrix. \\
\textbf{Part 3:} We note that when $V$ is a normal matrix, this result follows from part 2 because $W_{\delta} \in \R^{q \times q}_{\succeq 0}$ for $\delta \in \big[0,\frac{q-1}{q}\big]$, as can be seen from part 2 of Proposition \ref{Prop: Properties of Symmetric Channel Matrices}. For a general doubly stochastic channel $V$, we need to prove that $\Dirichlet_{V}\left(f,f\right) \geq \Dirichlet_{W_{\delta}}\left(f,f\right) = \frac{q\delta}{q-1}\Dirichlet_{\textsf{std}}(f,f)$ for every $f \in \L^2\left(\X,\unif\right)$ (where we use \eqref{Eq: Symmetric Channel Dirichlet Form Quadratic Form}). Following the proof of part 2, it is sufficient to prove \eqref{Eq: Gramian implies Symmetric Part} with $W = W_{\delta}$:\footnote{Note that \eqref{Eq: Gramian implies Symmetric Part} trivially holds for $W = W_{\delta}$ with $\delta = (q-1)/q$, because $W_{(q-1)/q} = W_{(q-1)/q}^2 = \1 \unif \succeq_{\textsf{\tiny PSD}} V V^T$ implies that $V = W_{(q-1)/q}$.}
$$ W_{\delta}^2 \succeq_{\textsf{\tiny PSD}} VV^T \enspace \Rightarrow \enspace W_{\delta} \succeq_{\textsf{\tiny PSD}} \frac{V + V^T}{2} $$
where $W_{\delta}^2 = W_{\delta} W_{\delta}^T$ and $W_{\delta} = \left(W_{\delta} + W_{\delta}^T\right)\!/2$. Recall the \textit{L\"{o}wner-Heinz theorem} \cite{LoewnerHeinzTheorem1,LoewnerHeinzTheorem2}, (cf. \cite[Section 6.6, Problem 17]{MatrixAnalysis}), which states that for $A,B \in \R^{q \times q}_{\succeq 0}$ and $0 \leq p \leq 1$:
\begin{equation}
\label{Eq: Loewner-Heinz}
A \succeq_{\textsf{\tiny PSD}} B \enspace \Rightarrow \enspace A^p \succeq_{\textsf{\tiny PSD}} B^p 
\end{equation}
or equivalently, $f:\left[0,\infty\right) \rightarrow \R, \, f(x) = x^p$ is an \textit{operator monotone} function for $p \in [0,1]$. Using \eqref{Eq: Loewner-Heinz} with $p = \frac{1}{2}$ (cf. \cite[Corollary 7.7.4 (b)]{BasicMatrixAnalysis}), we have:
$$ W_{\delta}^2 \succeq_{\textsf{\tiny PSD}} VV^T \enspace \Rightarrow \enspace W_{\delta} \succeq_{\textsf{\tiny PSD}} \left(VV^T\right)^{\frac{1}{2}} $$
because the Gramian matrix $VV^T \in \R^{q \times q}_{\succeq 0}$. (Here, $\left(VV^T\right)^{\frac{1}{2}}$ is the unique positive semidefinite square root matrix of $VV^T$.) 

Let $VV^T = Q\Lambda Q^T$ and $(V + V^T)/2 = U \Sigma U^T$ be the spectral decompositions of $VV^T$ and $(V + V^T)/2$, where $Q$ and $U$ are orthogonal matrices with eigenvectors as columns, and $\Lambda$ and $\Sigma$ are diagonal matrices of eigenvalues. Since $VV^T$ and $(V + V^T)/2$ are both doubly stochastic, they both have the unit norm eigenvector $\1/\sqrt{q}$ corresponding to the maximum eigenvalue of unity. In fact, we have:
$$ \left(VV^T\right)^{\frac{1}{2}} \frac{\1}{\sqrt{q}} = \frac{\1}{\sqrt{q}} \quad \text{and} \quad \left(\frac{V + V^T}{2}\right) \frac{\1}{\sqrt{q}} = \frac{\1}{\sqrt{q}} $$
where we use the fact that $(VV^T)^{\frac{1}{2}} =  Q\Lambda^{\frac{1}{2}} Q^T$ is the spectral decomposition of $(VV^T)^{\frac{1}{2}}$. For any matrix $A \in \R^{q \times q}_{\textsf{sym}}$, let $\lambda_1(A) \geq \lambda_2(A) \geq \cdots \geq \lambda_q(A)$ denote the eigenvalues of $A$ in descending order. Without loss of generality, we assume that $\left[\Lambda\right]_{j,j} = \lambda_j(VV^T)$ and $\left[\Sigma\right]_{j,j} = \lambda_j((V + V^T)/2)$ for every $1 \leq j \leq q$. So, $\lambda_1((VV^T)^{\frac{1}{2}}) = \lambda_1((V + V^T)/2) = 1$, and the first columns of both $Q$ and $U$ are equal to $\1/\sqrt{q}$. 

From part 2 of Proposition \ref{Prop: Properties of Symmetric Channel Matrices}, we have $W_{\delta} = Q D Q^T = U D U^T$, where $D$ is the diagonal matrix of eigenvalues such that $\left[D\right]_{1,1} = \lambda_1(W_{\delta}) = 1$ and $\left[D\right]_{j,j} = \lambda_j(W_{\delta}) = 1 - \delta - \frac{\delta}{q-1}$ for $2 \leq j \leq q$. Note that we may use either of the eigenbases, $Q$ or $U$, because they both have first column $\1/\sqrt{q}$, which is the eigenvector of $W_{\delta}$ corresponding to $\lambda_1(W_{\delta}) = 1$ since $W_{\delta}$ is doubly stochastic, and the remaining eigenvector columns are permitted to be any orthonormal basis of $\textsf{\small span}(\1/\sqrt{q})^{\perp}$ as $\lambda_j(W_{\delta}) = 1 - \delta - \frac{\delta}{q-1}$ for $2 \leq j \leq q$. Hence, we have:
\begin{align*}
W_{\delta} \succeq_{\textsf{\tiny PSD}} \left(VV^T\right)^{\frac{1}{2}} & \Leftrightarrow QDQ^T \succeq_{\textsf{\tiny PSD}} Q\Lambda^{\frac{1}{2}}Q^T \Leftrightarrow D \succeq_{\textsf{\tiny PSD}} \Lambda^{\frac{1}{2}} , \\
W_{\delta} \succeq_{\textsf{\tiny PSD}} \frac{V + V^T}{2} & \Leftrightarrow UDU^T \succeq_{\textsf{\tiny PSD}} U\Sigma U^T \Leftrightarrow D \succeq_{\textsf{\tiny PSD}} \Sigma .
\end{align*}

In order to show that $D \succeq_{\textsf{\tiny PSD}} \Lambda^{\frac{1}{2}} \Rightarrow D \succeq_{\textsf{\tiny PSD}} \Sigma$, it suffices to prove that $\Lambda^{\frac{1}{2}} \succeq_{\textsf{\tiny PSD}} \Sigma$. Recall from \cite[Corollary 3.1.5]{MatrixAnalysis} that for any matrix $A \in \R^{q \times q}$, we have:\footnote{This states that for any matrix $A \in \R^{q \times q}$, the $i$th largest eigenvalue of the symmetric part of $A$ is less than or equal to the $i$th largest singular value of $A$ (which is the $i$th largest eigenvalue of the unique positive semidefinite part $(AA^T)^{1/2}$ in the polar decomposition of $A$) for every $1 \leq i \leq q$.}
\begin{equation}
\forall i \in \left\{1,\dots,q\right\}, \enspace \lambda_i\left(\left(AA^T\right)^{\frac{1}{2}}\right) \geq \lambda_i\left(\frac{A + A^T}{2}\right) . 
\end{equation}
Hence, $\Lambda^{\frac{1}{2}} \succeq_{\textsf{\tiny PSD}} \Sigma$ is true, cf. \cite[Lemma 2.5]{NashInequalities}. This completes the proof.
\end{proof}

Theorem \ref{Thm: Extended Domination of Dirichlet Forms} includes Theorem \ref{Thm: Domination of Dirichlet Forms} from section \ref{Main Results} as part 3, and also provides two other useful pointwise Dirichlet form domination results. Part 1 of Theorem \ref{Thm: Extended Domination of Dirichlet Forms} states that less noisy domination implies discrete Dirichlet form domination. In particular, if we have $W_{\delta} \succeq_{\textsf{\tiny ln}} V$ for some irreducible $q$-ary symmetric channel $W_{\delta} \in \R^{q \times q}_{\textsf{sto}}$ and irreducible doubly stochastic channel $V \in \R^{q \times q}_{\textsf{sto}}$, then part 1 implies that:
\begin{equation}
\label{Eq: Discrete Ergodicity of V}
\forall n \in \N, \enspace D\left(\mu V^n || \unif\right) \leq\left(1 - \alpha\!\left(W_{\delta} W_{\delta}^{*}\right)\right)^n D\left(\mu || \unif\right)
\end{equation}
for all pmfs $\mu \in \P_q$, where $\alpha\!\left(W_{\delta} W_{\delta}^{*}\right)$ is computed in part 2 of Proposition \ref{Prop: Constants of Symmetric Channels}. However, it is worth mentioning that \eqref{Eq: Discrete Ergodicity} for $W_{\delta}$ and Proposition \ref{Prop: KL Divergence Characterization of Less Noisy} directly produce \eqref{Eq: Discrete Ergodicity of V}. So, such ergodicity results for the discrete-time Markov chain $V$ do not require the full power of the Dirichlet form domination in part 1. Regardless, Dirichlet form domination results, such as in parts 2 and 3, yield several functional inequalities (like Poincar\'{e} inequalities and LSIs) which have many other potent consequences as well. 

Parts 2 and 3 of Theorem \ref{Thm: Extended Domination of Dirichlet Forms} convey that less noisy domination also implies the usual (continuous) Dirichlet form domination under regularity conditions. We note that in part 2, the channel $W$ is more general than that in part 3, but the channel $V$ is restricted to be normal (which includes the case where $V$ is an additive noise channel). The proofs of these parts essentially consist of two segments. The first segment uses part 1, and the second segment illustrates that pointwise domination of discrete Dirichlet forms implies pointwise domination of Dirichlet forms (as shown in \eqref{Eq: Dirichlet Form Domination}). This latter segment is encapsulated in Lemma \ref{Lemma: Gramian Loewner Domination implies Symmetric Part Loewner Domination} of Appendix \ref{App: Auxiliary Results} for part 2, and requires a slightly more sophisticated proof pertaining to $q$-ary symmetric channels in part 3.

\section{Conclusion}
\label{Conclusion}

In closing, we briefly reiterate our main results by delineating a possible program for proving LSIs for certain Markov chains. Given an arbitrary irreducible doubly stochastic channel $V \in \R^{q \times q}_{\textsf{sto}}$ with minimum entry $\nu = \min\{\left[V\right]_{i,j} : 1 \leq i,j \leq q\} > 0$ and $q \geq 2$, we can first use Theorem \ref{Thm: Sufficient Condition for Degradation by Symmetric Channels} to generate a $q$-ary symmetric channel $W_{\delta} \in \R^{q \times q}_{\textsf{sto}}$ with $\delta = \nu/\big(1-(q-1)\nu+\frac{\nu}{q-1}\big)$ such that $W_{\delta} \succeq_{\textsf{\tiny deg}} V$. This also means that $W_{\delta} \succeq_{\textsf{\tiny ln}} V$, using Proposition \ref{Prop: Relations between Channel Preorders}. Moreover, the $\delta$ parameter can be improved using Theorem \ref{Thm: Additive Less Noisy Domination and Degradation Regions for Symmetric Channels} (or Propositions \ref{Prop: Degradation by Symmetric Channels} and \ref{Prop: Less Noisy Domination by Symmetric Channels}) if $V$ is an additive noise channel. We can then use Theorem \ref{Thm: Extended Domination of Dirichlet Forms} to deduce a pointwise domination of Dirichlet forms. Since $W_{\delta}$ satisfies the LSIs \eqref{Eq: Log-Sobolev Inequality} and \eqref{Eq: Discrete Log-Sobolev Inequality} with corresponding LSI constants given in Proposition \ref{Prop: Constants of Symmetric Channels}, Theorem \ref{Thm: Extended Domination of Dirichlet Forms} establishes the following LSIs for $V$:
\begin{align}
D\left(f^2 \unif \, || \, \unif\right) & \leq \frac{1}{\alpha\!\left(W_{\delta}\right)} \Dirichlet_V\left(f,f\right) \\
D\left(f^2 \unif \, || \, \unif\right) & \leq \frac{1}{\alpha\!\left(W_{\delta} W_{\delta}^{*}\right)} \Dirichlet_{V V^{*}}\!\left(f,f\right)
\end{align}
for every $f \in \L^2\left(\X,\unif\right)$ such that $\left\|f\right\|_{\unif} = 1$. These inequalities can be used to derive a myriad of important facts about $V$. We note that the equivalent characterizations of the less noisy preorder in Theorem \ref{Thm: Chi Squared Divergence Characterization of Less Noisy} and Proposition \ref{Prop: Loewner and Spectral Characterizations of Less Noisy} are particularly useful for proving some of these results. Finally, we accentuate that Theorems \ref{Thm: Sufficient Condition for Degradation by Symmetric Channels} and \ref{Thm: Additive Less Noisy Domination and Degradation Regions for Symmetric Channels} address our motivation in subsection \ref{Main Question and Motivation} by providing analogs of the relationship between less noisy domination by $q$-ary erasure channels and contraction coefficients in the context of $q$-ary symmetric channels.

\appendices

\section{Basics of majorization theory} 
\label{App: Basics of Majorization Theory}

Since we use some majorization arguments in our analysis, we briefly introduce the notion of \textit{group majorization} over row vectors in $\R^q$ (with $q \in \N$) in this appendix. Given a group $\G \subseteq \R^{q \times q}$ of matrices (with the operation of matrix multiplication), we may define a preorder called $\G$-majorization over row vectors in $\R^q$. For two row vectors $x,y \in \R^q$, we say that $x$ $\G$-majorizes $y$ if $y \in \textsf{\small conv}\left(\left\{xG:G \in \G\right\}\right)$, where $\left\{xG:G \in \G\right\}$ is the orbit of $x$ under the group $\G$. Group majorization intuitively captures a notion of ``spread'' of vectors. So, $x$ $\G$-majorizes $y$ when $x$ is more spread out than $y$ with respect to $\G$. We refer readers to \cite[Chapter 14, Section C]{Majorization} and the references therein for a thorough treatment of group majorization. If we let $\G$ be the symmetric group of all permutation matrices in $\R^{q \times q}$, then $\G$-majorization corresponds to traditional majorization of vectors in $\R^q$ as introduced in \cite{HardyLittlewoodPolya}. The next proposition collects some results about traditional majorization.

\begin{proposition}[Majorization \cite{HardyLittlewoodPolya, Majorization}]
\label{Prop: Majorization} 
Given two row vectors $x = \left(x_1,\dots,x_q\right), y = \left(y_1,\dots,y_q\right) \in \R^q$, let $x_{(1)} \leq \cdots \leq x_{(q)}$ and $y_{(1)} \leq \cdots \leq y_{(q)}$ denote the re-orderings of $x$ and $y$ in ascending order. Then, the following are equivalent:
\begin{enumerate}
\item $x$ majorizes $y$, or equivalently, $y$ resides in the convex hull of all permutations of $x$.
\item $y = xD$ for some doubly stochastic matrix $D \in \R^{q \times q}_{\textsf{sto}}$.
\item The entries of $x$ and $y$ satisfy:
\begin{align*}
\sum_{i = 1}^{k}{x_{(i)}} & \leq \sum_{i = 1}^{k}{y_{(i)}} \, , \enspace \text{for} \enspace k = 1,\dots,q-1 \, , \\
\text{and} \enspace \sum_{i = 1}^{q}{x_{(i)}} & = \sum_{i = 1}^{q}{y_{(i)}} \, .
\end{align*} 
\end{enumerate}
When these conditions are true, we write $x \succeq_{\textsf{\tiny maj}} y$.
\end{proposition}

In the context of subsection \ref{Symmetric Channels and their Properties}, given an Abelian group $(\X,\oplus)$ of order $q$, another useful notion of $\G$-majorization can be obtained by letting $\G = \left\{P_x \in \R^{q \times q}: x \in \X\right\}$ be the group of permutation matrices defined in \eqref{Eq: Permutation Representation} that is isomorphic to $(\X,\oplus)$. For such choice of $\G$, we write $x \succeq_{\text{\tiny $\X$}} y$ when $x$ $\G$-majorizes (or $\X$-majorizes) $y$ for any two row vectors $x,y \in \R^q$. We will only require one fact about such group majorization, which we present in the next proposition.

\begin{proposition}[Group Majorization]
\label{Prop: Group Majorization}
Given two row vectors $x,y \in \R^q$, $x \succeq_{\text{\tiny $\X$}} y$ if and only if there exists $\lambda \in \P_q$ such that $y = x \, \textsf{\small circ}_{\X}(\lambda)$. 
\end{proposition}

\begin{proof}
Observe that:
\begin{align*}
x \succeq_{\text{\tiny $\X$}} y & \Leftrightarrow y \in \textsf{\small conv}\left(\left\{x P_z : z \in \X\right\}\right) \\
& \Leftrightarrow y = \lambda \, \textsf{\small circ}_{\X}(x) \enspace \text{for some} \enspace \lambda \in \P_q \\
& \Leftrightarrow y = x \, \textsf{\small circ}_{\X}(\lambda) \enspace \text{for some} \enspace \lambda \in \P_q 
\end{align*}
where the second step follows from \eqref{Eq: Group Circulant Matrix Row-Column Decomposition}, and the final step follows from the commutativity of $\X$-circular convolution.
\end{proof}

Proposition \ref{Prop: Group Majorization} parallels the equivalence between parts 1 and 2 of Proposition \ref{Prop: Majorization}, because $\textsf{\small circ}_{\X}(\lambda)$ is a doubly stochastic matrix for every pmf $\lambda \in \P_q$. In closing this appendix, we mention a well-known special case of such group majorization. When $(\X,\oplus)$ is the cyclic Abelian group $\mathbb{Z}/q\mathbb{Z}$ of integers with addition modulo $q$, $\G = \left\{I_q,P_q,P_q^2,\dots,P_q^{q-1}\right\}$ is the group of all cyclic permutation matrices in $\R^{q \times q}$, where $P_q \in \R^{q \times q}$ is defined in \eqref{Eq: Generator Cyclic Permutation Matrix}. The corresponding notion of $\G$-majorization is known as \textit{cyclic majorization}, cf. \cite{CyclicMajorization}.

\section{Proofs of propositions \ref{Prop: Properties of Symmetric Channel Matrices} and \ref{Prop: Constants of Symmetric Channels}}
\label{App: Proofs of Propositions}

\renewcommand{\proofname}{Proof of Proposition \ref{Prop: Properties of Symmetric Channel Matrices}}

\begin{proof} ~\newline
\textbf{Part 1:} This is obvious from \eqref{Eq: Symmetric Channel}. \\
\textbf{Part 2:} Since the DFT matrix jointly diagonalizes all circulant matrices, it diagonalizes every $W_{\delta}$ for $\delta \in \R$ (using part 1). The corresponding eigenvalues are all real because $W_{\delta}$ is symmetric. To explicitly compute these eigenvalues, we refer to \cite[Problem 2.2.P10]{BasicMatrixAnalysis}. Observe that for any row vector $x = (x_0,\dots,x_{q-1}) \in \R^q$, the corresponding circulant matrix satisfies:
\begin{align*}
\textsf{\small circ}_{\mathbb{Z}/q\mathbb{Z}}(x) & = \sum_{k = 0}^{q-1}{x_k P_q^k} = F_q \left(\sum_{k = 0}^{q-1}{x_k D_q^k}\right) F_q^H \\
& = F_q \, \textsf{\small diag}\!\left(\sqrt{q} \, x F_q\right) F_q^H 
\end{align*}
where the first equality follows from \eqref{Eq: Group Circulant Matrix Decomposition} for the group $\mathbb{Z}/q\mathbb{Z}$ \cite[Section 0.9.6]{BasicMatrixAnalysis}, $D_q = \textsf{\small diag}((1,\omega,\omega^2,\dots,\omega^{q-1}))$, and $P_q = F_q D_q F_q^H \in \R^{q \times q}$ is defined in \eqref{Eq: Generator Cyclic Permutation Matrix}. Hence, we have:
\begin{align*}
\lambda_j\left(W_{\delta}\right) & = \sum_{k = 1}^{q}{\left(w_{\delta}\right)_k \omega^{(j-1)(k-1)}} \\
& = \left\{
     \begin{array}{ll}
       1 , & j = 1 \\
       1-\delta-\frac{\delta}{q-1} , & j = 2,\dots,q
     \end{array} \right.
\end{align*}
where $\w = (1-\delta,\delta/(q-1),\dots,\delta/(q-1))$. \\
\textbf{Part 3:} This is also obvious from \eqref{Eq: Symmetric Channel}\textemdash{}recall that a square stochastic matrix is doubly stochastic if and only if its stationary distribution is uniform \cite[Section 8.7]{BasicMatrixAnalysis}. \\
\textbf{Part 4:} For $\delta \neq \frac{q-1}{q}$, we can verify that $W_{\tau}W_{\delta} = I_q$ when $\tau = \frac{-\delta}{1-\delta-\frac{\delta}{q-1}}$ by direct computation:
\begin{align*}
\left[W_{\tau}W_{\delta}\right]_{j,j} & = \left(1-\tau\right)\left(1-\delta\right) + \left(q-1\right)\left(\frac{\tau}{q-1}\right)\left(\frac{\delta}{q-1}\right) \\
& = 1 \, , \enspace \text{for} \enspace j = 1,\dots,q \, , \\
\left[W_{\tau}W_{\delta}\right]_{j,k} & = \frac{\delta \left(1-\tau\right)}{q-1} + \frac{\tau \left(1-\delta\right)}{q-1} + \left(q-2\right)\frac{\tau \delta}{\left(q-1\right)^2} \\
& = 0 \, , \enspace \text{for} \enspace j \neq k \enspace \text{and} \enspace 1 \leq j,k \leq q \, .
\end{align*}
The $\delta = \frac{q-1}{q}$ case follows from \eqref{Eq: Symmetric Channel}. \\
\textbf{Part 5:} The set $\big\{W_{\delta}:\delta \in \R\backslash\big\{\frac{q-1}{q}\big\}\big\}$ is closed under matrix multiplication. Indeed, for $\epsilon,\delta \in \R\backslash\big\{\frac{q-1}{q}\big\}$, we can straightforwardly verify that $W_{\epsilon}W_{\delta} = W_{\tau}$ with $\tau = \epsilon + \delta - \epsilon \delta - \frac{\epsilon \delta}{q-1}$. Moreover, $\tau \neq \frac{q-1}{q}$ because $W_{\tau}$ is invertible (since $W_{\epsilon}$ and $W_{\delta}$ are invertible using part 4). The set also includes the identity matrix as $W_0 = I_q$, and multiplicative inverses (using part 4). Finally, the associativity of matrix multiplication and the commutativity of circulant matrices proves that $\big\{W_{\delta}:\delta \in \R\backslash\big\{\frac{q-1}{q}\big\}\big\}$ is an Abelian group.  
\end{proof}

\renewcommand{\proofname}{Proof of Proposition \ref{Prop: Constants of Symmetric Channels}}

\begin{proof} ~\newline
\textbf{Part 1:} We first recall from \cite[Appendix, Theorem A.1]{LogSobolevInequalitiesDiaconis} that the Markov chain $\1\unif \in \R^{q \times q}_{\textsf{sto}}$ with uniform stationary distribution $\pi = \unif \in \P_q$ has LSI constant:
$$ \alpha\!\left(\1\unif\right) = \!\! \inf_{\substack{f \in \L^2\left(\X,\unif\right):\\\left\|f\right\|_{\unif} = 1\\D(f^2 \unif || \unif) \neq 0}}{\frac{\Dirichlet_{\textsf{std}}\left(f,f\right)}{D\left(f^2 \unif \, || \, \unif\right)}} = 
\left\{
     \begin{array}{ll}
		 \frac{1}{2} , & q = 2 \\
     \frac{1-\frac{2}{q}}{\log\left(q - 1\right)} , & q > 2      
     \end{array} \right. . $$
Now using \eqref{Eq: Symmetric Channel Dirichlet Form Quadratic Form}, $\alpha\!\left(W_{\delta}\right) = \frac{q\delta}{q-1}\alpha\!\left(\1\unif\right)$, which proves part 1. \\
\textbf{Part 2:} Observe that $W_{\delta}W_{\delta}^{*} = W_{\delta}W_{\delta}^T = W_{\delta}^2 = W_{\delta^{\prime}}$, where the first equality holds because $W_{\delta}$ has uniform stationary pmf, and $\delta^{\prime} = \delta \big(2 - \frac{q \delta}{q-1}\big)$ using the proof of part 5 of Proposition \ref{Prop: Properties of Symmetric Channel Matrices}. As a result, the discrete LSI constant $\alpha\!\left(W_{\delta} W_{\delta}^{*}\right) = \alpha\!\left(W_{\delta^{\prime}}\right)$, which we can calculate using part 1 of this proposition. \\
\textbf{Part 3:} It is well-known in the literature that $\rho_{\textsf{max}}\!\left(\unif,W_{\delta}\right)$ equals the second largest singular value of the divergence transition matrix $\textsf{\small diag}\!\left(\sqrt{\unif}\right)^{-1} W_{\delta} \, \textsf{\small diag}\!\left(\sqrt{\unif}\right) = W_{\delta}$ (see \cite[Subsection I-B]{BoundsbetweenContractionCoefficients} and the references therein). Hence, from part 2 of Proposition \ref{Prop: Properties of Symmetric Channel Matrices}, we have $\rho_{\textsf{max}}\!\left(\unif,W_{\delta}\right) = \big|1 - \delta - \frac{\delta}{q-1}\big|$. \\
\textbf{Part 4:} First recall the \textit{Dobrushin contraction coefficient} (for total variation distance) for any channel $W \in \R^{q \times r}_{\textsf{sto}}$:
\begin{align}
\eta_{\textsf{\tiny TV}}\!\left(W\right) & \triangleq \sup_{\substack{P_X,Q_X \in \P_q:\\P_X \neq Q_X}}{\frac{\left\|P_X W - Q_X W\right\|_{\ell^1}}{\left\|P_X - Q_X\right\|_{\ell^1}}} \\
& = \frac{1}{2} \max_{x,x^{\prime} \in [q]}{\left\|W_{Y|X}(\cdot|x) - W_{Y|X}(\cdot|x^{\prime})\right\|_{\ell^1}}
\end{align}
where $\left\|\cdot\right\|_{\ell^1}$ denotes the $\ell^1$-norm, and the second equality is Dobrushin's two-point characterization of $\eta_{\textsf{\tiny TV}}$ \cite{DobrushinContraction}. Using this characterization, we have:
$$ \eta_{\textsf{\tiny TV}}\!\left(W_{\delta}\right) = \frac{1}{2} \max_{x,x^{\prime} \in [q]}{\left\|\w P_q^{x} - \w P_q^{x^{\prime}} \right\|_{\ell^1}} = \left|1-\delta-\frac{\delta}{q-1}\right| $$
where $\w$ is the noise pmf of $W_{\delta}$ for $\delta \in [0,1]$, and $P_q \in \R^{q \times q}$ is defined in \eqref{Eq: Generator Cyclic Permutation Matrix}. It is well-known in the literature (see e.g. the introduction of \cite{ContractionCoefficients} and the references therein) that:
\begin{equation}
\rho_{\textsf{max}}\!\left(\unif,W_{\delta}\right)^2 \leq \eta_{\textsf{\tiny KL}}\!\left(W_{\delta}\right) \leq \eta_{\textsf{\tiny TV}}\!\left(W_{\delta}\right) . 
\end{equation}
Hence, the value of $\eta_{\textsf{\tiny TV}}\!\left(W_{\delta}\right)$ and part 3 of this proposition establish part 4. This completes the proof.
\end{proof}

\renewcommand{\proofname}{Proof}

\section{Auxiliary results}
\label{App: Auxiliary Results}

\begin{proposition}[Properties of Domination Factor Function]
\label{Prop: Properties of Domination Factor Function}
Given a channel $V \in \R^{q \times r}_{\textsf{sto}}$ that is strictly positive entry-wise, its domination factor function $\mu_V : \big(0,\frac{q-1}{q}\big) \rightarrow \R^{+}$ is continuous, convex, and strictly increasing. Moreover, we have $\lim_{\delta \rightarrow \frac{q-1}{q}}{\mu_V\left(\delta\right)} = +\infty$.
\end{proposition}

\begin{proof}
We first prove that $\mu_V$ is finite on $\big(0,\frac{q-1}{q}\big)$. For any $P_X,Q_X \in \P_q$ and any $\delta \in \big(0,\frac{q-1}{q}\big)$, we have:
\begin{align*}
D(P_X V || Q_X V) & \leq \chi^2(P_X V || Q_X V) \leq \frac{\left\|(P_X - Q_X) V\right\|_{\ell^2}^2}{\nu} \\
& \leq \frac{\left\|P_X - Q_X\right\|_{\ell^2}^2 \left\|V\right\|_{\textsf{op}}^2}{\nu} 
\end{align*}
where the first inequality is well-known (see e.g. \cite[Lemma 8]{BoundsbetweenContractionCoefficients}) and $\nu = \min\left\{[V]_{i,j}:1 \leq i \leq q, 1 \leq j \leq r\right\}$, and:
\begin{align*}
D(P_X W_{\delta} || Q_X W_{\delta}) & \geq \frac{1}{2} \left\|(P_X - Q_X) W_{\delta}\right\|_{\ell^2}^2 \\
& \geq \frac{1}{2} \left\|P_X - Q_X\right\|_{\ell^2}^2 \left(1 - \delta - \frac{\delta}{q-1}\right)^2
\end{align*}
where the first inequality follows from Pinsker's inequality (see e.g. \cite[Proof of Lemma 6]{BoundsbetweenContractionCoefficients}), and the second inequality follows from part 2 of Proposition \ref{Prop: Properties of Symmetric Channel Matrices}. Hence, we get:
\begin{equation}
\forall \delta \in \left(0,\frac{q-1}{q}\right), \enspace \mu_V \left(\delta\right) \leq \frac{2 \left\|V\right\|_{\textsf{op}}^2}{\nu \left(1 - \delta - \frac{\delta}{q-1}\right)^2} . 
\end{equation}

To prove that $\mu_V$ is strictly increasing, observe that $W_{\delta^{\prime}} \succeq_{\textsf{\tiny deg}} W_{\delta}$ for $0 < \delta^{\prime} < \delta < \frac{q-1}{q}$, because $W_{\delta} = W_{\delta^{\prime}} W_{p}$ with:
\begin{align*}
p & = \delta - \frac{\delta^{\prime}}{1-\delta^{\prime}-\frac{\delta^{\prime}}{q-1}} + \frac{\delta\delta^{\prime}}{1-\delta^{\prime}-\frac{\delta^{\prime}}{q-1}} + \frac{\frac{\delta \delta^{\prime}}{q-1}}{1-\delta^{\prime}-\frac{\delta^{\prime}}{q-1}} \\
& = \frac{\delta - \delta^{\prime}}{1-\delta^{\prime}-\frac{\delta^{\prime}}{q-1}} \in \left(0,\frac{q-1}{q}\right)
\end{align*}
where we use part 4 of Proposition \ref{Prop: Properties of Symmetric Channel Matrices}, the proof of part 5 of Proposition \ref{Prop: Properties of Symmetric Channel Matrices} in Appendix \ref{App: Proofs of Propositions}, and the fact that $W_{p} = W_{\delta^{\prime}}^{-1} W_{\delta}$. As a result, we have for every $P_X,Q_X \in \P_q$:
$$ D\left(P_X W_{\delta}||Q_X W_{\delta}\right) \leq \eta_{\textsf{\tiny KL}}\!\left(W_{p}\right) D\left(P_X W_{\delta^{\prime}}||Q_X W_{\delta^{\prime}}\right) $$
using the SDPI for KL divergence, where part 4 of Proposition \ref{Prop: Constants of Symmetric Channels} reveals that $\eta_{\textsf{\tiny KL}}\!\left(W_{p}\right) \in (0,1)$ since $p \in \big(0,\frac{q-1}{q}\big)$. Hence, we have for $0 < \delta^{\prime} < \delta < \frac{q-1}{q}$:
\begin{equation}
\mu_V\left(\delta^{\prime}\right) \leq \eta_{\textsf{\tiny KL}}\!\left(W_{p}\right) \mu_V\left(\delta\right) 
\end{equation}
using \eqref{Eq: Domination Factor}, and the fact that $0 < D\left(P_X W_{\delta^{\prime}}||Q_X W_{\delta^{\prime}}\right) < +\infty$ if and only if $0 < D\left(P_X W_{\delta}||Q_X W_{\delta}\right) < +\infty$. This implies that $\mu_V$ is strictly increasing. 

We next establish that $\mu_V$ is convex and continuous. For any fixed $P_X,Q_X \in \P_q$ such that $P_X \neq Q_X$, consider the function $\delta \mapsto D\left(P_X V||Q_X V\right)/D\left(P_X W_{\delta}||Q_X W_{\delta}\right)$ with domain $\big(0,\frac{q-1}{q}\big)$. This function is convex, because $\delta \mapsto D\left(P_X W_{\delta}||Q_X W_{\delta}\right)$ is convex by the convexity of KL divergence, and the reciprocal of a non-negative convex function is convex. Therefore, $\mu_V$ is convex since \eqref{Eq: Domination Factor} defines it as a pointwise supremum of a collection of convex functions. Furthermore, we note that $\mu_V$ is also continuous since a convex function is continuous on the interior of its domain. 

Finally, observe that: 
\begin{align*}
\liminf_{\delta \rightarrow \frac{q-1}{q}}{\mu_V\left(\delta\right)} & \geq \sup_{\substack{P_X, Q_X \in \P_q\\P_X \neq Q_X}}{\liminf_{\delta \rightarrow \frac{q-1}{q}}{\frac{D\left(P_X V||Q_X V\right)}{D\left(P_X W_{\delta}||Q_X W_{\delta}\right)}}} \\
& = \sup_{\substack{P_X, Q_X \in \P_q\\P_X \neq Q_X}}{\frac{D\left(P_X V||Q_X V\right)}{\displaystyle{\limsup_{\delta \rightarrow \frac{q-1}{q}}{D\left(P_X W_{\delta}||Q_X W_{\delta}\right)}}}} \\
& = +\infty
\end{align*}
where the first inequality follows from the minimax inequality and \eqref{Eq: Domination Factor} (note that $0 < D\left(P_X W_{\delta}||Q_X W_{\delta}\right) < +\infty$ for $P_X \neq Q_X$ and $\delta$ close to $\frac{q-1}{q}$), and the final equality holds because $P_X W_{(q-1)/q} = \unif$ for every $P_X \in \P_q$.
\end{proof} 

\begin{lemma}[Gramian L\"{o}wner Domination implies Symmetric Part L\"{o}wner Domination]
\label{Lemma: Gramian Loewner Domination implies Symmetric Part Loewner Domination}
Given $A \in \R^{q \times q}_{\succeq 0}$ and $B \in \R^{q \times q}$ that is normal, we have:
$$ A^2 = A A^T \succeq_{\textsf{\tiny PSD}} B B^T \enspace \Rightarrow \enspace A = \frac{A + A^T}{2} \succeq_{\textsf{\tiny PSD}} \frac{B + B^T}{2} . $$
\end{lemma}

\begin{proof}
Since $A A^T \succeq_{\textsf{\tiny PSD}} B B^T \succeq_{\textsf{\tiny PSD}} 0$, using the L\"{o}wner-Heinz theorem (presented in \eqref{Eq: Loewner-Heinz}) with $p = \frac{1}{2}$, we get:
$$ A = \left(A A^T\right)^{\frac{1}{2}} \succeq_{\textsf{\tiny PSD}} \left(B B^T\right)^{\frac{1}{2}} \succeq_{\textsf{\tiny PSD}} 0 $$
where the first equality holds because $A \in \R^{q \times q}_{\succeq 0}$. It suffices to now prove that $(B B^T)^{1/2} \succeq_{\textsf{\tiny PSD}} \left(B + B^T\right)\!/2$, as the transitive property of $\succeq_{\textsf{\tiny PSD}}$ will produce $A \succeq_{\textsf{\tiny PSD}} \left(B + B^T\right)\!/2$. Since $B$ is normal, $B = UDU^H$ by the complex spectral theorem \cite[Theorem 7.9]{LinearAlgebraAxler}, where $U$ is a unitary matrix and $D$ is a complex diagonal matrix. Using this unitary diagonalization, we have:
$$ U|D|U^H = \left(BB^T\right)^{\frac{1}{2}} \succeq_{\textsf{\tiny PSD}} \frac{B + B^T}{2} = U\Re\!\left\{D\right\}U^H $$
since $|D| \succeq_{\textsf{\tiny PSD}} \Re\!\left\{D\right\}$, where $|D|$ and $\Re\!\left\{D\right\}$ denote the element-wise magnitude and real part of $D$, respectively. This completes the proof. 
\end{proof}

\section*{Acknowledgment}

We would like to thank an anonymous reviewer and the Associate Editor, Chandra Nair, for bringing the reference \cite{vanDijk} to our attention. Yury Polyanskiy would like to thank Dr. Ziv Goldfeld for discussions on secrecy capacity.

\bibliographystyle{IEEEtran}
\bibliography{IEEEabrv,ChannelDominationRefs}

\end{document}